\documentclass{article}
\usepackage{microtype}
\usepackage{graphicx}
\usepackage{booktabs} 
\usepackage{amsfonts}
\usepackage{amsmath}
\usepackage{amsthm}
\usepackage{hyperref}

\usepackage[accepted]{icml2024}
\usepackage[capitalize,noabbrev]{cleveref}
\newtheorem{thm}{Theorem}

\newtheorem{lemma}{Lemma}
\newtheorem{corollary}{Corollary}

\usepackage{algorithmic}
\usepackage{array}
\usepackage{subfigure}
\usepackage{appendix}

\usepackage{booktabs}
\usepackage{multirow}   
\usepackage{colortbl}
\usepackage{xcolor}
\usepackage{caption}
\definecolor{celestialblue}{rgb}{0.29, 0.59, 0.82}
\definecolor{cerulean}{rgb}{0.0, 0.48, 0.65}
\definecolor{cadmiumorange}{rgb}{0.93, 0.53, 0.18}
\DeclareMathOperator*{\argmax}{arg\,max}

\icmltitlerunning{Collective Certified Robustness against Graph Injection Attacks}
\date{}

\begin{document}
\twocolumn[\icmltitle{Collective Certified Robustness against Graph Injection Attacks}
\icmlsetsymbol{equal}{*}

\begin{icmlauthorlist}
\icmlauthor{Yuni Lai}{equal,polyu1}
\icmlauthor{Bailin Pan}{equal,polyu2}
\icmlauthor{Kaihuang Chen}{polyu2}
\icmlauthor{Yancheng Yuan}{polyu2}
\icmlauthor{Kai Zhou}{polyu1}
\end{icmlauthorlist}

\icmlaffiliation{polyu1}{Department of Computing, The Hong Kong Polytechnic University, Hong Kong, China}
\icmlaffiliation{polyu2}{Department of Applied Mathematics, The Hong Kong Polytechnic University, Hong Kong, China}

\icmlcorrespondingauthor{Kai Zhou}{kaizhou@polyu.edu.hk}

\icmlkeywords{Machine Learning, ICML}

\vskip 0.3in
]

\printAffiliationsAndNotice{\icmlEqualContribution} 

\begin{abstract}

We investigate certified robustness for GNNs under graph injection attacks. Existing research only provides \textit{sample-wise} certificates by verifying each node independently, leading to very limited certifying performance. In this paper, we present the first \textit{collective} certificate, which certifies a set of target nodes simultaneously. To achieve it, we formulate the problem as a binary integer quadratic constrained linear programming (BQCLP). We further develop a customized linearization technique that allows us to relax the BQCLP into linear programming (LP) that can be efficiently solved. Through comprehensive experiments, we demonstrate that our collective certification scheme significantly improves certification performance with minimal computational overhead. For instance, by solving the LP within $1$ minute on the Citeseer dataset, we achieve a significant increase in the certified ratio from $0.0\%$ to $81.2\%$ when the injected node number is $5\%$ of the graph size. Our step marks a crucial step towards making provable defense more practical.


\end{abstract}

\section{Introduction}

Graph Neural Networks (GNNs) have emerged as the dominant models for graph learning tasks, demonstrating remarkable success across diverse applications. However, recent studies~\cite{zugner2018netattack,zugner2018metattack,liu2022towards} have revealed the vulnerability of GNNs to adversarial attacks, raising significant concerns regarding their security.
Notably, a new type of attack called Graph Injection Attack (GIA) has raised considerable attention. Unlike the commonly studied Graph Modification Attack (GMA), which involves inserting and deleting edges, GIA will inject carefully crafted malicious nodes into the graph. Recent research~\cite{chen2022understanding,tao2023adversarial,ju2023let} has demonstrated that GIA is not only more cost-efficient but also more powerful than GMA.


To counteract these attacks, significant efforts have been dedicated to robustifying GNNs. Representative defense approaches include adversarial training~\cite{gosch2023adversarial}, the development of robust GNN architectures~\cite{jin2020graph,zhu2019robust,zhang2020gnnguard}, and the detection of adversaries~\cite{zhang2019comparing,zhang2020gcn}. While these approaches are quite effective against \textit{known} attacks, there remains a concern that new adaptive attacks could undermine their robustness.
To tackle the challenge of emerging novel attacks,  researchers have explored \textit{provable defense} approaches~\cite{cohen2019certified,li2023sok,bojchevski2020efficient,scholten2022randomized,schuchardt2023localized} that offer \textit{certified robustness} for GNN models: the predictions of models are theoretically guaranteed to be consistent if the attacker's budget (e.g., the number of edges could be modified) is constrained in a certain range.

\begin{figure}
    \centering
    \includegraphics[width=1\linewidth]{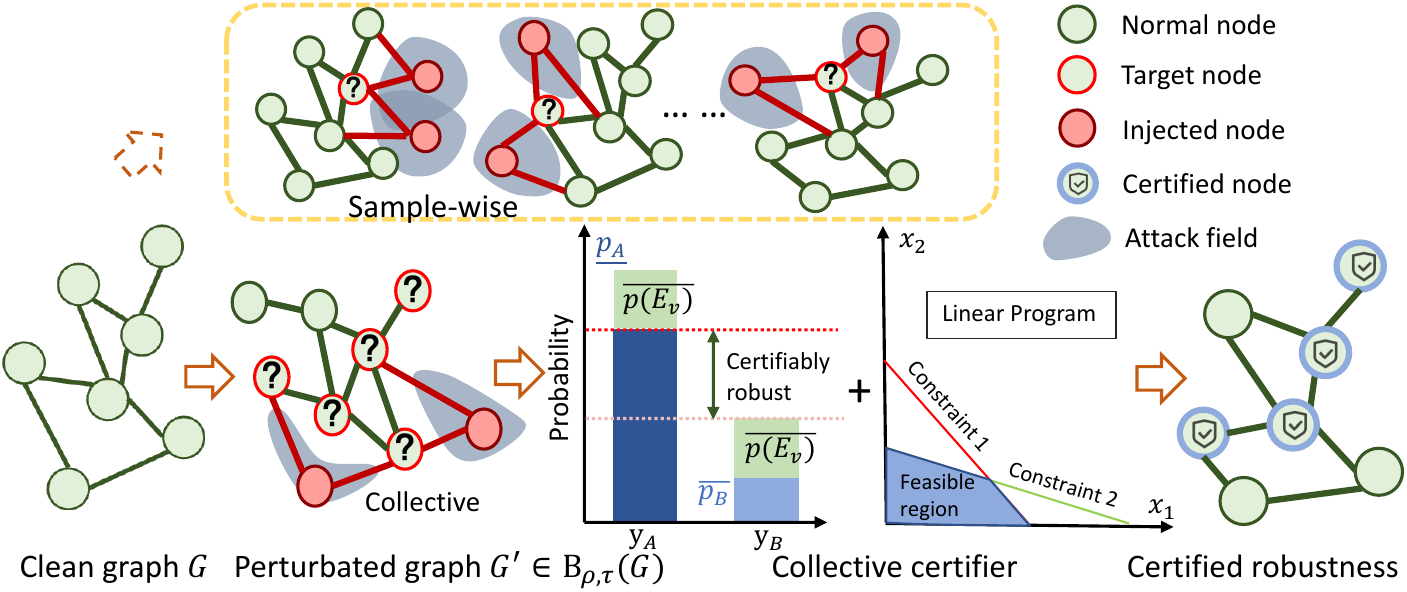}
    \caption{Illustration of collective certification.}
    \label{fig:overview}
    \vspace{-10pt}
\end{figure}

\textbf{Sample-wise vs. Collective certification}\quad The certification against attacks over graphs can be categorized into two types:  \textbf{sample-wise} and \textbf{collective}. Sample-wise certification approaches~\cite{cohen2019certified,bojchevski2020efficient,lai2023node}
essentially verify the prediction for a node \textit{one by one}, assuming that the attacker can craft \textit{a different perturbed graph each time} to attack a single node (Figure.~\ref{fig:overview}, top).
However, in reality, the attacker can only produce \textit{a single perturbed graph} to simultaneously disrupt all predictions for a set of target nodes $\mathbb{T}$  (Figure.~\ref{fig:overview}, bottom). Such a discrepancy makes sample-wise certificates rather pessimistic. 
In contrast, more recent works ~\cite{schuchardt2020collective,schuchardt2023localized} aim to certify the set of nodes at once, providing \textbf{collective certification} that can significantly improve the certifying performance.

In the domain of certifying GNNs, the majority of research works \cite{bojchevski2020efficient,wang2021certified,jia2020certified,jia2022almost,scholten2022randomized} focus on \textit{sample-wise} certification against \textit{GMA}.
The only \textit{collective} certification scheme against GMA proposed by \citet{schuchardt2020collective}, however, is not applicable to GIA. This is because the certification scheme assumes a fixed receptive field of GNNs, while GIA, which involves adding edges after injecting nodes, inevitably expands the receptive field.
Although there are emerging works~\cite{lai2023node,jia2023pore} specifically designed to tackle GIA, they only offer sample-wise certificates, resulting in limited certified performance.

We are therefore motivated to derive the \textbf{first collective certified robustness scheme} for GNNs against GIA.
To achieve collective robustness, we leverage the inherent locality property of GNNs, where the prediction of a node in a $k$-layer message-passing GNN is influenced solely by its $k$-hop neighbors. 
This ensures that injected edges by the attacker only impact a subset of the nodes. We address the collective certification problem by transforming it into a budget allocation problem, considering the attacker's objective of modifying the predictions of as many nodes as possible with a limited number of malicious nodes and maximum edges per node. By overestimating the number of modified nodes, we can certify the consistent classification of the remaining nodes.

However, the above problem yields a binary integer polynomial-constrained program, which is known to be NP-hard. We then propose a \textit{customized} linearization technique to relax the original problem to a Linear Programming (LP), which can be solved efficiently. The LP relaxation provides a lower bound on the achievable certified ratio, ensuring the soundness of the verification process.
We conduct comprehensive experiments to demonstrate the effectiveness as well as the computational efficiency of our collective certification scheme. For example, when the injected node number is $5\%$ of the graph size, our collective robustness models increase the certified ratio from $0.0\%$ to over $80.0\%$ in both Cora-ML and Citeseer datasets, and it only takes about $1$ minutes to solve the collective certifying problem. 

Overall, we propose the first collective certificate for GNNs against graph injection attacks. In particular, it is computationally efficient and can significantly improve the certified ratio. Moreover, this certification scheme is \textit{almost} model-agnostic as it is applicable for any message-passing GNNs.

\section{Background}
\subsection{Graph Node Classification}
We focus our study on graph node classification tasks. Let $G=(\mathcal{V},\mathcal{E},X)\in \mathbb{G}$ represent an undirected graph, where $\mathcal{V}=\{v_1, \cdots, v_n\}$ is the set of $n$ nodes, $\mathcal{E}=\{e_{ij}=(v_i,v_j)\}$ denotes the set of edges with each edge $e_{ij}$ connecting $v_i$ and $v_j$, and $X\in \mathbb{R}^{n\times d}$ represents the features associated with nodes. Equivalently, we can use an adjacency matrix $A\in \{0,1\}^{n\times n}$ with $A_{ij} = 1$ if $e_{ij} \in \mathcal{E}$ and $A_{ij} = 0$ if $e_{ij} \notin \mathcal{E}$ to encode the graph structure of $G$. Each node has its label $y\in\mathcal{Y}=\{1,\cdots, K\}$, but only a subset of these labels are known. The goal of a multi-output graph node classifier $f:\mathbb{G}\rightarrow \{1,\cdots, K\}^n$ is to predict the missing labels given the input graph $G$.

\subsection{Message-Passing Graph Neural Networks}
In this paper, we study certified robustness approaches that are applicable to the most commonly used GNNs that operate under the message-passing framework based on neighbor aggregation. These message-passing GNNs~\cite{kipf2016semi,gilmer2017neural,velivckovic2018graph,geisler2020reliable} encode the local information of each node by aggregating its neighboring node features (i.e., embedding) through various aggregation functions. 
During the inference, the receptive field of a node $v$ in $k$-layer GNN is just its $k$-hops neighbors, and the nodes/edges beyond the receptive field would not affect the prediction of the node when the model is given. This locality enables the application of collective certificates.


\subsection{Certified Robustness from Randomized Smoothing}
Certified robustness aims to provide a theoretical guarantee of the consistency of a model's prediction under a certain perturbation range on the input. Randomized smoothing is a widely adopted and versatile approach for achieving such certification across a range of models and tasks~\cite{jia2020certified,bojchevski2020efficient,li2023sok}. Take the graph model as an example; it adds random noise (such as randomly deleting edges) to the input graph. Then, given any classifier $f$, it builds a smoothed classifier $g$ which returns the ``majority vote" regarding the random inputs. Certification is achieved based on the fact that there is a probability of overlap between the random samples drawn from the clean graph and the perturbed graph, in which the predictions must be the same. 

\section{Problem Statements}
\subsection{Threat Model: Graph Injection Attack}
We focus on providing robustness certificates against graph injection attacks (GIAs) under the \textit{evasion} threat model, where the attack perturbation occurs after the model training. 
The adversaries aim to disrupt the node classifications of a set of target nodes, denoted by $\mathbb{T}$, as many as possible. To this end, it can inject $\rho$ additional nodes $\tilde{\mathcal{V}}=\{\tilde{v}_1,\cdots,\tilde{v}_{\rho}\}$ into the graph. These injected nodes possess \textit{arbitrary node features} represented by the matrix $\tilde{X}\in \mathbb{R}^{\rho \times d}$. Additionally, $\tilde{\mathcal{E}}$ represents the set of edges introduced by the injected nodes. To limit the power of the adversaries and avoid being detected by the defender, we assume that each injected node $\tilde{v}$ is only capable of injecting a maximum of $\tau$ edges. Thus, the degree of each injected node  $\delta(\tilde{v})$ is no more than $\tau$. Let us represent the perturbed graph as $G'$, with its corresponding adjacency matrix denoted as $A'$. We formally define the potential GIA as a perturbations set associated with a given graph $G=(\mathcal{V},\mathcal{E}, X)$:
\begin{align}
\label{eqn:pertb_ball}
    B_{\rho,\tau}(G)&:=\{G'(\mathcal{V}',\mathcal{E}',X')|\mathcal{V}'=\mathcal{V}\cup\tilde{\mathcal{V}}, \mathcal{E}'=\mathcal{E}\cup \tilde{\mathcal{E}}, \nonumber\\ 
    &X'=X \cup \tilde{X},|\tilde{\mathcal{V}}| \leq \rho,\delta(\tilde{v})\leq\tau,\forall \tilde{v} \in \tilde{\mathcal{V}}\}.
\end{align}


Given the absence of a collective certificate to address these types of perturbations, our first contribution is to define the problem of collective robustness.

\subsection{Problem of Collective Certified Robustness}
Following \cite{scholten2022randomized}, we employ randomized smoothing to serve as the foundation of our certification. Intuitively, by adding random noise to the graph, the message from the injected node to a target node has some probability of being intercepted in the randomization, such that the GNN models will not aggregate the inserted node's feature for prediction.
We adopt node-aware bi-smoothing~\cite{lai2023node}, which was proposed to certify against the GIA perturbation, as our smoothed classifier. Given a graph $G$, random graphs are created by a randomization scheme denoted as  $\phi(G)=(\phi_e(G),\phi_n(G))$. It consists of two components: edge deletion smoothing $\phi_e(G)$ and node deletion smoothing $\phi_n(G)$. Specifically, the former randomly deletes each edge with probability $p_e$, and the latter randomly deletes each node (together with its incident edges) with probability $p_n$. Based on these random graphs, a smoothed classifier $g(\cdot)$ is constructed as follows:
\begin{align}
\label{eqn:smooth_g}
    &g_v(G):=\argmax_{y\in \{1,\cdots,K\}}p_{v,y}(G), 
\end{align}
where $p_{v,y}(G):=\mathbb{P}(f_v(\phi(G))=y)$ represents the probability that the base GNN classifier $f$ returned the class $y$ for node $v$ under the smoothing distribution $\phi(G)$, and $g(\cdot)$ returns the ``majority votes'' of the base classifier $f(\cdot)$.



Given a specific attack budget $\rho$ and $\tau$, our objective is to provide certification for the number of target nodes in $\mathbb{T}$ that are guaranteed to maintain consistent robustness against any potential attack. We assume that the attacker's objective is to maximize the disruption of predictions for the target nodes, $\sum_{v\in \mathbb{T}} \mathbb{I}\{g_v(G')\neq g_v(G)\}$, through the allocation of inserting edges. By modeling a worst-case attacker that leads to a maximum number of non-robust nodes, we can certify that the remaining number of nodes is robust.
Such that the collective certification can be formulated as an optimization problem as follows:
\begin{align}
\label{opt:collective_1}
\min_{G'\in B_{\rho,\tau}(G)}\quad  &|\mathbb{T}|-\sum_{v\in \mathbb{T}} \mathbb{I}\{g_v(G')\neq g_v(G)\}, \\
\text{s.t.} \quad
&|\tilde{\mathcal{V}}|\leq \rho,\, \delta(\tilde{v})\leq \tau, \,\forall \tilde{v}\in \tilde{\mathcal{V}}.\nonumber
\end{align}

Typically, when setting the $\mathbb{T}$ as a single node, the problem degrades to a sample-wise certificate.

\section{Collective Certified Robustness}
In this section, we derive the collective certificate for the smoothed classifier with any message-passing GNNs as the base classifier. To ensure the clarity of the presentation, we begin by providing an overview of our approach.

\subsection{Overview}

The derivation of the robustness certificate relies on a \textit{worst-case} assumption: in the message-passing process, if a node receives even a single message from any injected node, its prediction will be altered. It is important to note that this assumption exaggerates the impact of the attack, thereby validating the guarantee of the defense. Accordingly, we define \textbf{message interference} for a node $v$ as the event  $E_v$ that the node $v$ receives \textit{at least} one message from injected nodes in message passing.

The achievement of collective certification then constitutes the following crucial steps. First, we derive an upper bound on the probability of the message interference event, denoted as $p(E_v)$ (Section.~\ref{subsec:Estimate_P_Ev}). Second, we establish the relation between the probability $p(E_v)$ and the prediction probability $p_{v,y}(G)$, which allows us to bound the change of $p_{v,y}(G)$ under the perturbation range $B_{\rho,\tau}(G)$ (Section.~\ref{subsec:bound_change}). Third, we derive the certifying condition for smoothed classifier $g$ based on the results from the previous sections (Section.~\ref{subsec:Certify_Con}). Finally, we formulate the collective certified robustness problem as an optimization problem (Section.~\ref{subsec:Collective}).

\subsection{Condition for Certified Robustness}

\subsubsection{Message interference event}
\label{subsec:Estimate_P_Ev}


We begin by introducing some necessary notations. We use $P_{\tilde{v}v}^k$ to represent all the existing paths from an injected node $\tilde{v}\in\tilde{\mathcal{V}}$ to a testing node $v$, where the length or distance of these paths is smaller than $k$. Each path $q$ in $P_{\tilde{v}v}^k$ consists of a series of linked edges. To simplify notation, we define $\phi_e(A)$ as an equivalent representation of $\phi_e(G)$, where $\phi_e(A)_{ij}=0$ if the edge $(i,j)$ does not exist after the sampling, and $\phi_e(A)_{ij}=1$ if the edge $(i,j)$ remains. Similarly, we represent $\phi_n(G)$ as $\phi_n(A)_i$, where $\phi_n(A)_{i}=0$ indicates the deletion of node $i$, and $\phi_n(A)_{i}=1$ denotes that the node remains unchanged.
%
Then, we formally define the event $E_v$ as:
\begin{align}
\label{eqn:E_v}
    \exists \tilde{v}\in \tilde{\mathcal{V}}:\ &(\exists q \in P_{\tilde{v}v}^k:  (\forall n_i\in q: \phi_n(A')_{n_i}=1) \\ 
    &\land (\forall (i,j)\in q: \phi_e(A')_{ij}=1)). \nonumber
\end{align}
That is at least one path from a malicious node $\tilde{v}$ to the testing node $v$ is effective (all edges and nodes are kept in the smoothing).
Below, our goal is to quantify the probability of $E_v$,  so that we can provide an estimation of the potential impact of injected nodes on the prediction probability.

However, directly estimating the event probability $p(E_v)$ is difficult because we need to find out all the possible paths $P_{\tilde{v}v}^k$ for each node. Similar to \cite{scholten2022randomized}, we have an upper bound for $p(E_v)\leq \overline{p(E_v)}$ by assuming the independence among the paths: 

\begin{lemma} 
\label{thm:P_e_general}
Let $A$ be the adjacency matrix of the perturbed graph with $\rho$ injected nodes, and the injected nodes are in the last $\rho$ rows and columns. With smoothing $p_n>0$ and $p_e>0$, we have the upper bound of $p(E_v)$: 
    \begin{align}
    \label{Eqn:P_e_General}
        &p(E_v)\leq \overline{p(E_v)}\\
        =&1- p_1^{||A_{n:(n+\rho),v}||_1}p_2^{||A_{n:(n+\rho),v}^2||_1}\cdots p_k^{||A_{n:(n+\rho),v}^k||_1},\nonumber
    \end{align}
where $p_i:=1-(\bar{p}_e\bar{p}_n)^{i},\, \forall i\in\{1,2,\cdots,k\}$, and adjacency matrix $A$ 
contains the injected nodes encoded in the $(n+1)^{th}$ to $(n+\rho)^{th}$ row, and $||\cdot||_1$ is $l_1$ norm.
\end{lemma}
\vspace{-8pt}
\begin{proof}
(Sketch) Let $p(\bar{E}_v^{\tilde{v}})$ denote the probability that all paths are intercepted from an injected node $\tilde{v}$ to node $v$ in the case that of considering each path independently. We have $p(\bar{E}_v^{\tilde{v}})=\prod_{q\in P_{\tilde{v}v}^k}(1-(\bar{p}_e\bar{p}_n)^{|q|})$, where $\bar{p}_e:=1-p_e$, $\bar{p}_n:=1-p_n$ and $|q|\in \{1,\cdots,k\}$ represent the length of the path $q\in P_{\tilde{v}v}^k$ from $\tilde{v}$ to $v$. 
Furthermore, $||A_{n:(n+\rho),v}^k||_1$ quantifies the number of paths with a length of $k$ originating from any malicious node and reaching node $v$. Finally, by considering multiple injected nodes, we have $\overline{p(E_v)} = 1-\prod_{\tilde{v}\in\tilde{\mathcal{V}}}p(\bar{E}_v^{\tilde{v}})$. See Appendix.~\ref{Sec:Appendix_A} for complete proof.
\end{proof}

\subsubsection{Bounding the change of prediction}
\label{subsec:bound_change}

Next, we first provide Lemma. \ref{thm:Ev_general_case} to demonstrate that the occurrence of the complement event of $E_v$, denoted as $\bar{E}_v$, is the condition for the consistent prediction of base classifier $f$. Then, we prove that the change of prediction probability for the smoothed classifier $g$ is bounded by $p(E_v)$:

\begin{lemma}
\label{thm:Ev_general_case}
    Given a testing node $v\in G$, perturbation range $B_{\rho,\tau}(G)$, $p_n>0$ and $p_e>0$, we have $f_v(\phi(G))=f_v(\phi(G')),\, \forall G' \in B_{\rho,\tau}(G)$ if event $\bar{E}_v$ occurs:
    \begin{align}
    \label{eqn:bar_E_v}
        \forall \tilde{v}\in \tilde{\mathcal{V}}:\ &(\forall q \in P_{\tilde{v}v}^k:  (\exists n_i\in q: \phi_n(A')_{n_i}=0) \\ 
        &\lor (\exists (i,j)\in q: \phi_e(A')_{ij}=0)). \nonumber
    \end{align}
\end{lemma}
\vspace{-5pt}
\begin{proof} 
    For each path $q\in P_{\tilde{v}v}^k$, the message from the injected node $\tilde{v}$ to the target node $v$ is intercepted if at least one of the edges or nodes along the path is deleted. Consequently, if all the paths are intercepted as a result of the smoothing randomization $\phi(G')$, the prediction for the target node $v$ remains unchanged. 
\end{proof}
\vspace{-5pt}
Now, we can establish a bound on the change in prediction probability of the smoothed classifier $g$, which serves as a crucial step for deriving the certifying condition.

\begin{thm}
\label{thm:prob_change}
Given a base GNN classifier $f$ trained on a graph $G$ and its smoothed classifier $g$ defined in \eqref{eqn:smooth_g}, a testing node $v \in G$ and a perturbation range $B_{\rho,\tau}(G)$, let $E_v$ be the event defined in Eq.~\eqref{eqn:E_v}. The absolute change in predicted probability $|p_{v,y}(G)-p_{v,y}(G')|$ for all perturbed graphs $G' \in B_{\rho,\tau}(G)$ is bounded by the probability of the event $E_v$: $|p_{v,y}(G)-p_{v,y}(G')|\leq p(E_v)$.
\end{thm}
\vspace{-8pt}
\begin{proof}
    (Sketch) $p_{v,y}(G)-p_{v,y}(G')\leq \mathbb{P}(f_v(\phi(G))=y\land E_v)=p(E_v)\cdot \mathbb{P}(f_v(\phi(G))=y|E_v)\leq p(E_v)$. See Appendix.~\ref{Sec:Appendix_A} for complete proof.
\end{proof}
\vspace{-5pt}

\subsubsection{Certifying Condition}
\label{subsec:Certify_Con}
With the upper bound of the probability change $p_{v,y}(G)$ provided in Theorem.~\ref{thm:prob_change} and upper bound of $p(E_v)$ provided in Lemma.~\ref{thm:P_e_general}, we can derive the certifying condition for smoothed classifier $g$ under a given perturbation range:

\begin{corollary}
\label{thm:certify_condition}
    Given a base GNN classifier $f$ trained on a graph $G$ and its smoothed classifier $g$, a testing node $v \in G$ and a perturbation range $B_{\rho,\tau}(G)$, let $E_v$ be the event defined in Eq.~\eqref{eqn:E_v}. We have $g_v(G')=g_v(G)$ for all perturbed graphs $G' \in B_{\rho,\tau}(G)$ if:
    \begin{equation}
    \label{eqn:certify_condition}
        \overline{p(E_v)}< [p_{v,y^*}(G)-max_{y\neq y^*}p_{v,y}(G)]/2,
    \end{equation}
    where $y^*\in \mathcal{Y}$ is the predicted class of $g_v(G)$.
\end{corollary}
\vspace{-8pt}
\begin{proof}
    With Theorem.~\ref{thm:prob_change}, we have $g_v(G')=g_v(G)$ if $p_{v,y^*}(G)-\overline{p(E_v)}> max_{y\neq y^*}p_{v,y}(G)+\overline{p(E_v)}$, which is equivalent to $\overline{p(E_v)}< [p_{v,y^*}(G)-max_{y\neq y^*}p_{v,y}(G)]/2$.
\end{proof}
\vspace{-5pt}
Nevertheless, quantifying $\overline{p(E_v)}$ is still challenging due to the unknown paths $P_{\tilde{v}v}^k$ or the perturbed adjacency matrix. To tackle the challenge, we introduce the following collective certifying framework that models the problem of certifying node injection perturbation as an optimization problem. More importantly, 
we can certify a set of nodes at the same time to enhance the certifying performance. 
\subsection{Collective Certification as Optimization}
\label{subsec:Collective}


With Corollary.~\ref{thm:certify_condition}, we know that node $v$ is not certifiably robust if $\overline{p(E_v)}\geq [p_{v,y^*}(G)-max_{y\neq y^*}p_{v,y}(G)]/2$. Under a limited attack budget, the worst-case attacker can lead to a maximum number of non-robust nodes among target nodes in $\mathbb{T}$, which can be formulated as follows:
\begin{align}
\label{opt:collective}
\max_{G'\in B_{\rho,\tau}(G)}\quad  &M=\sum_{v\in \mathbb{T}} \mathbb{I}\{\overline{p(E_v)}\geq c_v/2\}, \\
\text{s.t.} \quad
&|\tilde{\mathcal{V}}|\leq \rho,\, \delta(\tilde{v})\leq \tau, \,\forall \tilde{v}\in \tilde{\mathcal{V}},\nonumber
\end{align}
where $c_v:=p_{v,y^*}(G)-max_{y\neq y^*}p_{v,y}(G)$, is the classification gap of smoothed classifier.
To obtain the certifiably robust node number among all testing nodes, the optimal objective value $M^*$ of \eqref{opt:collective} can serve as an upper bound for non-robust nodes, and hence the remaining $|\mathbb{T}|-M^*$ nodes are certified robust. Plugging in $\overline{p(E_v)}$ with \eqref{Eqn:P_e_General}, and taking the logarithm of the $\overline{p(E_v)}$, we transformed the problem~\eqref{opt:collective} to a binary integer \textit{polynomial-constrained programming} (We put the problem and formulation details in Appendix.~\ref{Sec:Appendix_C}). 

Typically, for two-layer GNNs ($k=2$), we formulate the problem into a binary integer quadratic constrained linear programming problem (BQCLP). Let $A_0$ be the original adjacency matrix of the existing $n$ nodes in the graph $G$, and $A_1$ denote the adjacency matrix from injected $\rho$ malicious nodes to the existing nodes, and $A_2$ be the adjacency matrix representing the internal connection between the malicious nodes. 
Then the problem \eqref{opt:collective2} becomes the BQCLP problem as follows (See Appendix.~\ref{Sec:Appendix_C} for detailed formulation):
\begin{align}
\label{opt:collective-BQP}
\max_{A_1,A_2,\mathbf{m}}&\quad  M=\mathbf{t}^\top \mathbf{m},\\ 
\text{s.t.} \quad 
& \tilde{p_1}A_1^\top\mathbf{1}_{\rho}  + \tilde{p_2}(A_1A_0+A_2A_1)^\top\mathbf{1}_{\rho}\leq \mathbf{C}\circ \mathbf{m},\nonumber\\
&A_1\mathbf{1}_n +A_2 \mathbf{1}_{\rho} \leq \tau,\,A_2^{\top}=A_2, \nonumber\\
& A_1\in\{0,1\}^{\rho\times n},\,A_2\in\{0,1\}^{\rho\times \rho},\,\mathbf{m}\in\{0,1\}^{n},\nonumber
\end{align}
where $\mathbf{t}$ is a constant zero-one vector that encodes the position of the target node set $\mathbb{T}$, $\mathbf{m}$ is a vector that indicates whether the nodes are non-robust, $\tilde{p_1}=log(p_1)$ and $\tilde{p_2}=log(p_2)$ are two negative constants, $\mathbf{C}\in \mathbb{R}^n$ is a vector with negative constant elements $log(1-\frac{c_v}{2})$, $\mathbf{1}_n$ denotes all-ones vector with length $n$, $\top$ represents matrix transposition, and $\circ$ denotes element-wise multiplication.

\section{Effective Optimization Methods}
\label{Sec:Linearization}
The BQCLP problem~\eqref{opt:collective-BQP} is non-convex and known to be NP-hard. In this section, we introduce \textit{two} effective methods to relax problem~\eqref{opt:collective-BQP} to a Linear Programming (LP) to solve it efficiently. The first method (termed \textbf{Collective-LP1}) relies on standard techniques to avoid quadratic terms; the second method (termed \textbf{Collective-LP2}) employs a novel customized reformulation that can significantly improve the solution quality and computational efficiency.

\subsection{Standard Linear Relaxation (Collective-LP1)}
To solve problem~\eqref{opt:collective-BQP} efficiently, one common solution is to replace the quadratic terms in the constraint with linear terms by introducing extra slack variables.
We adopt the standard technique~\cite{wei2020tutorials} to address the quadratic terms in $A_2A_1$.
Specifically, let $A_{2(ij)}$ denotes the element of $i^{th}$ row and $j^{th}$ column in matrix $A_2$ and $A_{1(jv)}$ denotes the element in matrix $A_1$. For each quadratic term $A_{2(ij)}A_{1(jv)}$ ($\forall i\in \{1,\cdots,\rho\},\forall j \in \{1,\cdots,\rho\},\forall v \in \{1,\cdots,n\}$) in $A_2A_1$, we can equivalently reformulate $Q_{v(ij)}:=A_{2(ij)}A_{1(jv)}$ with corresponding constraints: $Q_{v(ij)}\in \mathbb{B}$, $Q_{v(ij)}\leq A_{2(ij)}$, $Q_{v(ij)}\leq A_{1(jv)}$, and $A_{2(ij)}+A_{1(jv)}-Q_{v(ij)}\leq 1$. 
We further relax all the binary constraints to the box constraints $[0,1]$, leading to an LP as follows: 
\begin{align}
\label{opt:collective-linear1}
\max_{A_1,A_2,m,\atop Q_1,Q_2,\cdots,Q_n}&\quad  M=\mathbf{t}^\top \mathbf{m},\\ 
\text{s.t.} \quad 
& \tilde{p_1}A_1^\top\mathbf{1}_{\rho} + \tilde{p_2} A_0^\top A_1^\top\mathbf{1}_{\rho} + \tilde{p_2}O \leq \mathbf{C}\circ \mathbf{m},\nonumber\\
&A_1\mathbf{1}_n +A_2 \mathbf{1}_{\rho} \leq \tau,\,A_2^{\top}=A_2, \nonumber\\
&Q_v=(Q_{v(ij)})_{\rho \times \rho},\, v \in \{1,2,\cdots,n\}, \nonumber\\
&O=[\mathbf{1}_{\rho}^\top Q_1\mathbf{1}_{\rho},\mathbf{1}_{\rho}^\top Q_2\mathbf{1}_{\rho},\cdots, \mathbf{1}_{\rho}^\top Q_n\mathbf{1}_{\rho}]^\top, \nonumber\\
&Q_v \leq \mathbf{1}_{\rho} {[A_{1(:,v)}]}^\top,\, Q_v\leq A_2,\,Q_v\in [0,1]^{\rho\times \rho},\nonumber\\
&\mathbf{1}_{\rho} {[A_{1(:,v)}]}^\top+A_2-Q_v \leq 1,\nonumber\\
& A_1\in[0,1]^{\rho\times n},\,A_2\in[0,1]^{\rho\times \rho},\,\mathbf{m}\in[0,1]^{n}.\nonumber
\end{align}
The more detailed formulation of problem \eqref{opt:collective-linear1} is supplied in Appendix.~\ref{Sec:Appendix_C}. This transformation makes our collective robustness problem solvable in polynomial time. 

\paragraph{Validity of relaxation for certification.}
It is important to note that the relaxed LP problem always has a larger feasible region than the original BQCLP problem. As a result, the optimal $\bar{M}^*$ (i.e., the maximum number of non-robust nodes) of the relaxed problem is always greater than the original problem. That is, the number of robust nodes ($|\mathbb{T}|-\bar{M}^*$) certified by the relaxed problem is always smaller or equal to that obtained from the original problem, such that the relaxation always yields sound verification.

Nevertheless, this technique results in introducing $O(\rho^2 |\mathbb{T}|)$ (extra) variables among the matrix $O$. To improve efficiency, we next design a more efficient reformulation that only requires $O(\rho |\mathbb{T}|)$ extra variables.

\subsection{Customized Linear Relaxation (Collective-LP2)}
To reduce the number of the extra variables, we notice that there is a vector in the quadratic term $A_1^\top A_2^\top \mathbf{1}_{\rho}$, and we can first combine the $A_2^\top \mathbf{1}_{\rho}$ to reduce the dimension.  We define a vector variable $\mathbf{z}:=A_2^\top \mathbf{1}_{\rho}$ to replace the term $A_2^\top \mathbf{1}_{\rho}$ in the problem~\eqref{opt:collective-BQP}. Then we can reformulate it as:
\begin{align}
\label{opt:collective-BQP-z}
\max_{A_1,z,\mathbf{m}}\quad & M=\mathbf{t}^\top \mathbf{m},\\ 
\text{s.t.} \quad 
& \tilde{p_1}A_1^\top\mathbf{1}_{\rho}  + \tilde{p_2}A_0^\top A_1^\top \mathbf{1}_{\rho}+\tilde{p_2}A_1^\top \mathbf{z}\leq \mathbf{C}\circ \mathbf{m},\nonumber\\
&A_1\mathbf{1}_n + \mathbf{z} \leq \tau, \,A_1\in\{0,1\}^{\rho\times n},\nonumber\\
& \mathbf{z}\in \{0,1,\cdots,\min(\rho,\tau)\}^{\rho\times 1},\,\mathbf{m}\in\{0,1\}^{n}.\nonumber
\end{align}
To linearize the problem, we need to deal with the quadratic term $A_1^\top \mathbf{z}$. If a binary variable $x\in \mathbb{B}$, and a continuous variable $z\in [0,u]$, then $w:=xz$ is equivalent to~\cite{wei2020tutorials}: $w\leq ux, w\leq z, ux+z-w\leq u, 0\leq w$. To apply it, we first relax the $\mathbf{z}$ to $[0,min(\tau,\rho)]$. Assuming that $\tau \leq \rho$, for each quadratic term $A^\top_{1(ij)}z_j$ ($\forall i\in \{1,\cdots,n\},\forall j \in \{1,\cdots,\rho\}$) in $A_1^\top \mathbf{z}$, we create a substitution variable $Q_{(ij)}=A^\top_{1(ij)}z_j$ with corresponding constraints: $0 \leq Q_{(ij)}$, $Q_{(ij)}\leq \tau A^\top_{1(ij)}$, $Q_{(ij)}\leq z_j$, and $\tau A^\top_{1(ij)} + z_j -Q_{(ij)} \leq \tau$. 
We further relax all the binary constraints to $[0,1]$ interval constraints. Then the problem~\eqref{opt:collective-BQP} can be relaxed to an LP as follows:
\begin{align}
\label{opt:collective-linear2}
\max_{A_1,m,\atop Q \in \mathbb{R}^{n \times \rho}}\quad & M=\mathbf{t}^\top \mathbf{m},\\ 
\text{s.t.} \quad 
& \tilde{p_1}A_1^\top\mathbf{1}_{\rho} + \tilde{p_2} A_0^\top A_1^\top\mathbf{1}_{\rho} + \tilde{p_2} Q \mathbf{1}_{\rho} \leq \mathbf{C}\circ \mathbf{m},\nonumber\\
&A_1\mathbf{1}_n + \mathbf{z} \leq \tau,\,A_1\in[0,1]^{\rho\times n}, \nonumber\\
&Q \leq \tau A_{1}^\top,\,Q\leq \mathbf{1}_n \mathbf{z}^\top,\nonumber\\
&\tau A^\top_{1}+\mathbf{1}_n \mathbf{z}^\top-Q \leq \tau,\nonumber\\
&Q \in[0,\tau]^{n\times \rho},\,\mathbf{z}\in  [0,\tau]^{\rho\times 1},\,\mathbf{m}\in[0,1]^{n}.\nonumber
\end{align}
We put the detailed formulation in Appendix.~\ref{Sec:Appendix_C}. Next, we analyze the complexity of problem~\eqref{opt:collective-linear1} and \eqref{opt:collective-linear2}.
\subsection{Comparison of Computational Complexity}

For problem~\eqref{opt:collective-linear1}, in the first constraints, the rows corresponding to the nodes that do not belong to the target node set $\mathbb{T}$ will not affect the objective $M$. Although we define $n$ matrix $Q_v$ for the sake of convenience, only $|\mathbb{T}|$ of them are actually effective.  For the node with $t_i=0$, the value $m_i$ will not affect the objective $M$, such that we can always set $m_i=0$, and the first constraint always holds. Hence, there are $O(3\rho^2 |\mathbb{T}|+\rho^2+\rho+|\mathbb{T}|)$ effective linear constraints, and $O(\rho^2 |\mathbb{T}|+\rho^2+\rho n +|\mathbb{T}|)$ effective variables. 

For problem~\eqref{opt:collective-linear2}, similar to~\eqref{opt:collective-linear1}, only $|\mathbb{T}|$ rows of $Q$ are actually effective.  
There are $O(3\rho |\mathbb{T}|+\rho +|\mathbb{T}|)$ effective linear constraints, and $O(\rho n+\rho |\mathbb{T}|+|\mathbb{T}|)$ effective variables. 
Our well-designed formulation makes the collective problem scalable regarding the number of injected nodes $\rho$ or the target node number $|\mathbb{T}|$. In the next section, we show that this improved LP formulation is both more efficient and effective by experimental evaluation.

\section{Experimental Evaluation}


In this section, we conduct a comprehensive evaluation of our proposed collective certificate. Given the absence of other collective baselines for graph injection attacks (GIA), we compare our collective certification \textbf{Collective-LP1} and \textbf{Collective-LP2}, with the existing \textbf{Sample-wise} approach~\cite{lai2023node}. We present a detailed analysis of the experimental results, highlighting the strengths and advantages of our collective certification methods.

\subsection{Experimental Setup}
\label{sec:experiments}
\paragraph{Datasets and Base Model.} We follow the literature~\cite{schuchardt2020collective,lai2023node} on certified robustness and evaluate our methods on two graph datasets: Cora-ML~\cite{bojchevski2017deep} and Citeseer~\cite{sen2008collective}. The Cora-ML dataset contains $2,810$ nodes, $7,981$ edges, $7$ classes, and the Citeseer contains $2,110$ nodes, $3,668$ edges, $6$ classes. We employ two representative message-passing GNNs, Graph Convolution Network (GCN)~\cite{kipf2016semi} and Graph Attention Network (GAT)~\cite{velivckovic2017graph}, with a hidden layer size of $64$ as our base classifiers. 
We use $50$ nodes per class for training and validation respectively, while the remaining as testing nodes. We also train the base model with random noise augmentation following ~\cite{lai2023node}.  

\paragraph{Threat Models and Certificate.} We set the degree constraint per injected node as the average degree of existing nodes, which are $6=\lceil 5.68 \rceil$ and $4=\lceil 3.48 \rceil$ respectively on Cora-ML and Citeseer datasets. We evaluate our proposed collective certificate with various amounts of injected nodes $\rho \in\{20,50,80,100,120,140,160\}$. Grid search is employed to find the suitable smoothing parameters $p_e$ and $p_n$ from $0.5$ to $0.9$ respectively. We exclude those parameters that lead to poor accuracy that are worse than the Multilayer Perceptron (MLP) model which does not depend on graph structure. Following~\cite{bojchevski2020efficient,lai2023node}, we employ Monte Carlo to estimate the smoothed classifier with a sample size of $N=100,000$. We apply the Clopper-Pearson confidence interval with Bonferroni correction to obtain the lower bound of $p_A$ and upper bound of $p_B$. We set the confidence level as $\alpha=0.01$. Due to the overwhelming computation cost of the original collective certifying problem known as NP-hard, we solve our proposed relaxed LP problems by default. All our collective certifying problem is solved using MOSEK~\cite{mosek} through the CVXPY~\cite{diamond2016cvxpy} interface.

\paragraph{Evaluation Metrics.}Among the testing nodes that are correctly classified, we randomly select $100$ nodes as the target node set $\mathbb{T}$. We report the \textbf{\textit{certified ratio}} on the target nodes set, which is the ratio of nodes that are certifiably robust under a given threat model. We repeat $5$ times with different random selections and report the average results. Additionally, we evaluate the global attack scenario in which the $\mathbb{T}$ is all the nodes in the graph in Appendix.~\ref{Sec:Appendix_global}. 

\subsection{Effectiveness of Collective Certified Robustness}
In this section, we aim to verify the effectiveness of our proposed collective approach in enhancing the certified robustness performance. 

\subsubsection{Comparing Collective with Sample-wise.} 
\begin{figure}[hbt]
\centering
    \subfigure[Certified Ratio (GCN)]{
    \includegraphics[width=0.19\textwidth,height=2.8cm]{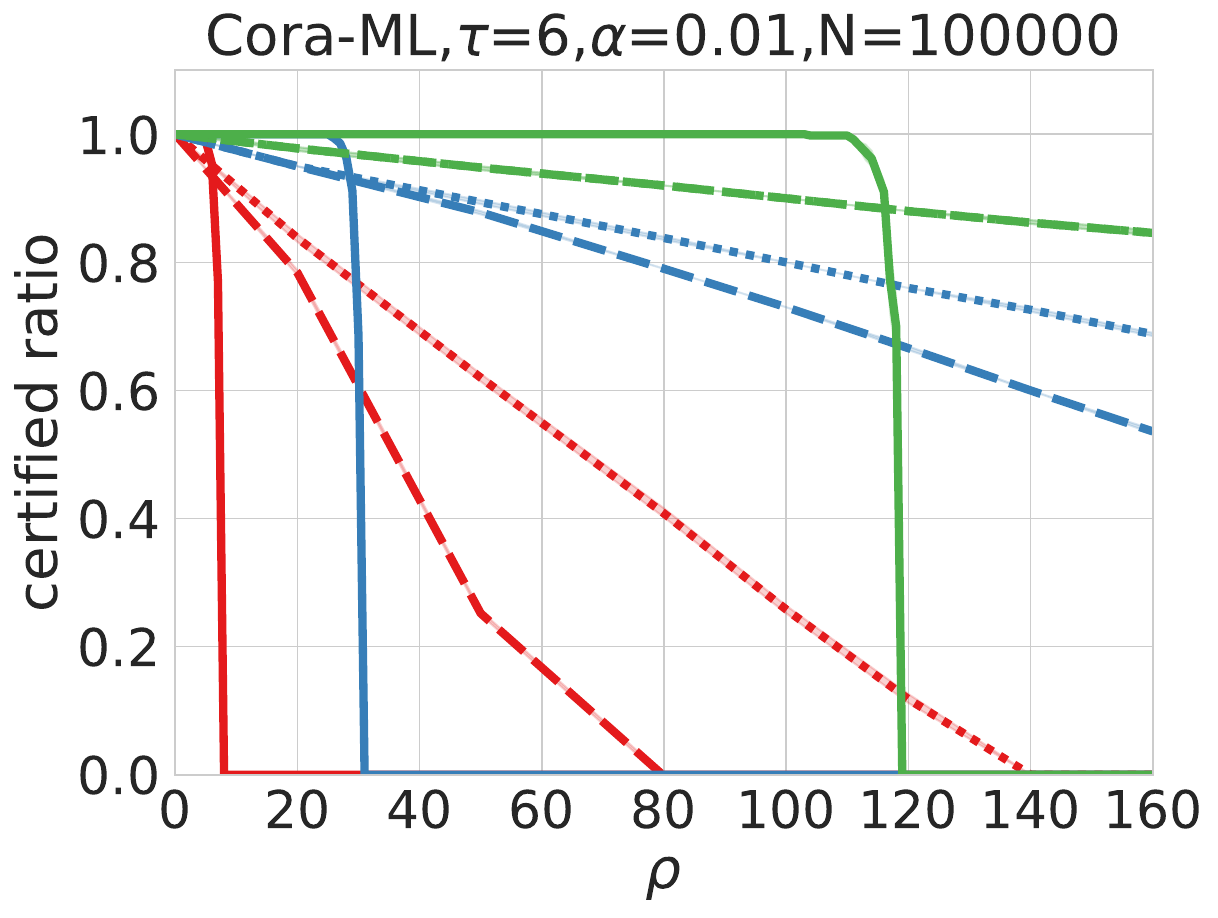}
    }
    \subfigure[Certified Ratio (GCN)]{
    \includegraphics[width=0.255\textwidth,height=2.8cm]{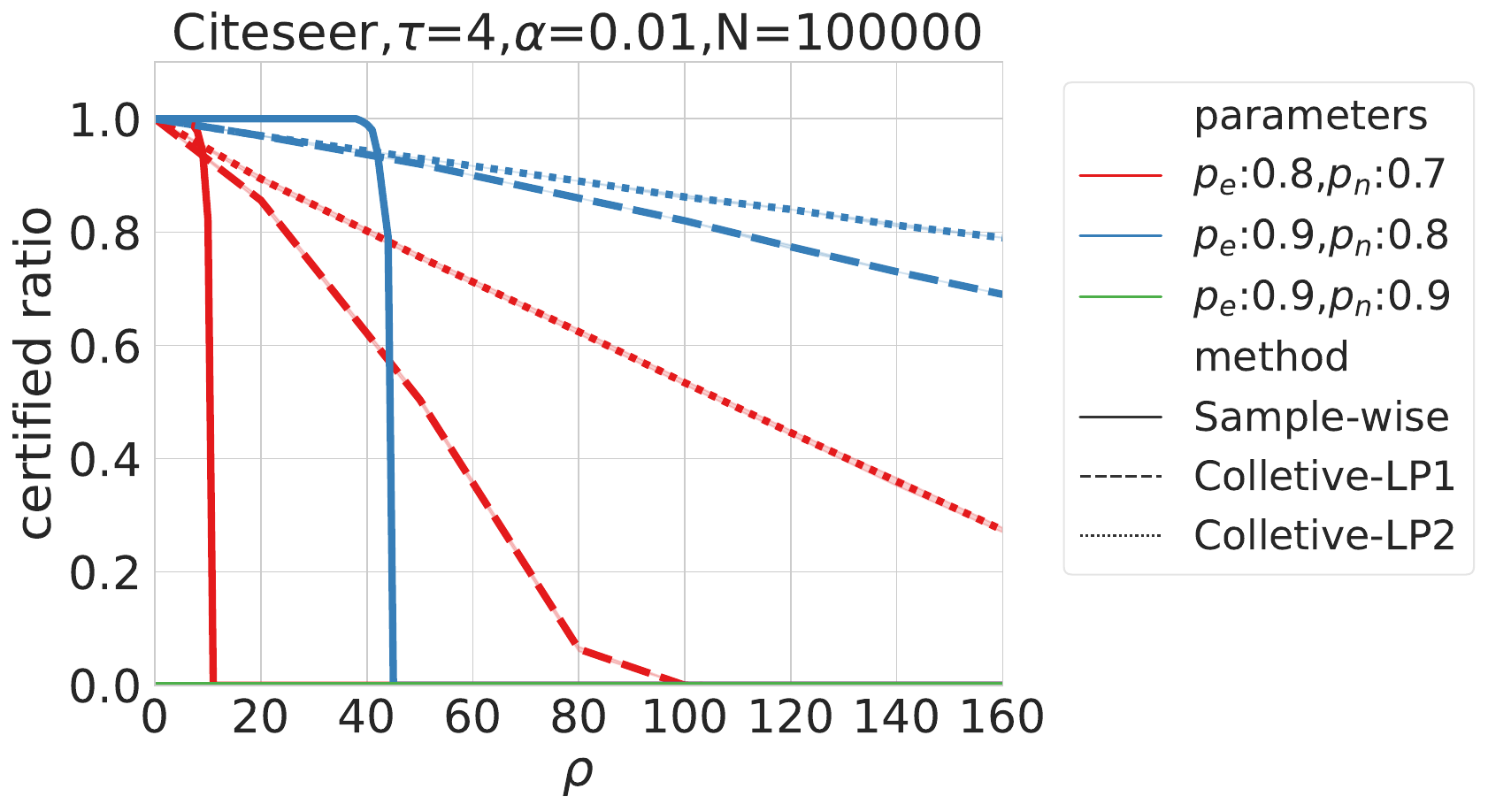}
    }
    \subfigure[Certified Ratio (GAT)]{
    \includegraphics[width=0.19\textwidth,height=2.8cm]{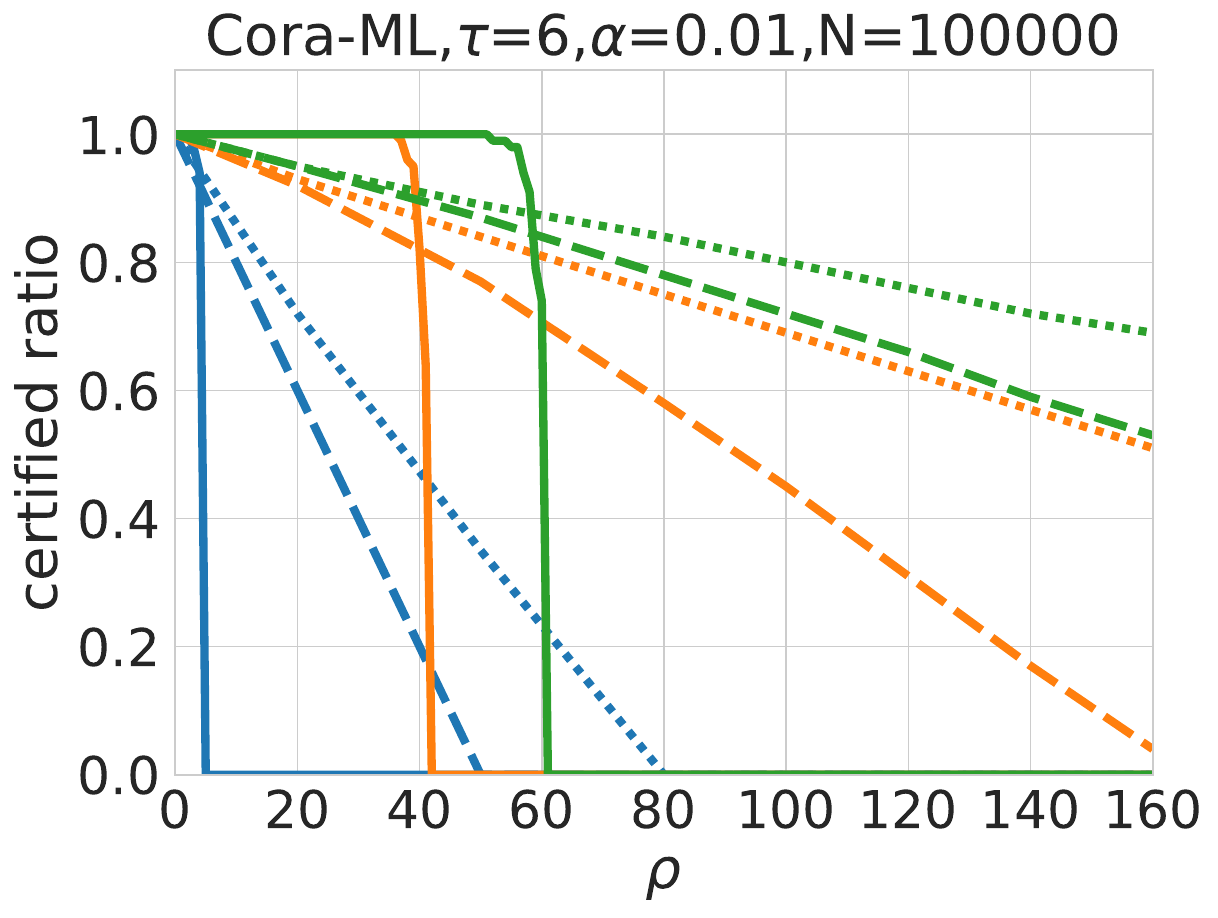}
    }
    \subfigure[Certified Ratio (GAT)]{
    \includegraphics[width=0.255\textwidth,height=2.8cm]{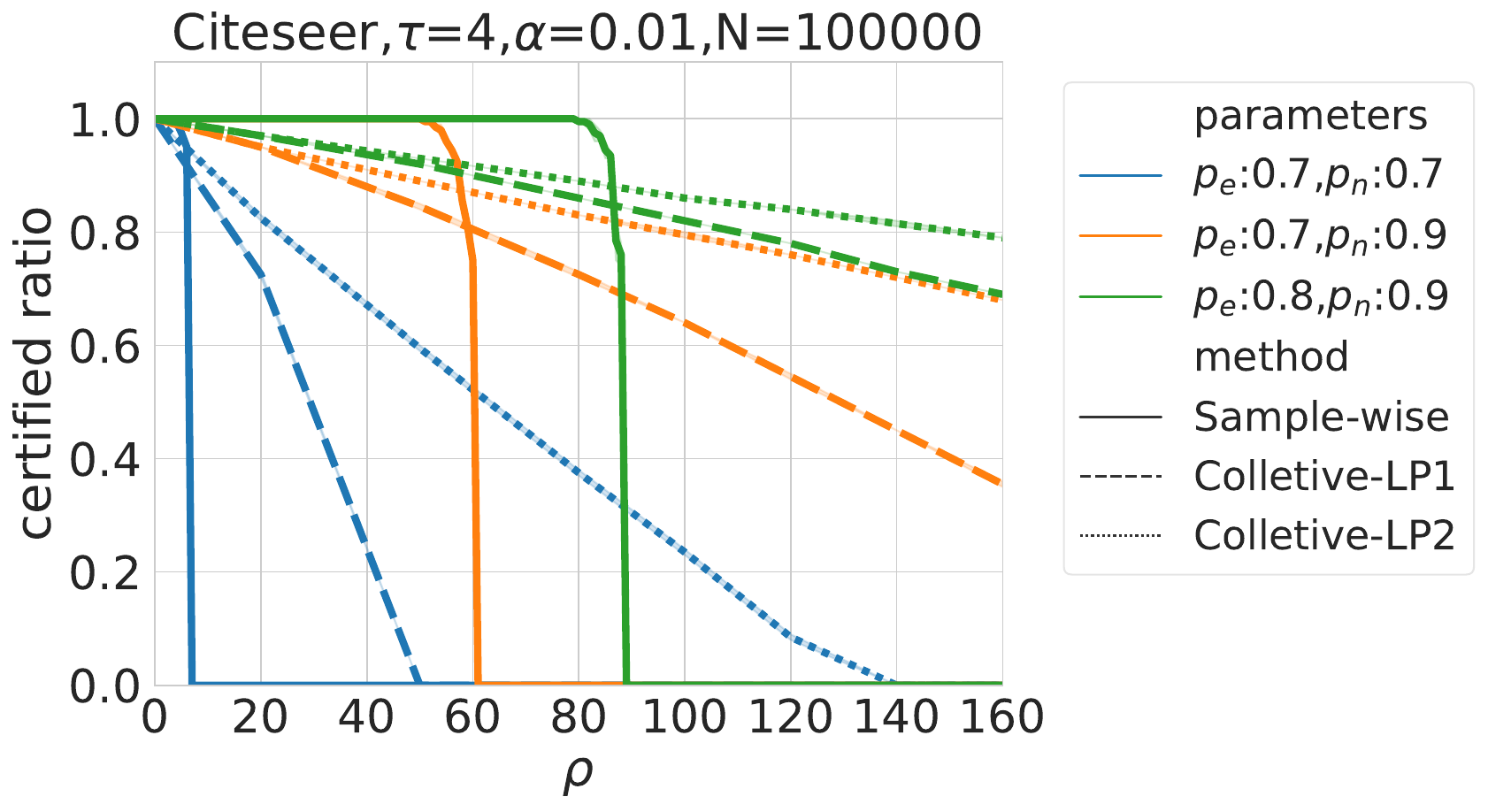}
    }
\vspace{-8pt}
\caption{Comparison of certified performance (More results with other parameters are shown in Appendix.~\ref{Sec:Appendix_D}).}
\label{fig:certify_rho}
\vspace{-5pt}
\end{figure}
In Figure.~\ref{fig:certify_rho} and Table.~\ref{tab:certify_rho}, we exhibit the certified ratio of the three certificates regarding various numbers of injected nodes $\rho$.
With the same smoothing parameter, both proposed collective certificates achieve a higher certifiable radius, outperforming the sample-wise approach significantly when the $\rho$ is large.
For example, in the Citeseer dataset, when $\rho=140$, our Collective-LP1 and Collective-LP2 have the certified ratios of $73.0\%$, and $81.2\%$, while sample-wise can certify $0.0\%$ nodes. 
Moreover, the improvement of our collective certificate is even more significant in the global attack setting (Appendix~\ref{Sec:Appendix_global}).

When the $\rho$ is small, the LP collective robustness does not outperform the sample-wise robustness. This can be attributed to the integrity gap of the relaxation technique utilized in the LP formulation, which we further illustrated in Section.~\ref{subsec:Integrity}. \textit{Interestingly, this difference becomes negligible in the case of a global attack}, as shown in Appendix~\ref{Sec:Appendix_global}.
Nevertheless, in practical scenarios, we can easily combine the sample-wise and collective certificates with minimal effort to achieve stronger certified performance in both small and large attack budgets. Since the sample-wise and collective models share the same smoothed model, we only need to estimate the smoothing prediction once to avoid extra computation. By integrating both certificates, we can leverage their respective strengths and enhance the overall robustness of the system.



\begin{table}[hbt]
\centering
\caption{Comparison of certified ratio between sample-wise and collective certifying schemes under various parameters. }
\label{tab:certify_rho}
\setlength{\tabcolsep}{1.5pt}
\footnotesize
\begin{tabular}{clrrrrr}
\toprule[1pt]
\multicolumn{2}{c}{Cora-ML   ($\tau=6$)} & \multicolumn{5}{c}{$\rho$} \\ \hline
\multicolumn{1}{c}{\begin{tabular}[c]{@{}c@{}}parameters\\  ($p_e$-$p_n$)\end{tabular}} & \multicolumn{1}{c}{methods} & \multicolumn{1}{c}{20} & \multicolumn{1}{c}{50} & \multicolumn{1}{c}{100} & \multicolumn{1}{c}{120} & \multicolumn{1}{c}{140} \\ \hline
\multirow{3}{*}{0.7-0.9} & Sample-wise & 1.000 & 0.000 & 0.000 & 0.000 & 0.000 \\
 & Collective-LP1 & 0.920 & \cellcolor{blue!6}0.768 & \cellcolor{blue!6}0.452 & \cellcolor{blue!6}0.316 & \cellcolor{blue!6}0.178 \\
 & Collective-LP2 & 0.926 & \cellcolor{blue!10}0.836 & \cellcolor{blue!10}0.686 & \cellcolor{blue!10}0.624 & \cellcolor{blue!10}0.564 \\ \hline
\multirow{3}{*}{0.9-0.8} & Sample-wise & 1.000 & 0.000 & 0.000 & 0.000 & 0.000 \\
 & Collective-LP1 & 0.950 & \cellcolor{blue!6}0.878 & \cellcolor{blue!6}0.730 & \cellcolor{blue!6}0.666 & \cellcolor{blue!6}0.600 \\
 & Collective-LP2 & 0.950 & \cellcolor{blue!10}0.894 & \cellcolor{blue!10}0.800 & \cellcolor{blue!10}0.760 & \cellcolor{blue!10}0.726 \\ \hline
\multirow{3}{*}{0.9-0.9} & Sample-wise & 1.000 & 1.000 & 1.000 & 0.000 & 0.000 \\
 & Collective-LP1 & 0.978 & 0.948 & 0.900 & \cellcolor{blue!10}0.880 & \cellcolor{blue!10}0.862 \\
 & Collective-LP2 & 0.978 & 0.948 & 0.900 & \cellcolor{blue!10}0.880 & \cellcolor{blue!10}0.862 \\ \toprule[1pt]
\multicolumn{2}{c}{Citeseer ($\tau=4$)} & \multicolumn{1}{c}{20} & \multicolumn{1}{c}{50} & \multicolumn{1}{c}{100} & \multicolumn{1}{c}{120} & \multicolumn{1}{c}{140} \\ \hline
\multirow{3}{*}{0.7-0.9} & Sample-wise & 1.000 & 0.990 & 0.000 & 0.000 & 0.000 \\
 & Collective-LP1 & 0.950 & 0.846 & \cellcolor{blue!6}0.640 & \cellcolor{blue!6}0.546 & \cellcolor{blue!6}0.452 \\
 & Collective-LP2 & 0.950 & 0.892 & \cellcolor{blue!10}0.796 & \cellcolor{blue!10}0.756 & \cellcolor{blue!10}0.718 \\ \hline
\multirow{3}{*}{0.8-0.7} & Sample-wise & 0.000 & 0.000 & 0.000 & 0.000 & 0.000 \\
 & Collective-LP1 & \cellcolor{blue!6}0.856 & \cellcolor{blue!6}0.504 & \cellcolor{blue!4}0.000 & \cellcolor{blue!4}0.000 & \cellcolor{blue!4}0.000 \\
 & Collective-LP2 & \cellcolor{blue!10}0.894 & \cellcolor{blue!10}0.756 &\cellcolor{blue!10}0.534 & \cellcolor{blue!10}0.446 & \cellcolor{blue!10}0.360 \\ \hline
\multirow{3}{*}{0.9-0.8} & Sample-wise & 1.000 & 0.000 & 0.000 & 0.000 & 0.000 \\
 & Collective-LP1 & 0.970 & \cellcolor{blue!6}0.920 & \cellcolor{blue!6}0.820 & \cellcolor{blue!6}0.775 & \cellcolor{blue!6}0.730 \\
 & Collective-LP2 & 0.970 & \cellcolor{blue!10}0.930 & \cellcolor{blue!10}0.862 & \cellcolor{blue!10}0.840 & \cellcolor{blue!10}0.812 \\ \toprule[1pt]
\end{tabular}
\vspace{-10pt}
\end{table}

\begin{figure}[hbt!]
\centering
    \subfigure[smaller $\rho$ (GCN)]{
    \includegraphics[width=0.19\textwidth,height=2.8cm]{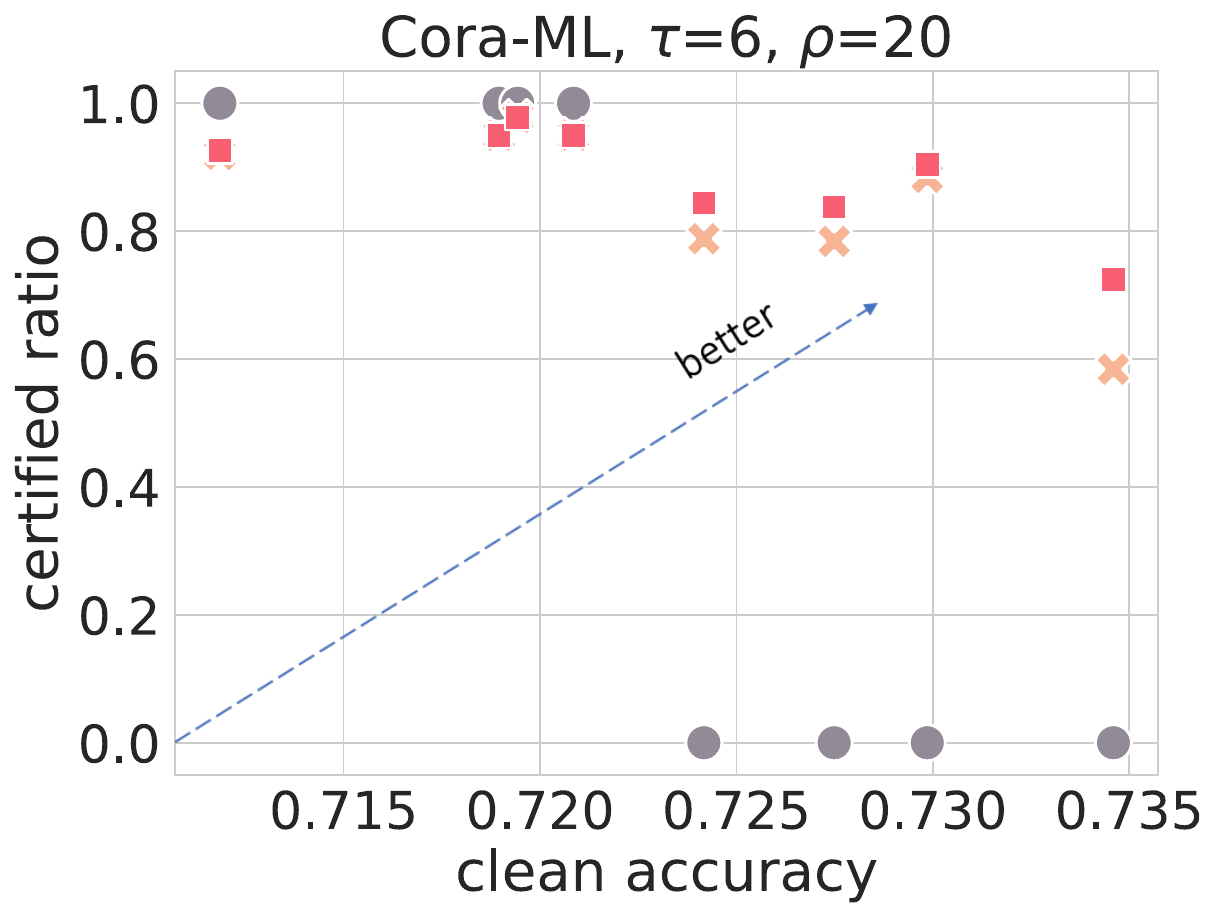}
    }
    \subfigure[smaller $\rho$ (GCN)]{
    \includegraphics[width=0.265\textwidth,height=2.8cm]{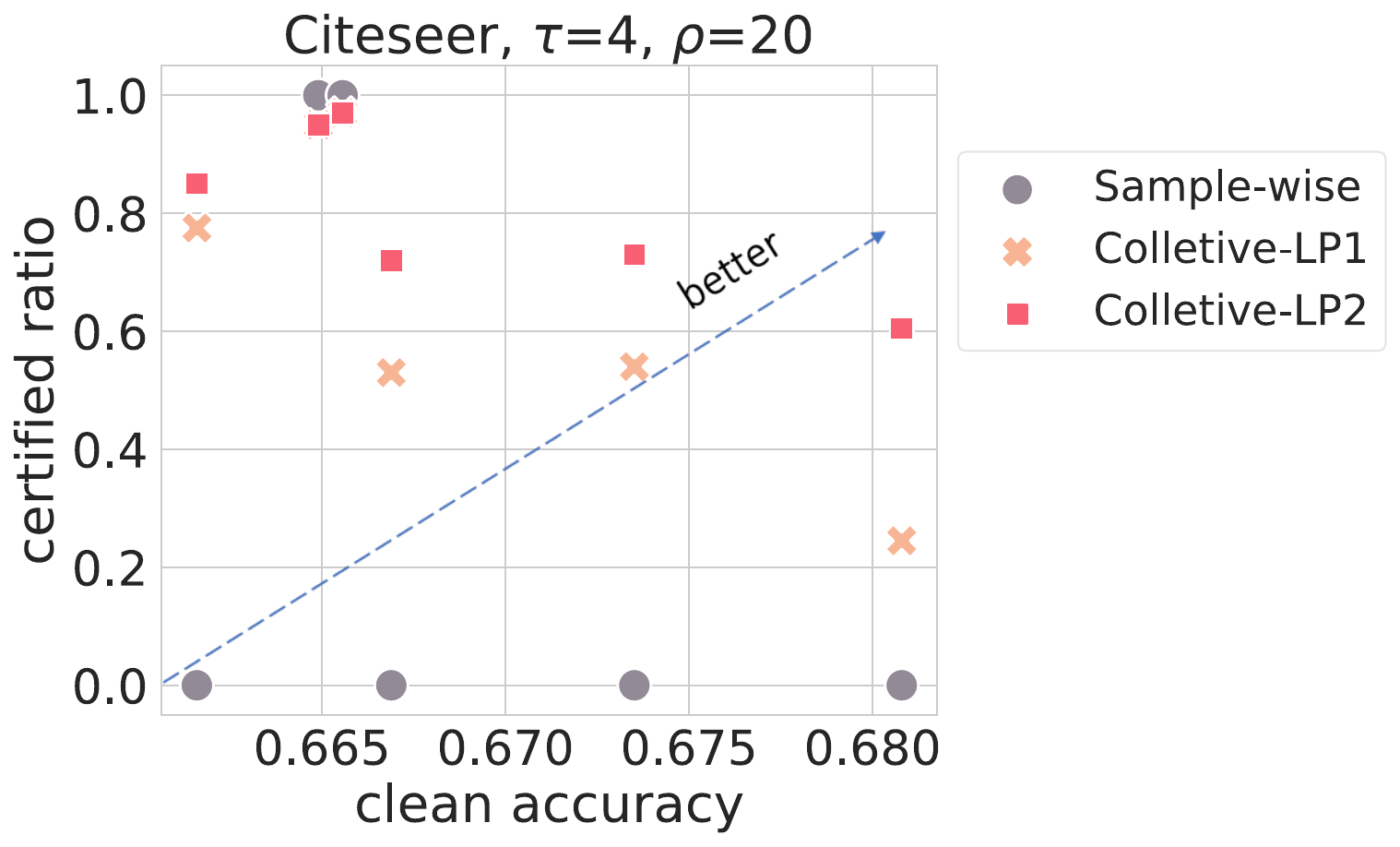}
    }
    \subfigure[larger $\rho$ (GCN)]{
    \includegraphics[width=0.19\textwidth,height=2.8cm]{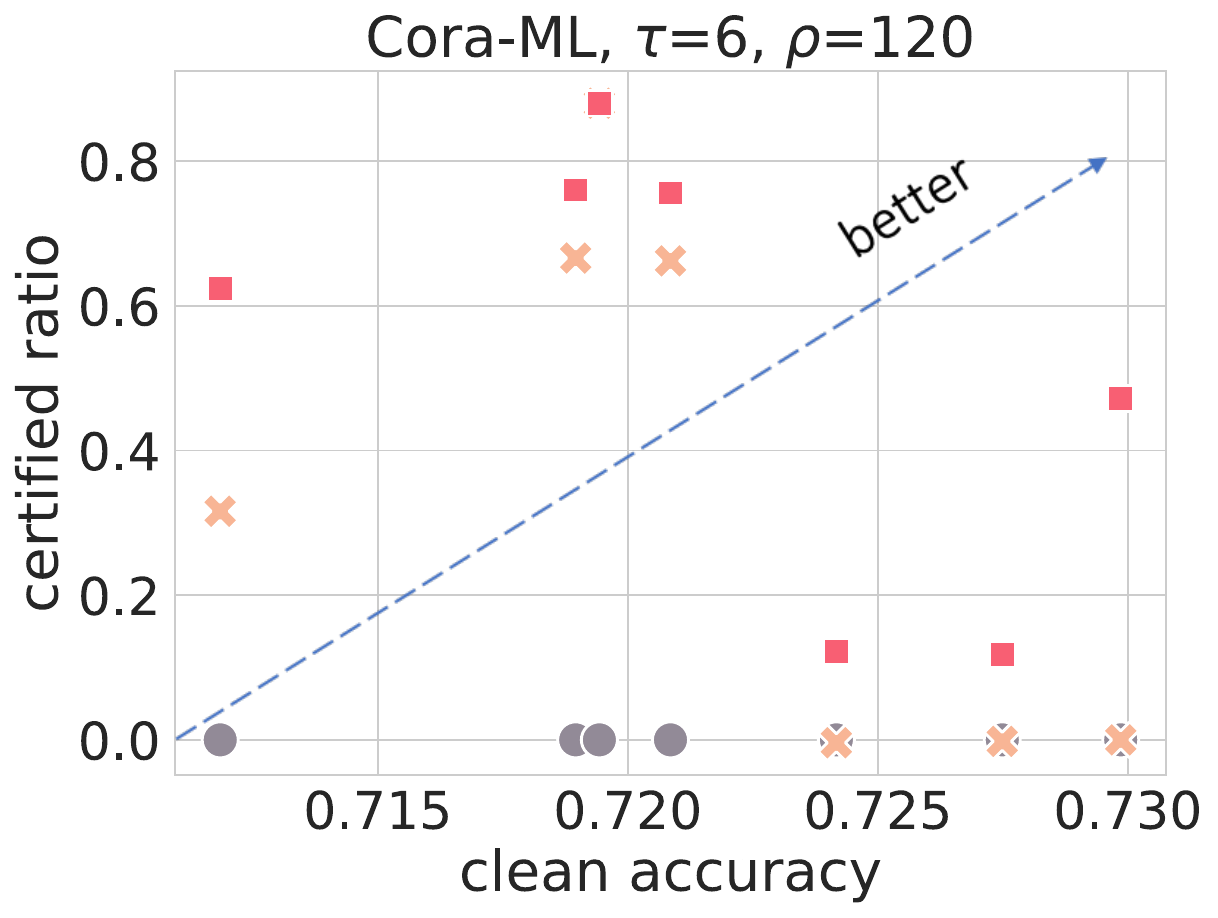}
    }
    \subfigure[larger $\rho$ (GCN)]{
    \includegraphics[width=0.265\textwidth,height=2.8cm]{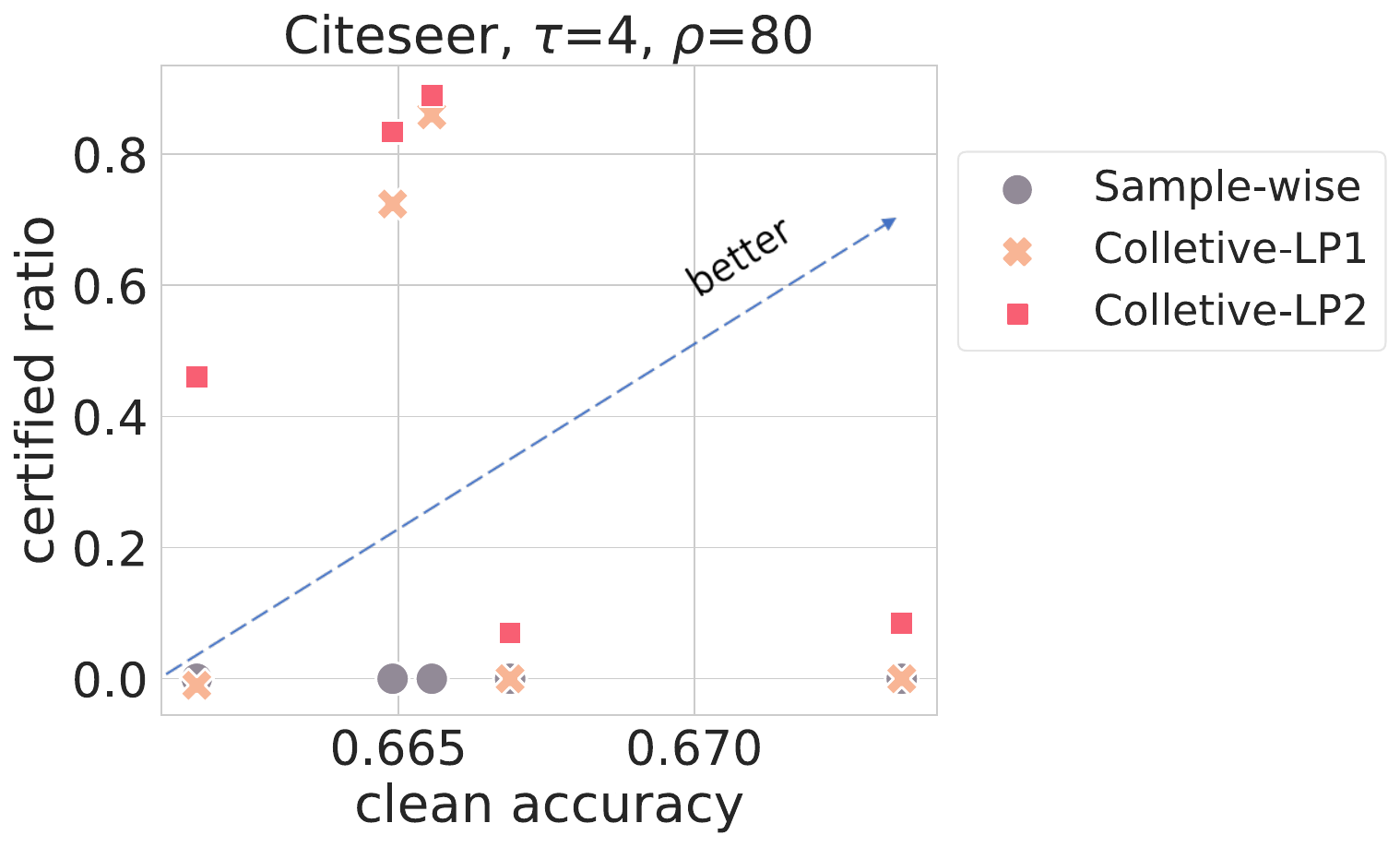}
    }
    \subfigure[larger $\rho$ (GAT)]{
    \includegraphics[width=0.19\textwidth,height=2.8cm]{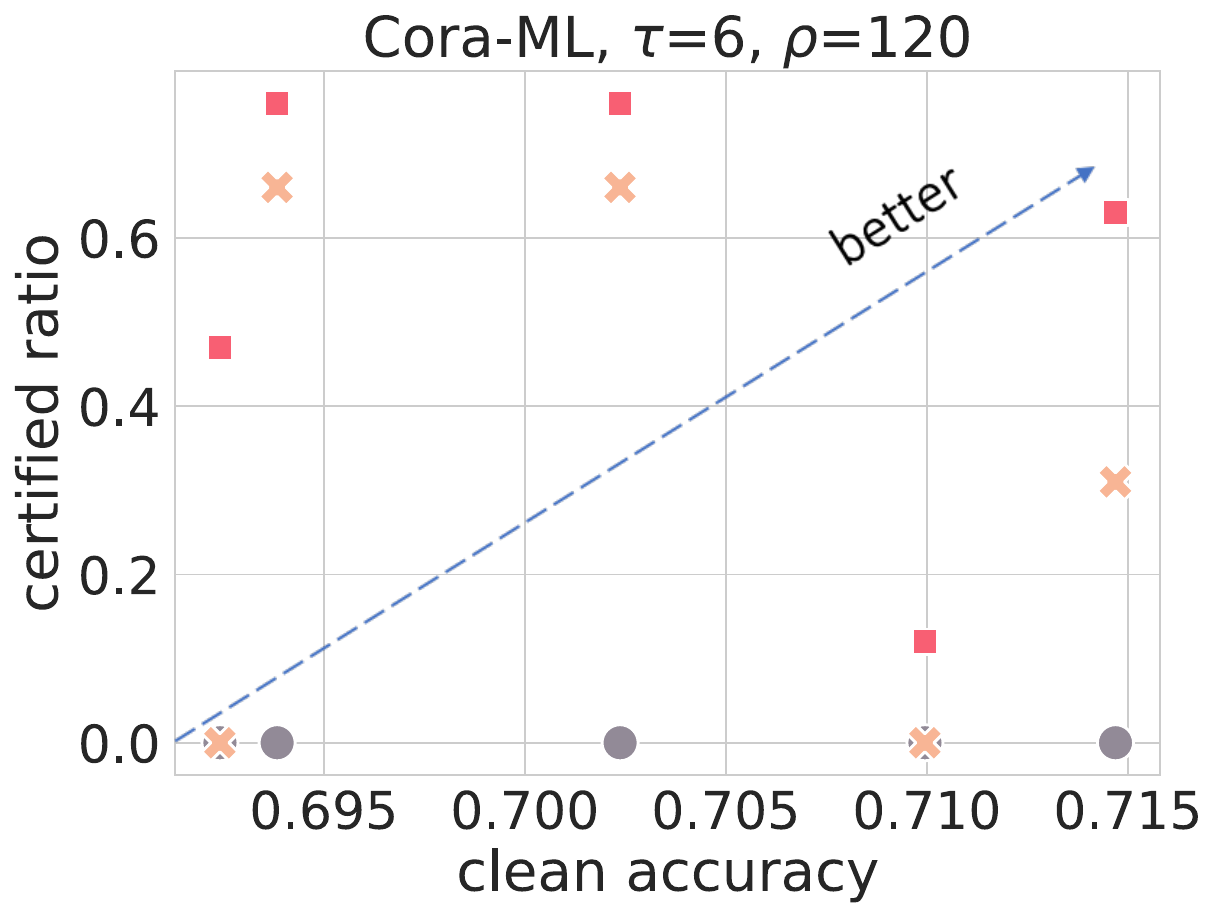}
    }
    \subfigure[larger $\rho$ (GAT)]{
    \includegraphics[width=0.255\textwidth,height=2.8cm]{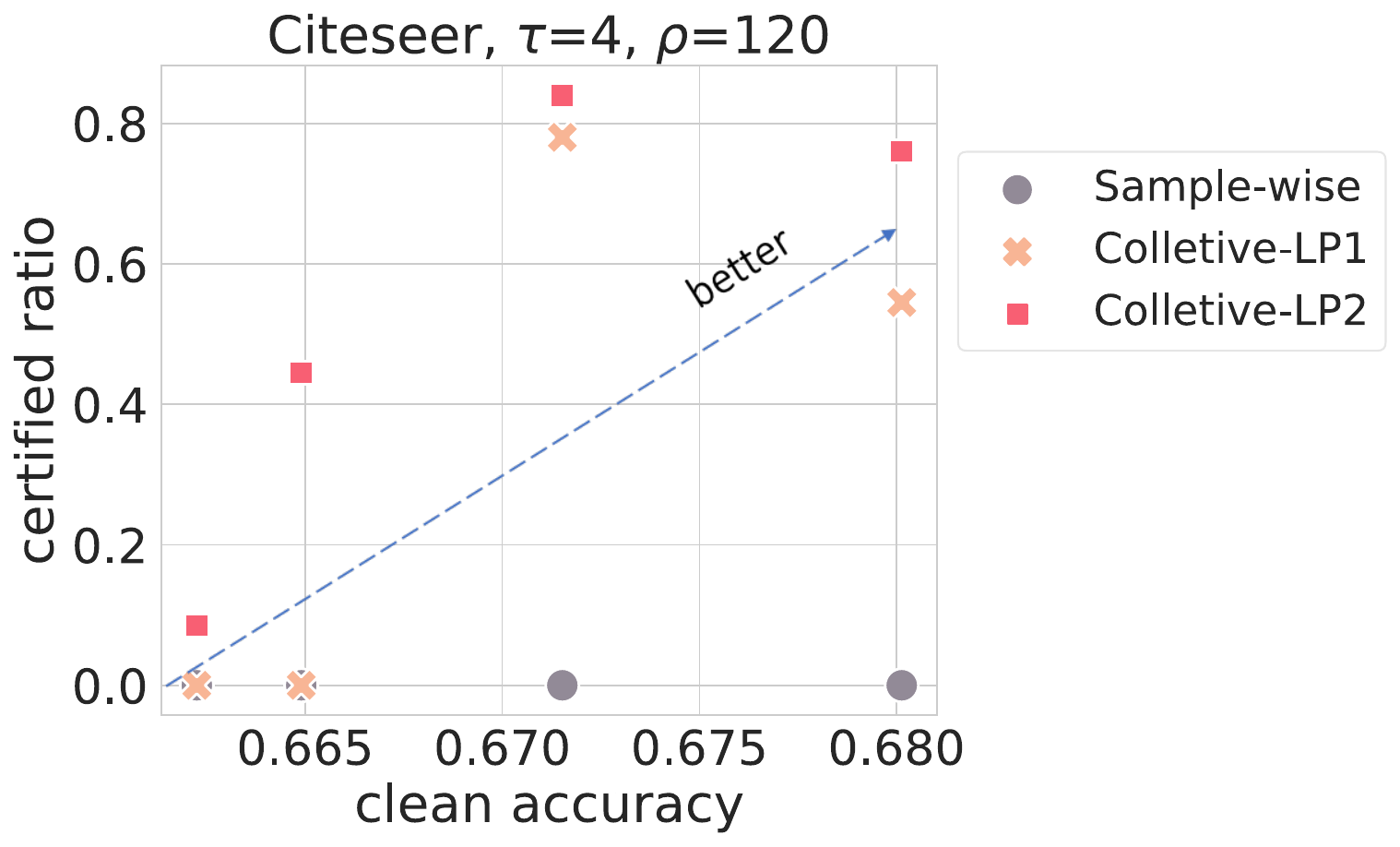}
    }
\vspace{-8pt}
\caption{Trade-off between clean accuracy and certified ratio (More results with other $\rho$ are shown in Appendix.~\ref{Sec:Appendix_D}).}
\label{fig:certify_clean}
\vspace{-8pt}
\end{figure}

A superior certifying scheme should not only possess a higher certified ratio but also a higher clean accuracy that represents the initial performance of the model. We also evaluate the trade-off between the certified ratio and the clean accuracy of the smoothed model in Figure.~\ref{fig:certify_clean}. 
As we employ the same smoothed model, both the collective scheme and the sample-wise scheme exhibit the same clean accuracy when they share identical smoothing parameters, while our collective approach consistently achieves a higher certified ratio, particularly when $\rho$ exceeds the certifiable radius of the sample-wise approach. Finally, these results highlight the advantageous trade-off achieved by our proposed collective approach in both smaller $\rho$ and larger $\rho$.

\subsubsection{Comparing two Collective Certificates.} 
In comparing our two LP-based collective certificates, it is evident that our customized relaxation (Collective-LP2) consistently achieves higher or equivalent certified ratios compared to the standard technique (Collective-LP1). 
For instance, in the Cora-ML dataset, when $p_e=0.7$, $p_n=0.9$, and $\rho=140$, Collective-LP2 improves the certified ratio by $216\%$ compared to Collective-LP1 (Table.~\ref{tab:certify_rho}).
Furthermore, with the same clean accuracy, Collective-LP2 is always superior to Collective-LP1 in certified ratios (Figure.~\ref{fig:certify_clean}). 

\begin{figure}[hbt]
\centering
    \subfigure[Runtime]{\includegraphics[width=0.19\textwidth,height=2.8cm]{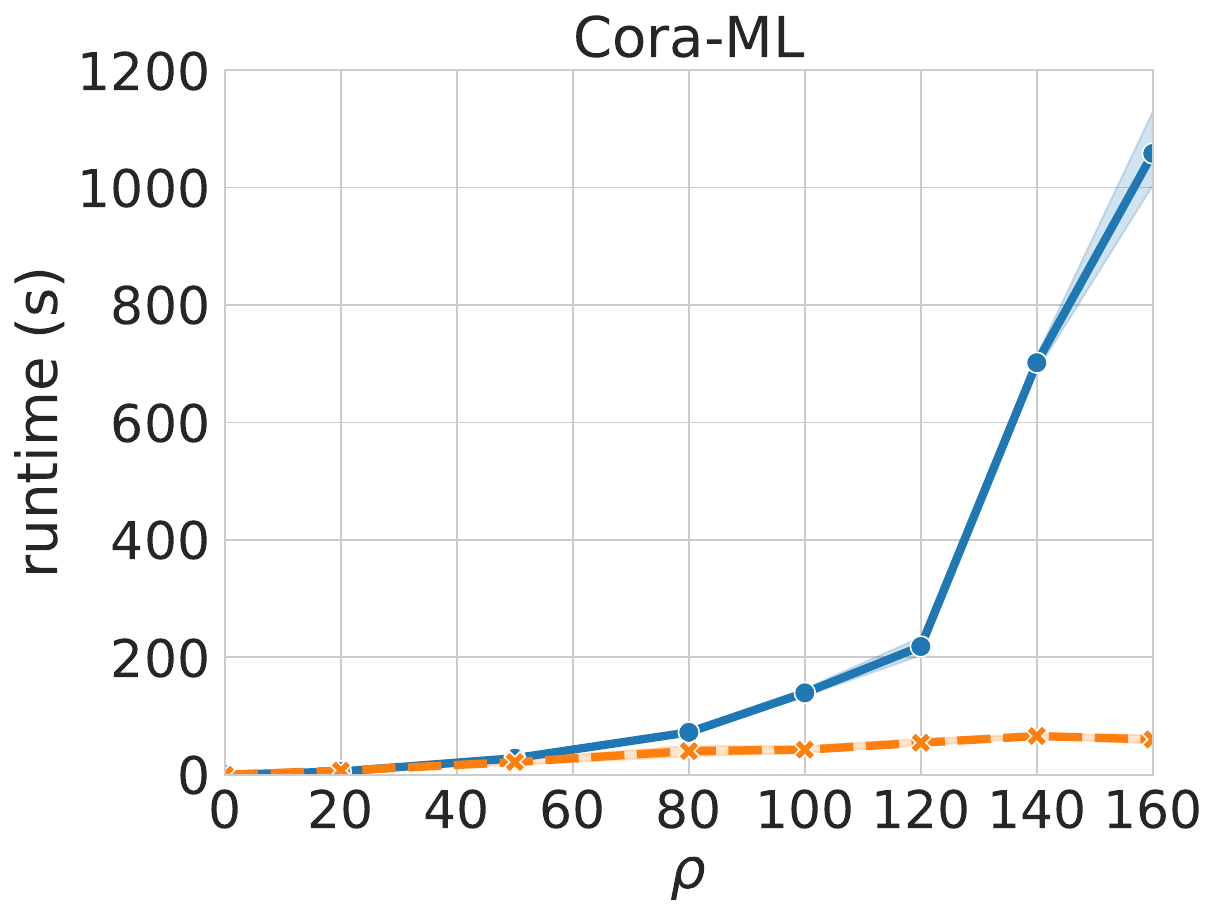}
    }
    \subfigure[Runtime]{\includegraphics[width=0.255\textwidth,height=2.8cm]{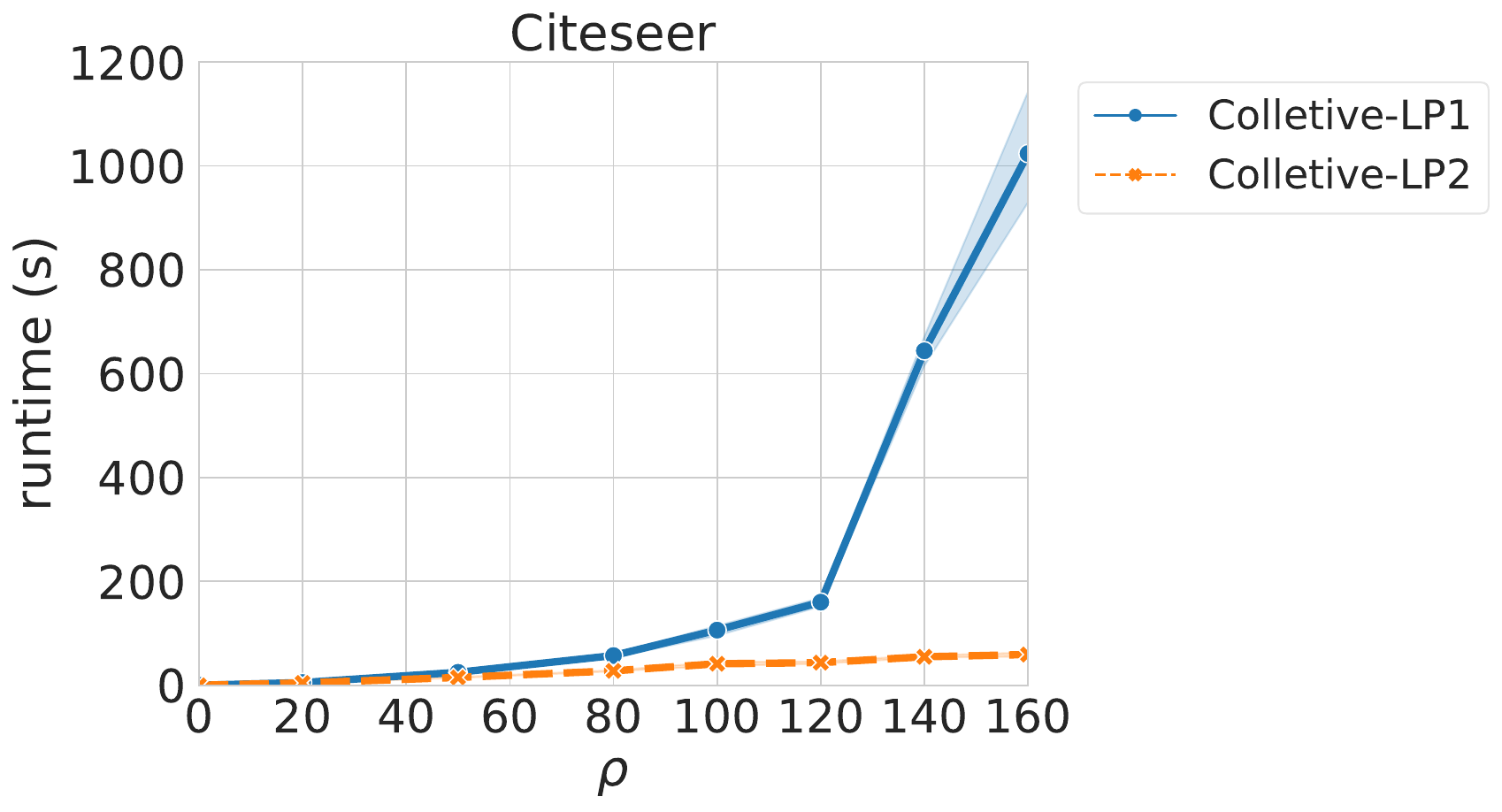}
    }
\vspace{-8pt}
\caption{Runtime comparison of LP collective models.}
\label{fig:cer_time}
\vspace{-5pt}
\end{figure}
In Figure~\ref{fig:cer_time}, we present a comparison of the runtime between our two LP-based collective certificates. It is evident that Collective-LP2 exhibits a significantly lower runtime compared to Collective-LP1, particularly as $\rho$ increases.
Remarkably, even for a larger value of $\rho$ like $\rho=140$, our Collective-LP2 can be solved in approximately $1$ minute. This indicates the practicality and efficiency of our proposed method, making it feasible for real-world scenarios with larger attack budgets.

\subsection{Effectiveness of Linear Relaxation}
\label{subsec:Integrity}
\begin{figure}[hbt]
\centering
\vspace{-5pt}
\subfigure[Integrity Gap]{
\includegraphics[width=0.19\textwidth,height=2.8cm]{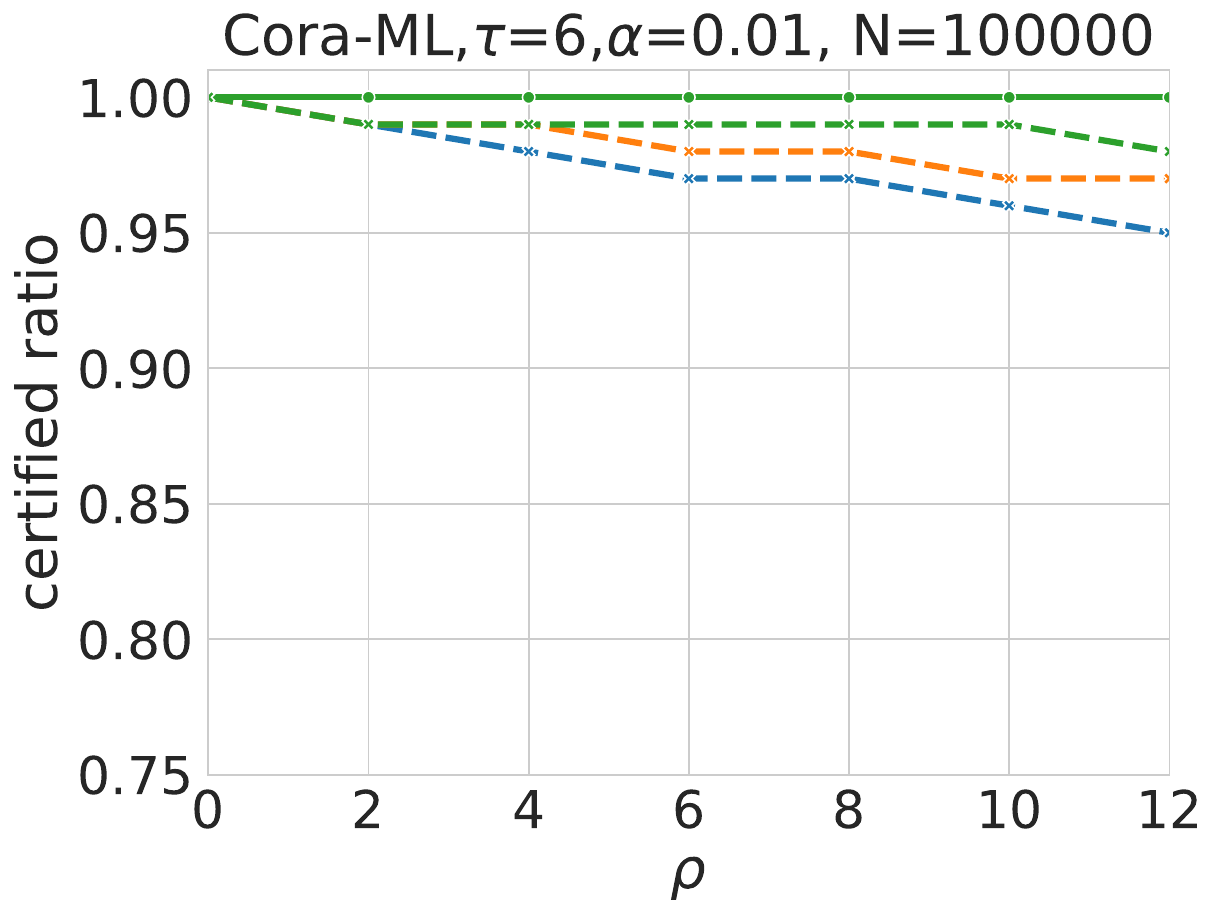}
    }
    \subfigure[Integrity Gap]{
    \includegraphics[width=0.255\textwidth,height=2.8cm]{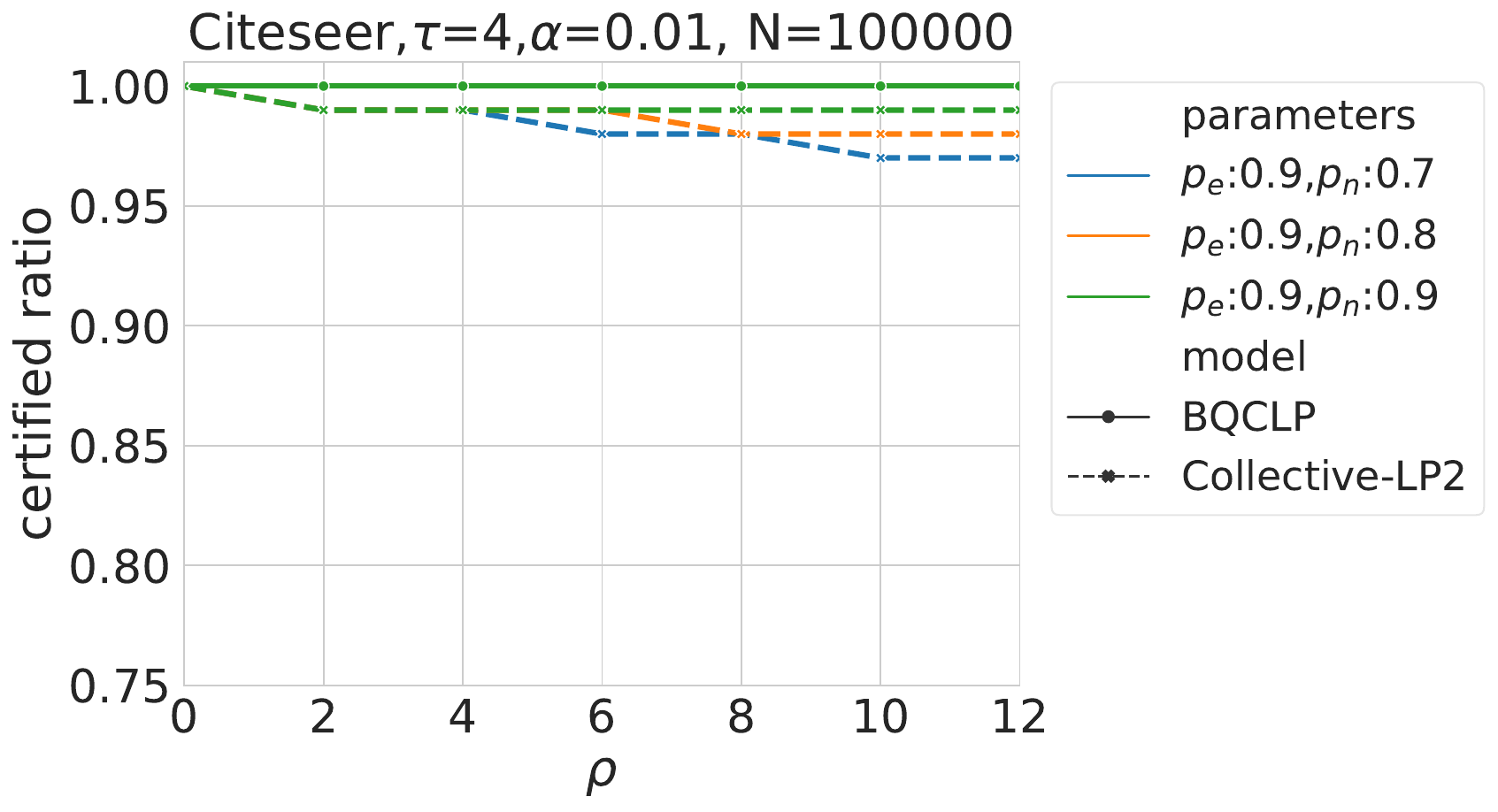}
    }
\vspace{-8pt}
\caption{Certified ratio comparison between optimizing original BQCLP problem and relaxed LP problem. }
\label{fig:integrity}
\vspace{-5pt}
\end{figure}


In this section, we investigate the impact of our LP relaxation technique on the certified performance of our collective certification method. Specifically, we compare the certified ratios obtained from both the original integer problem (BQCLP) and the LP problem (Collective-LP2). Figure.~\ref{fig:integrity} provides a graphical representation of these results.
Due to the computational overhead associated with solving the integer problem, we limit our analysis to a smaller attack budget, $\rho\leq 12$. We observe that the certified ratio of the integer problem remains relatively stable as $\rho$ increases. However, the certified ratio of Collective-LP2 undergoes a decline of approximately $5\%$.
This decrease in certified performance is attributed to the sacrifice made in the relaxation process of the LP formulation. It also partially explains why our approach may exhibit a weaker certified ratio compared to the sample-wise approach when $\rho$ is small.


\vspace{-1pt}
\section{Related Work}
In this section, we summarize the previous work that is closely related to certified robustness. 
Randomized smoothing has emerged as a prominent black-box technique that provides certified robustness. It was first proposed for defending against $l_2$ norm ball perturbation in the computer vision models~\cite{cohen2019certified}. Recent work extends it to certify graph node classification tasks~\cite{bojchevski2020efficient,wang2021certified,jia2020certified,jia2022almost,scholten2022randomized} against $l_0$-norm ball perturbation, typically the graph modification attacks (GMAs). To improve the certified performance, some researchers~\cite{schuchardt2020collective,schuchardt2023localized} develop collective robustness schemes. These schemes assume a realistic attacker whose objective is to perturb a set of nodes simultaneously, thereby improving the overall robustness against adversarial attacks.  

Despite the progress made in defending against GMAs, the robustness against graph injection attacks (GIAs) has received relatively little attention. \cite{jia2023pore,lai2023node} further extended it to certify against GIAs. However, these models provide sample-wise certificates instead of collective ones. To the best of our knowledge, there is currently no collective certificate designed for GIAs.

\vspace{-1pt}
\section{Conclusion}

In this paper, we present the \textit{first} collective robustness certificate specifically designed for defending against graph injection attacks (GIAs), which encompass edge addition perturbations known to be more challenging to certify than edge deletions.
Our collective certificate improves the certified performance by assuming that the attacker's objective is to disrupt the predictions of as many target nodes as possible, using a shared single graph instead of different graphs for each node. We model the collective certifying problem by upper-bounding the number of non-robust nodes under a worst-case attacker, such that the remaining nodes are guaranteed to be robust. However, it yields a binary quadratic constrained programming that is NP-hard. To address this, we propose novel relaxations to formulate the problem into linear programming that can be efficiently solved. Extensive experimental results demonstrate that our proposed collective certificate achieves significantly higher certified ratios and larger certifiable radii compared to existing approaches. 

\bibliographystyle{icml2024}
\bibliography{references}

\begin{thebibliography}{33}
\providecommand{\natexlab}[1]{#1}
\providecommand{\url}[1]{\texttt{#1}}
\expandafter\ifx\csname urlstyle\endcsname\relax
  \providecommand{\doi}[1]{doi: #1}\else
  \providecommand{\doi}{doi: \begingroup \urlstyle{rm}\Url}\fi

\bibitem[ApS(2019)]{mosek}
ApS, M.
\newblock \emph{The MOSEK optimization toolbox for MATLAB manual. Version
  9.0.}, 2019.
\newblock URL \url{http://docs.mosek.com/9.0/toolbox/index.html}.

\bibitem[Bojchevski \& G{\"u}nnemann(2017)Bojchevski and
  G{\"u}nnemann]{bojchevski2017deep}
Bojchevski, A. and G{\"u}nnemann, S.
\newblock Deep gaussian embedding of graphs: Unsupervised inductive learning
  via ranking.
\newblock \emph{arXiv preprint arXiv:1707.03815}, 2017.

\bibitem[Bojchevski et~al.(2020)Bojchevski, Gasteiger, and
  G{\"u}nnemann]{bojchevski2020efficient}
Bojchevski, A., Gasteiger, J., and G{\"u}nnemann, S.
\newblock Efficient robustness certificates for discrete data: Sparsity-aware
  randomized smoothing for graphs, images and more.
\newblock In \emph{International Conference on Machine Learning}, pp.\
  1003--1013. PMLR, 2020.

\bibitem[Chen et~al.(2022)Chen, Yang, Zhang, KAILI, Liu, Han, and
  Cheng]{chen2022understanding}
Chen, Y., Yang, H., Zhang, Y., KAILI, M., Liu, T., Han, B., and Cheng, J.
\newblock Understanding and improving graph injection attack by promoting
  unnoticeability.
\newblock In \emph{International Conference on Learning Representations}, 2022.

\bibitem[Cohen et~al.(2019)Cohen, Rosenfeld, and Kolter]{cohen2019certified}
Cohen, J., Rosenfeld, E., and Kolter, Z.
\newblock Certified adversarial robustness via randomized smoothing.
\newblock In \emph{international conference on machine learning}, pp.\
  1310--1320. PMLR, 2019.

\bibitem[Diamond \& Boyd(2016)Diamond and Boyd]{diamond2016cvxpy}
Diamond, S. and Boyd, S.
\newblock {CVXPY}: {A} {P}ython-embedded modeling language for convex
  optimization.
\newblock \emph{Journal of Machine Learning Research}, 17\penalty0
  (83):\penalty0 1--5, 2016.

\bibitem[Geisler et~al.(2020)Geisler, Z{\"u}gner, and
  G{\"u}nnemann]{geisler2020reliable}
Geisler, S., Z{\"u}gner, D., and G{\"u}nnemann, S.
\newblock Reliable graph neural networks via robust aggregation.
\newblock \emph{Advances in Neural Information Processing Systems},
  33:\penalty0 13272--13284, 2020.

\bibitem[Gilmer et~al.(2017)Gilmer, Schoenholz, Riley, Vinyals, and
  Dahl]{gilmer2017neural}
Gilmer, J., Schoenholz, S.~S., Riley, P.~F., Vinyals, O., and Dahl, G.~E.
\newblock Neural message passing for quantum chemistry.
\newblock In \emph{International conference on machine learning}, pp.\
  1263--1272. PMLR, 2017.

\bibitem[Gosch et~al.(2023)Gosch, Geisler, Sturm, Charpentier, Z{\"u}gner, and
  G{\"u}nnemann]{gosch2023adversarial}
Gosch, L., Geisler, S., Sturm, D., Charpentier, B., Z{\"u}gner, D., and
  G{\"u}nnemann, S.
\newblock Adversarial training for graph neural networks: Pitfalls, solutions,
  and new directions.
\newblock In \emph{Thirty-seventh Conference on Neural Information Processing
  Systems}, 2023.

\bibitem[Jia et~al.(2020)Jia, Wang, Cao, and Gong]{jia2020certified}
Jia, J., Wang, B., Cao, X., and Gong, N.~Z.
\newblock Certified robustness of community detection against adversarial
  structural perturbation via randomized smoothing.
\newblock In \emph{Proceedings of The Web Conference 2020}, pp.\  2718--2724,
  2020.

\bibitem[Jia et~al.(2022)Jia, Wang, Cao, Liu, and Gong]{jia2022almost}
Jia, J., Wang, B., Cao, X., Liu, H., and Gong, N.~Z.
\newblock Almost tight l0-norm certified robustness of top-k predictions
  against adversarial perturbations.
\newblock In \emph{International Conference on Learning Representations}, 2022.

\bibitem[Jia et~al.(2023)Jia, Liu, Hu, and Gong]{jia2023pore}
Jia, J., Liu, Y., Hu, Y., and Gong, N.~Z.
\newblock Pore: Provably robust recommender systems against data poisoning
  attacks.
\newblock In \emph{32nd USENIX Security Symposium (USENIX Security 23)}, pp.\
  1703--1720, 2023.

\bibitem[Jin et~al.(2020)Jin, Ma, Liu, Tang, Wang, and Tang]{jin2020graph}
Jin, W., Ma, Y., Liu, X., Tang, X., Wang, S., and Tang, J.
\newblock Graph structure learning for robust graph neural networks.
\newblock In \emph{Proceedings of the 26th ACM SIGKDD international conference
  on knowledge discovery \& data mining}, pp.\  66--74, 2020.

\bibitem[Ju et~al.(2023)Ju, Fan, Zhang, and Ye]{ju2023let}
Ju, M., Fan, Y., Zhang, C., and Ye, Y.
\newblock Let graph be the go board: gradient-free node injection attack for
  graph neural networks via reinforcement learning.
\newblock In \emph{Proceedings of the AAAI Conference on Artificial
  Intelligence}, pp.\  4383--4390, 2023.

\bibitem[Kipf \& Welling(2016)Kipf and Welling]{kipf2016semi}
Kipf, T.~N. and Welling, M.
\newblock Semi-supervised classification with graph convolutional networks.
\newblock \emph{arXiv preprint arXiv:1609.02907}, 2016.

\bibitem[Lai et~al.(2023)Lai, Zhu, Pan, and Zhou]{lai2023node}
Lai, Y., Zhu, Y., Pan, B., and Zhou, K.
\newblock Node-aware bi-smoothing: Certified robustness against graph injection
  attacks.
\newblock \emph{arXiv preprint arXiv:2312.03979}, 2023.

\bibitem[Li et~al.(2023)Li, Xie, and Li]{li2023sok}
Li, L., Xie, T., and Li, B.
\newblock Sok: Certified robustness for deep neural networks.
\newblock In \emph{2023 IEEE Symposium on Security and Privacy (SP)}, pp.\
  1289--1310. IEEE, 2023.

\bibitem[Liu et~al.(2022)Liu, Luo, Wu, Liu, and Li]{liu2022towards}
Liu, Z., Luo, Y., Wu, L., Liu, Z., and Li, S.~Z.
\newblock Towards reasonable budget allocation in untargeted graph structure
  attacks via gradient debias.
\newblock \emph{Advances in Neural Information Processing Systems},
  35:\penalty0 27966--27977, 2022.

\bibitem[Scholten et~al.(2022)Scholten, Schuchardt, Geisler, Bojchevski, and
  G{\"u}nnemann]{scholten2022randomized}
Scholten, Y., Schuchardt, J., Geisler, S., Bojchevski, A., and G{\"u}nnemann,
  S.
\newblock Randomized message-interception smoothing: Gray-box certificates for
  graph neural networks.
\newblock In \emph{Advances in Neural Information Processing Systems}, 2022.
\newblock URL \url{https://openreview.net/forum?id=t0VbBTw-o8}.

\bibitem[Schuchardt et~al.(2020)Schuchardt, Bojchevski, Gasteiger, and
  G{\"u}nnemann]{schuchardt2020collective}
Schuchardt, J., Bojchevski, A., Gasteiger, J., and G{\"u}nnemann, S.
\newblock Collective robustness certificates: Exploiting interdependence in
  graph neural networks.
\newblock In \emph{International Conference on Learning Representations}, 2020.

\bibitem[Schuchardt et~al.(2023)Schuchardt, Wollschl{\"a}ger, Bojchevski, and
  G{\"u}nnemann]{schuchardt2023localized}
Schuchardt, J., Wollschl{\"a}ger, T., Bojchevski, A., and G{\"u}nnemann, S.
\newblock Localized randomized smoothing for collective robustness
  certification.
\newblock In \emph{International Conference on Learning Representations}, 2023.

\bibitem[Sen et~al.(2008)Sen, Namata, Bilgic, Getoor, Galligher, and
  Eliassi-Rad]{sen2008collective}
Sen, P., Namata, G., Bilgic, M., Getoor, L., Galligher, B., and Eliassi-Rad, T.
\newblock Collective classification in network data.
\newblock \emph{AI magazine}, 29\penalty0 (3):\penalty0 93--93, 2008.

\bibitem[Tao et~al.(2023)Tao, Cao, Shen, Wu, Hou, Sun, and
  Cheng]{tao2023adversarial}
Tao, S., Cao, Q., Shen, H., Wu, Y., Hou, L., Sun, F., and Cheng, X.
\newblock Adversarial camouflage for node injection attack on graphs.
\newblock \emph{Information Sciences}, 649:\penalty0 119611, 2023.

\bibitem[Veli{\v{c}}kovi{\'c} et~al.(2017)Veli{\v{c}}kovi{\'c}, Cucurull,
  Casanova, Romero, Lio, and Bengio]{velivckovic2017graph}
Veli{\v{c}}kovi{\'c}, P., Cucurull, G., Casanova, A., Romero, A., Lio, P., and
  Bengio, Y.
\newblock Graph attention networks.
\newblock \emph{arXiv preprint arXiv:1710.10903}, 2017.

\bibitem[Veli{\v{c}}kovi{\'c} et~al.(2018)Veli{\v{c}}kovi{\'c}, Cucurull,
  Casanova, Romero, Li{\`o}, and Bengio]{velivckovic2018graph}
Veli{\v{c}}kovi{\'c}, P., Cucurull, G., Casanova, A., Romero, A., Li{\`o}, P.,
  and Bengio, Y.
\newblock Graph attention networks.
\newblock In \emph{International Conference on Learning Representations}, 2018.

\bibitem[Wang et~al.(2021)Wang, Jia, Cao, and Gong]{wang2021certified}
Wang, B., Jia, J., Cao, X., and Gong, N.~Z.
\newblock Certified robustness of graph neural networks against adversarial
  structural perturbation.
\newblock In \emph{Proceedings of the 27th ACM SIGKDD Conference on Knowledge
  Discovery \& Data Mining}, pp.\  1645--1653, 2021.

\bibitem[Wei(2020)]{wei2020tutorials}
Wei, W.
\newblock Tutorials on advanced optimization methods.
\newblock \emph{arXiv preprint arXiv:2007.13545}, 2020.

\bibitem[Zhang et~al.(2020)Zhang, Yin, Chen, Hung, Huang, and
  Cui]{zhang2020gcn}
Zhang, S., Yin, H., Chen, T., Hung, Q. V.~N., Huang, Z., and Cui, L.
\newblock Gcn-based user representation learning for unifying robust
  recommendation and fraudster detection.
\newblock In \emph{Proceedings of the 43rd international ACM SIGIR conference
  on research and development in information retrieval}, pp.\  689--698, 2020.

\bibitem[Zhang \& Zitnik(2020)Zhang and Zitnik]{zhang2020gnnguard}
Zhang, X. and Zitnik, M.
\newblock Gnnguard: Defending graph neural networks against adversarial
  attacks.
\newblock \emph{Advances in neural information processing systems},
  33:\penalty0 9263--9275, 2020.

\bibitem[Zhang et~al.(2019)Zhang, Khan, and Coates]{zhang2019comparing}
Zhang, Y., Khan, S., and Coates, M.
\newblock Comparing and detecting adversarial attacks for graph deep learning.
\newblock In \emph{Proc. Representation Learning on Graphs and Manifolds
  Workshop, Int. Conf. Learning Representations, New Orleans, LA, USA}, 2019.

\bibitem[Zhu et~al.(2019)Zhu, Zhang, Cui, and Zhu]{zhu2019robust}
Zhu, D., Zhang, Z., Cui, P., and Zhu, W.
\newblock Robust graph convolutional networks against adversarial attacks.
\newblock In \emph{Proceedings of the 25th ACM SIGKDD international conference
  on knowledge discovery \& data mining}, pp.\  1399--1407, 2019.

\bibitem[Z{\"u}gner et~al.(2018)Z{\"u}gner, Akbarnejad, and
  G{\"u}nnemann]{zugner2018netattack}
Z{\"u}gner, D., Akbarnejad, A., and G{\"u}nnemann, S.
\newblock Adversarial attacks on neural networks for graph data.
\newblock In \emph{Proceedings of the 24th ACM SIGKDD international conference
  on knowledge discovery \& data mining}, pp.\  2847--2856, 2018.

\bibitem[Zügner \& Günnemann(2019)Zügner and
  Günnemann]{zugner2018metattack}
Zügner, D. and Günnemann, S.
\newblock Adversarial attacks on graph neural networks via meta learning.
\newblock In \emph{International Conference on Learning Representations}, 2019.
\newblock URL \url{https://openreview.net/forum?id=Bylnx209YX}.

\end{thebibliography}

\newpage

\appendix
\onecolumn

\section{Theorectical Proofs}
\label{Sec:Appendix_A}

\setcounter{lemma}{0}

\begin{lemma} (Restate) Let $A$ be the adjacency matrix of the perturbed graph with $\rho$ injected nodes, and the injected nodes are in the last $\rho$ rows and columns. With smoothing $p_n>0$ and $p_e>0$, we have the upper bound of $p(E_v)$: 
    \begin{align}
        &p(E_v)\leq \overline{p(E_v)}\\
        =&1- p_1^{||A_{n:(n+\rho),v}||_1}p_2^{||A_{n:(n+\rho),v}^2||_1}\cdots p_k^{||A_{n:(n+\rho),v}^k||_1},\nonumber
    \end{align}
where $p_i:=1-(\bar{p}_e\bar{p}_n)^{i},\, \forall i\in\{1,2,\cdots,k\}$, and adjacency matrix $A$ 
contains the injected nodes encoded in the $(n+1)^{th}$ to $(n+\rho)^{th}$ row, and $||\cdot||_1$ is $l_1$ norm. 
\end{lemma}

\begin{proof}
According to \cite{scholten2022randomized}, we have an upper bound for $p(E_v)\leq \overline{p(E_v)}$ by assuming the independence among the paths. 
Let $p(\bar{E}_v^{\tilde{v}})$ denote the probability that all paths are intercepted from an injected node $\tilde{v}$ to node $v$ in the case that of considering each path independently. We have $p(\bar{E}_v^{\tilde{v}})=\prod_{q\in P_{\tilde{v}v}^k}(1-(\bar{p}_e\bar{p}_n)^{|q|})$, where $\bar{p}_e:=1-p_e$, $\bar{p}_n:=1-p_n$ and $|q|\in \{1,\cdots,k\}$ represent the length of the path $q\in P_{\tilde{v}v}^k$ from $\tilde{v}$ to $v$. $(\bar{p}_e\bar{p}_n)^{|q|}$ is the probability that all edges and all nodes in the path $q$ are not deleted, $1-(\bar{p}_e\bar{p}_n)^{|q|}$ is the probability that at least one of edges or one of nodes are deleted, such that the path $q$ is intercepted. 
Then, by considering multiple injected nodes, we have $\overline{p(E_v)} = 1-\prod_{\tilde{v}\in\tilde{\mathcal{V}}}p(\bar{E}_v^{\tilde{v}})$. Finally, we have the $\overline{p(E_v)}$ as follows:
    \begin{align}
    &\quad\,\, \overline{p(E_v)} \\
    &= 1-\prod_{\tilde{v}\in\tilde{\mathcal{V}}}p(\bar{E}_v^{\tilde{v}})\nonumber\\
    &= 1-\prod_{\tilde{v}\in\tilde{\mathcal{V}}} \{\prod_{q\in P_{\tilde{v}v}^k}(1-(\bar{p}_e\bar{p}_n)^{|q|}) \}\nonumber\\
    &= 1-\prod_{\tilde{v}\in\tilde{\mathcal{V}}} \{(1-\bar{p}_e\bar{p}_n)^{A_{\tilde{v}v}}(1-(\bar{p}_e\bar{p}_n)^2)^{A_{\tilde{v}v}^2}\cdots(1-(\bar{p}_e\bar{p}_n)^{k})^{A_{\tilde{v}v}^k} \}\nonumber\\
    &= 1- p_1^{||A_{n:(n+\rho),v}||_1}p_2^{||A_{n:(n+\rho),v}^2||_1}\cdots p_k^{||A_{n:(n+\rho),v}^k||_1},\nonumber
\end{align}
where $p_i:=1-(\bar{p}_e\bar{p}_n)^{i}$. In particular, the constant $p_k$ denotes the probability that a path with a length of $k$ is intercepted. According to graph theory, $A_{\tilde{v}v}^k$ is the number of paths from node $\tilde{v}$ to node $v$ with distance/length/steps of exactly $k$ in the graph. Let $A_{n:(n+\rho),v}$ denote the slicing of matrix $A$, taking the $v^{th}$ column and the rows from $(n+1)^{th}$ to $(n+\rho)^{th}$. Then $||A_{n:(n+\rho),v}^k||_1$ quantifies the number of paths with a length of $k$ originating from any malicious node and reaching node $v$. 
\end{proof}

\setcounter{thm}{0}
\begin{thm}
(Restate) Given a base GNN classifier $f$ trained on a graph $G$ and its smoothed classifier $g$ defined in \eqref{eqn:smooth_g}, a testing node $v \in G$ and a perturbation range $B_{\rho,\tau}(G)$, let $E_v$ be the event defined in Eq.~\eqref{eqn:E_v}. The absolute change in predicted probability $|p_{v,y}(G)-p_{v,y}(G')|$ for all perturbed graphs $G' \in B_{\rho,\tau}(G)$ is bounded by the probability of the event $E_v$: $|p_{v,y}(G)-p_{v,y}(G')|\leq p(E_v)$.
\end{thm}
\begin{proof}
    By the law of total probability, we have 
\begin{align}
    &\quad \,\, \mathbb{P}(f_v(\phi(G'))=y)\nonumber\\
    &=\mathbb{P}(f_v(\phi(G'))=y\land E_v)+\mathbb{P}(f_v(\phi(G'))=y\land \bar{E_v}).\nonumber
\end{align}

Note that, we define the event $E_v$ based on the sampling of perturbed graph $\phi(G')$. However, the clean graph $G$ is smaller than $G'$, and the intersection/overlap graph of them is $G\cap G'=G$. Subtly, we can still use the event $E_v$ defined on $\phi(G')$ to divide the sample space of $\phi(G)$ by regarding the model $f_{v}(\phi(G))$ only take part of the $\phi(G')$ as input, which is the intersected part of $G$: $\phi(G')\cap G$, and the result does not relate to the part that beyond $G$ (i.e., the injected nodes). Such that, we also have 
\begin{align}
    &\quad\,\,\mathbb{P}(f_v(\phi(G))=y) \nonumber\\
    &=\mathbb{P}(f_v(\phi(G))=y\land E_v)+\mathbb{P}(f_v(\phi(G))=y\land \bar{E_v}).\nonumber
\end{align}


Due to the fact that the injected node does not have any message passing to $v$ would not affect the $p_{v,y}(G)$, we have $\mathbb{P}(f_v(\phi(G'))=y|\bar{E_v})=\mathbb{P}(f_v(\phi(G))=y|\bar{E_v})$, so that $\mathbb{P}(f_v(\phi(G))=y\land \bar{E_v})=\mathbb{P}(f_v(\phi(G'))=y\land \bar{E_v})$. Following \cite{scholten2022randomized}, we have similar deduction as follows:
\begin{align}
    &\quad\,\, p_{v,y}(G)-p_{v,y}(G')\nonumber\\
    &=\mathbb{P}(f_v(\phi(G))=y\land E_v)+\mathbb{P}(f_v(\phi(G))=y\land \bar{E_v})\nonumber\\
    &\quad-\mathbb{P}(f_v(\phi(G'))=y\land E_v)-\mathbb{P}(f_v(\phi(G'))=y\land \bar{E_v}) \nonumber\\ 
    &=\mathbb{P}(f_v(\phi(G))=y\land E_v)-\mathbb{P}(f_v(\phi(G'))=y\land E_v)\nonumber\\  
    &\leq \mathbb{P}(f_v(\phi(G))=y\land E_v)\nonumber\\
    &=p(E_v)\cdot \mathbb{P}(f_v(\phi(G))=y|E_v)\nonumber\\
    &\leq p(E_v).  \nonumber
\end{align}
\end{proof}


\section{Details of Optimization Formulation}
\label{Sec:Appendix_C}

\subsection{Formulating problem~\eqref{opt:collective} as polynomial constrained programming.}
For problem~\eqref{opt:collective}, we plug in $\overline{p(E_v)}$ with \eqref{Eqn:P_e_General}, and then we have the following optimization problem:
\begin{align}
\label{opt:collective2}
\max_{A_{n:,:},\,\mathbf{m}}\quad & M=\sum_{v\in \mathbb{T}} m_v,\\ 
\text{s.t.} \quad 
&2\,\overline{p(E_v)} \geq c_v\cdot m_v,\, \forall v \in \mathbb{T}, \nonumber\\
&\overline{p(E_v)}=1- (p_1^{||A_{n:(n+\rho),v}||_1}p_2^{||A_{n:(n+\rho),v}^2||_1}\cdots p_k^{||A_{n:(n+\rho),v}^k||_1}),\nonumber\\
&||A_{\tilde{v}:}||_1 \leq \tau, \,\forall \tilde{v}\in \{n+1,\cdots,n+\rho\},\nonumber\\
& A_{ij}\in\{0,1\}, \,\forall i\in \{n+1,\cdots,n+\rho\},\,\forall j \in \{1,\cdots,n+\rho\},\nonumber\\
& m_v\in\{0,1\}, \forall \, v \in \{1,\cdots, n\},\nonumber
\end{align}
where $m_v=1$ (the element in vector $\mathbf{m}$) indicates that the robustness for node $v$ can not be verified. Specifically, it means that $2\,\overline{p(E_v)} \geq c_v$, and it disobeys our certifying condition.  

There are exponential terms in $\overline{p(E_v)}$, which is difficult to solve by existing optimization tools. We further formalize the problem. By taking the logarithm of the $\overline{p(E_v)}$, we are able to transform the exponential constraint in problem~\eqref{opt:collective2} into polynomial constraint:
\begin{align}
    &\tilde{P}_v \leq log(1-\frac{c_v}{2})\cdot m_v,\\
    &\tilde{P}_v = ||A_{n:(n+\rho),v}||_1\cdot \tilde{p_1}+ ||A_{n:(n+\rho),v}^2||_1\cdot \tilde{p_2}+\cdots+ ||A_{n:(n+\rho),v}^k||_1\cdot \tilde{p_k},\nonumber
\end{align}
where $\tilde{p}_k=log(p_k)$ is a constant, and $\tilde{P}_v$ is equivalent to $log(1-\overline{p(E_v)})$. Then the problem~\eqref{opt:collective2} is transformed to a binary polynomial constrained programming:
\begin{align}
\label{opt:collective-polynomial}
\max_{A_{n:,:},\,\mathbf{m}}\quad & M=\sum_{v\in \mathbb{T}} m_v,\\ 
\text{s.t.} \quad 
&\tilde{P}_v \leq log(1-\frac{c_v}{2})\cdot m_v,\nonumber\\
&\tilde{P}_v = ||A_{n:(n+\rho),v}||_1\cdot \tilde{p_1}+ ||A_{n:(n+\rho),v}^2||_1\cdot \tilde{p_2}+ \cdots+ ||A_{n:(n+\rho),v}^k||_1\cdot \tilde{p_k},\nonumber\\
&||A_{\tilde{v}:}||_1 \leq \tau, \,\forall \tilde{v}\in \{n+1,\cdots,n+\rho\},\nonumber\\
& A_{ij}\in\{0,1\}, \,\forall i\in \{n+1,\cdots,n+\rho\},\,\forall j \in \{1,\cdots,n+\rho\},\nonumber\\
& A^{\top}=A,\nonumber\\
& m_v\in\{0,1\}, \forall \, v \in \{1,\cdots, n\}.\nonumber
\end{align}

\subsection{Formulating problem~\eqref{opt:collective-polynomial} as BQCLP~\eqref{opt:collective-BQP}.}

In this section, we discuss the process from \eqref{opt:collective-polynomial} to \eqref{opt:collective-BQP}. In the case of $k=2$, the problem \eqref{opt:collective-polynomial} becomes a binary quadratic constrained problem as follows:
\begin{align}
\label{opt:collective-polynomial_simplified_by_2}
\max_{A_{n:,:},\,\mathbf{m}}\quad & M=\sum_{v\in \mathbb{T}} m_v,\\ 
\text{s.t.} \quad 
&||A_{n:(n+\rho),v}||_1\cdot \tilde{p_1}+ ||A_{n:(n+\rho),v}^2||_1 \cdot \tilde{p_2} \leq log(1-\frac{c_v}{2})\cdot m_v,\nonumber\\
&||A_{\tilde{v}:}||_1 \leq \tau, \,\forall \tilde{v}\in \{n+1,\cdots,n+\rho\},\nonumber\\
& A_{ij}\in\{0,1\}, \,\forall i\in \{n+1,\cdots,n+\rho\},\,\forall j \in \{1,\cdots,n+\rho\},\nonumber\\
& A^{\top}=A,\nonumber\\
& m_v\in\{0,1\}, \forall \, v \in \{1,\cdots, n\}.\nonumber
\end{align}
\begin{figure}[htb!]
   \centering
   \includegraphics[width=0.19\textwidth,height=1.5cm]{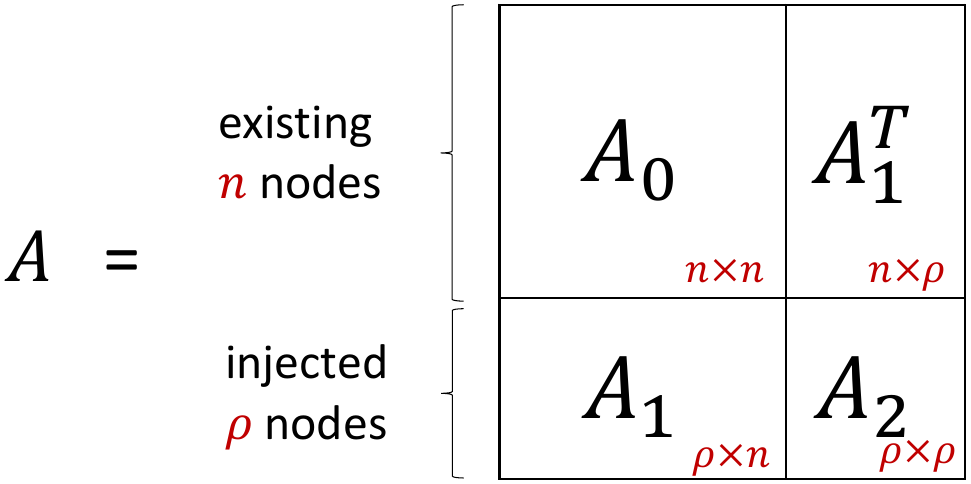}
   \caption{Illustration of adjacency matrix notation.}
   \label{fig:matrix}
\end{figure}
Next, we divide the adjacency matrix $A$ into four parts as shown in Fig.\ref{fig:matrix}, and then the $A^2$ can be interpreted as:
\begin{align}
    A^2 = 
    \begin{bmatrix}
        (A_0A_0+ A_1^\top A_1)_{(n\times n)} & (A_0A_1^\top+ A_1^\top A_2)_{(\rho\times n)}   \\
        ( A_1A_0+A_2A_1)_{(\rho\times n)} &  (A_1A_1^\top +A_2 A_2)_{(\rho\times\rho)}
    \end{bmatrix}.\nonumber
\end{align}
Then, the $l_1$ norm of $A_{n:(n+\rho),v}^2$ can be represented as:
\begin{align}
\label{eq:A_nnp_norm_1_blockbyblock}
     [||A_{n:(n+\rho),1}^2||_1,||A_{n:(n+\rho),2}^2||_1,\cdots,||A_{n:(n+\rho),n}^2||_1]^\top = (A_1A_0+A_2A_1)\mathbf{1}_{\rho}.
\end{align}
Also, same as above, together with Fig.\ref{fig:matrix}, $||A_{\tilde{v}:}||_1$ is described as:
\begin{align}
    \label{eq:A_v_norm_1_blockbyblock}
    [||A_{n:}||_1,||A_{(n+2):}||_1,\cdots,||A_{(n+\rho):}||_1]^\top = A_1\mathbf{1}_n + A_2\mathbf{1}_\rho.
\end{align}
Finally, combine \eqref{eq:A_nnp_norm_1_blockbyblock} and \eqref{eq:A_v_norm_1_blockbyblock}, problem~\eqref{opt:collective-polynomial_simplified_by_2} can be formulated as:
\begin{align}
    \max_{A_1,A_2,\mathbf{m}}&\quad  M=\mathbf{t}^\top \mathbf{m}, \nonumber\\ 
    \text{s.t.} \quad 
    & \tilde{p_1}A_1^\top\mathbf{1}_{\rho}  + \tilde{p_2}(A_1A_0+A_2A_1)^\top\mathbf{1}_{\rho}\leq \mathbf{C}\circ \mathbf{m},\nonumber\\
    &A_1\mathbf{1}_n +A_2 \mathbf{1}_{\rho} \leq \tau,\,A_2^{\top}=A_2, \nonumber\\
    & A_1\in\{0,1\}^{\rho\times n},\,A_2\in\{0,1\}^{\rho\times \rho},\,\mathbf{m}\in\{0,1\}^{n},\nonumber
\end{align}
where $\mathbf{t}$ is a constant zero-one vector that encodes the position of the target node set $\mathbb{T}$, $\mathbf{m}$ is a vector that indicates whether the nodes are successfully attacked, $\mathbf{C}\in \mathbb{R}^n$ is a vector with negative constant elements $log(1-\frac{c_v}{2})$, for $v=1,2,\cdots,n$.

\subsection{Formulating problem~\eqref{opt:collective-BQP} as Linear Programming Problem~\eqref{opt:collective-linear1}.}
Here, we discuss the details of the process of relaxing the BQCLP problem~\eqref{opt:collective-BQP} to the LP problem~\eqref{opt:collective-linear1}.
In problem~\eqref{opt:collective-BQP}, there are $\rho^2n$ quadratic terms among $A_2A_1$. To tackle the challenge, we introduce the following transformation to transform it into an LP problem. Specifically, we first substitute the quadratic terms with linear terms and relax all the binary variables to continuous variables in $[0,1]$.

If $x\in \mathbb{B}$, $y\in \mathbb{B}$ are two integer binary variables, then the quadratic term $xy$ can be substitute by a single variable $z:=xy$ with the combination of linear constraints~\cite{wei2020tutorials}: $z\leq x,\,z\leq y,\,x+y-z\leq 1,\,z\in \mathbb{B}$.
We use $a_{(ij)}$ and $b_{(ij)}$ to denotes the element in $i^{th}$ row and $j^{th}$ column of matrix $A_1$ and $A_2$ respectively. For each quadratic term $b_{(ij)}a_{(jv)}$ ($\forall i\in \{1,\cdots,\rho\},\forall j \in \{1,\cdots,\rho\},\forall v \in \{1,\cdots,n\}$) in $A_2A_1$, we create a substitution variable $Q_{v(ij)}:=b_{(ij)}a_{(jv)}$ with corresponding constraints: $Q_{v(ij)}\in \mathbb{B}$, $Q_{v(ij)}\leq b_{(ij)}$, $Q_{v(ij)}\leq a_{(jv)}$, and $b_{(ij)}+a_{(jv)}-Q_{v(ij)}\leq 1$. The existing linear terms remain unchanged. Now, the BQCLP problem has transformed into binary linear programming (BLP). 

Next, we formulate the problem using matrix representation. We firstly use $O$ to substitute $(A_2A_1)^\top \mathbf{1}_{\rho}$, and we have the first constraint as:
\begin{align}
    \label{constraint:cons1_substituted_with_O}
    \tilde{p_1}A_1^\top\mathbf{1}_{\rho} + \tilde{p_2} A_0^\top A_1^\top\mathbf{1}_{\rho} + \tilde{p_2}O &\leq \mathbf{C}\circ \mathbf{m}.\nonumber
\end{align}


We list the elements of the $A_1$ and $A_2$ as follows:
\begin{equation}
    \label{eq:definition of A1A2}
    A_1 =\begin{bmatrix}
        a_{11} & a_{12} & a_{13} & \cdots & a_{1n} \\
        a_{21} &  &  &  &  \\
        a_{31} &  & \ddots &  & \vdots \\
        \vdots &  &  &  &  \\
        a_{\rho1} &  & \cdots &  & a_{\rho n} 
        \end{bmatrix},  \,
    A_2 =\begin{bmatrix}
        b_{11} & b_{12} & b_{13} & \cdots & b_{1\rho} \\
        b_{21} &  &  &  &  \\
        b_{31} &  & \ddots &  & \vdots \\
        \vdots &  &  &  &  \\
        b_{\rho1} &  & \cdots &  & b_{\rho\rho} 
        \end{bmatrix}.\\ 
\end{equation}
Then, the matrix multiplication of $A_2$ and $A_1$ is
\begin{equation}
    A_2A_1 = 
    \begin{bmatrix}
    {b_{11}a_{11}+b_{12}a_{21}+\cdots+b_{1\rho}a_{\rho 1}} &{b_{11}a_{12}+b_{12}a_{22}+\cdots+b_{1\rho}a_{\rho 2}}& \cdots & {b_{11}a_{1n}+b_{12}a_{2n}+\cdots +b_{1\rho}a_{\rho n}}\\
    {b_{21}a_{11}+b_{22}a_{21}+\cdots+b_{2\rho}a_{\rho 1}} &{b_{21}a_{12}+b_{22}a_{22}+\cdots+b_{2\rho}a_{\rho 2}}& \cdots & {b_{21}a_{1n}+b_{22}a_{2n}+\cdots +b_{2\rho}a_{\rho n}}\\
    \vdots & \vdots & \ddots& \vdots \\
    {b_{\rho 1}a_{11}+b_{\rho 2}a_{21}+\cdots+b_{\rho \rho}a_{\rho 1}} &{b_{\rho 1}a_{12}+b_{\rho 2}a_{22}+\cdots+b_{\rho \rho}a_{\rho 2}}& \cdots & {b_{\rho 1}a_{1n}+b_{\rho 2}a_{2n}+\cdots+b_{\rho \rho}a_{\rho n}}
    \end{bmatrix}.
    \nonumber
\end{equation}

By the definition of matrix $Q_v$, for $v\in\{1,2,\cdots,n\}$, we have the following equivalent representation:
\begin{align}
    \label{eq:Q_v_A2_A1}
    Q_v = 
    \begin{bmatrix}
    Q_{v(11)} & Q_{v(12)} &\cdots & Q_{v(1 \rho)} \\
    Q_{v(21)} & Q_{v(22)} &  & Q_{v(2 \rho)} \\
    \vdots & \vdots & \ddots &  \vdots\\
    Q_{v(\rho 1)} & Q_{v(\rho 2)} & \cdots & Q_{v(\rho \rho)} 
    \end{bmatrix} &:=
    \begin{bmatrix}
    b_{11}a_{1v} & b_{21}a_{1v} &\cdots & b_{\rho 1}a_{1v} \\
    b_{12}a_{2v} & b_{22}a_{2v} &  & b_{\rho 2}a_{2v} \\
    \vdots & \vdots & \ddots &  \vdots\\
    b_{1\rho}a_{\rho v} & b_{2\rho}a_{\rho v} & \cdots & b_{\rho \rho}a_{\rho v}  
    \end{bmatrix}. \nonumber
\end{align}

We notice that $(A_2A_1)^\top \mathbf{1}_{\rho}$ is to sum the $A_2A_1$ by its column, and each $Q_v$ contains all the terms for each vector summation. Then we have
$O = (A_2A_1)^\top = [\mathbf{1}_{\rho}^\top Q_1\mathbf{1}_{\rho},\mathbf{1}_{\rho}^\top Q_2\mathbf{1}_{\rho},\cdots, \mathbf{1}_{\rho}^\top Q_n\mathbf{1}_{\rho}]^\top$.

Further, by decomposing the meaning of $Q_v$, we have 
\begin{align}
    Q_v := 
    \begin{bmatrix}
    b_{11}& b_{21} &\cdots & b_{\rho 1} \\
    b_{12} & b_{22} &\cdots  & b_{\rho 2} \\
    \vdots & \vdots & \ddots &  \vdots\\
    b_{1\rho} & b_{2\rho} & \cdots & b_{\rho \rho}
    \end{bmatrix}
    \circ
    \begin{bmatrix}
    a_{1v} & a_{1v} &\cdots  & a_{1v} \\
    a_{2v} & a_{2v} &\cdots  & a_{2v} \\
    \vdots & \vdots & \ddots &  \vdots\\
    a_{\rho v} & a_{\rho v} & \cdots & a_{\rho v}
    \end{bmatrix} \nonumber
    = A_2 \circ 
    \mathbf{1}_{\rho} \begin{bmatrix}
        a_{1v}\\
        a_{2v}\\
        \vdots\\
        a_{\rho v}
    \end{bmatrix}^\top \nonumber
    = A_2
    \circ
    \mathbf{1}_{\rho} {[A_{1(:,v)}]}^\top.
\end{align}
To make the $Q_v$ equivalent to the quadratic terms, for every $Q_v$, we need to add its constraints:  
\begin{align}
    Q_v \leq A_2, \, Q_v &\leq \mathbf{1}_{\rho} {[A_{1(:,v)}]}^\top,\,\mathbf{1}_{\rho}{[A_{1(:,v)}]}^\top+A_2-Q_v \leq 1.\nonumber
\end{align}
Finally, we relaxed $A_1$, $A_2$, $Q_v$ to relax all the binary variables to continuous variables in $[0, 1]$:
\begin{align}
    Q_v\in [0,1]^{\rho\times \rho},\, A_1\in[0,1]^{\rho\times n},\,A_2\in[0,1]^{\rho\times \rho},\,\mathbf{m}\in[0,1]^{n}.\nonumber
\end{align}

Then we have the linear programming problem \eqref{opt:collective-linear1} as follows:
\begin{align}
\max_{A_1,A_2,m,\atop Q_1,Q_2,\cdots,Q_n}\quad & M=\mathbf{t}^\top \mathbf{m},\nonumber \\ 
\text{s.t.} \quad 
& \tilde{p_1}A_1^\top\mathbf{1}_{\rho} + \tilde{p_2} A_0^\top A_1^\top\mathbf{1}_{\rho} + \tilde{p_2}O \leq \mathbf{C}\circ \mathbf{m}\nonumber\\
&A_1\mathbf{1}_n +A_2 \mathbf{1}_{\rho} \leq \tau, \nonumber\\
&Q_v=(Q_{v(ij)})_{\rho \times \rho},\, v\in\{1,2,\cdots,n\}, \nonumber\\
&O=[\mathbf{1}_{\rho}^\top Q_1\mathbf{1}_{\rho},\mathbf{1}_{\rho}^\top Q_2\mathbf{1}_{\rho},\cdots, \mathbf{1}_{\rho}^\top Q_n\mathbf{1}_{\rho}]^\top, \nonumber\\
&Q_v \leq \mathbf{1}_{\rho} {[A_{1(:,v)}]}^\top,\nonumber\\ 
&Q_v\leq A_2,\nonumber\\
&\mathbf{1}_{\rho} {[A_{1(:,v)}]}^\top+A_2-Q_v \leq 1,\nonumber\\
&Q_v\in [0,1]^{\rho\times \rho},\nonumber\\
& A_1\in[0,1]^{\rho\times n},\nonumber\\
& A_2\in[0,1]^{\rho\times \rho},\nonumber\\
& A_2^{\top}=A_2,\nonumber\\
& \mathbf{m}\in[0,1]^{n}.\nonumber
\end{align}

\subsection{Formulating problem~\eqref{opt:collective-BQP} as Linear Programming Problem~\eqref{opt:collective-linear2}.}
\label{}
We start from \eqref{opt:collective-BQP}, and we have the first constraint:
\begin{equation}
    \tilde{p_1}A_1^\top\mathbf{1}_{\rho}  + \tilde{p_2}A_0^\top A_1^\top \mathbf{1}_{\rho}+\tilde{p_2}A_1^\top A_2^\top\mathbf{1}_{\rho}\leq \mathbf{C}\circ  \mathbf{m}.\nonumber\\
\end{equation}
Then, we substitute $A_2^\top\mathbf{1}_{\rho}$ with $\mathbf{z}$,
\begin{align}
\label{eq:definition of Z}
    \mathbf{z}:=A_2^\top\mathbf{1}_{\rho}
    =\begin{bmatrix}
        b_{11} & b_{12} & b_{13} & \cdots & b_{1\rho} \\
        b_{21} &  &  &  &  \\
        b_{31} &  & \ddots &  & \vdots \\
        \vdots &  &  &  &  \\
        b_{\rho1} &  & \cdots &  & b_{\rho\rho} 
        \end{bmatrix}_{(\rho,\rho)}
        \begin{bmatrix}
            1\\
            1\\
            1\\
            \vdots\\
            1
        \end{bmatrix}_{(\rho,1)}=
        \begin{bmatrix}
            b_{11}+b_{12}+b_{13}+\cdots+b_{1\rho}\\
            b_{21}+b_{22}+b_{23}+\cdots+b_{2\rho}\\
            \vdots\\
            b_{\rho 1}+b_{\rho 2}+b_{\rho 3}+\cdots+b_{\rho \rho}
        \end{bmatrix}_{(\rho, 1)}. 
\end{align}
Then, from \eqref{eq:definition of Z}, the constraint is transformed into 
\begin{align}
\label{constraint:trasformed constraint with Z}
        \tilde{p_1}A_1^\top\mathbf{1}_{\rho}  + \tilde{p_2}A_0^\top A_1^\top \mathbf{1}_{\rho}+\tilde{p_2}A_1^\top \mathbf{z} \leq \mathbf{C}\circ \mathbf{m},  \\
        z_i \in \{0,1,2,\cdots,min(\tau,\rho)\}\ \ \ \forall i \in \{0,1,2,\cdots,\rho\}. \nonumber
\end{align}

In \eqref{opt:collective-BQP}, since there exists the constraint: $A_1\mathbf{1}_n +A_2 \mathbf{1}_{\rho} \leq \tau$, so we have $z_i$ satisfies $z_i \in \{0,1,2,\cdots,min(\tau,\rho)\}$. Next, we deal with the quadratic term $A_1^\top \mathbf{z}$. 

If $x\in \mathbb{B}$ is a binary variable, and $z\in [0,u]$ is a continuous variable, then the quadratic term $xy$ can be substitute by a single variable $z:=xy$ with the combination of linear constraints~\cite{wei2020tutorials}: $w\leq ux, w\leq z, ux+z-w\leq u, 0\leq w$. To apply it, we first relax the $\mathbf{z}$ to $[0,min(\tau,\rho)]$. 

We know that $A_1^\top \mathbf{z}$ satisfies that
\begin{align}
    A_1^\top \mathbf{z} = \begin{bmatrix}
        a_{11} & a_{21} & a_{31} & \cdots & a_{\rho 1} \\
        a_{12} & a_{22} & a_{32} & \cdots & a_{\rho 2} \\
        a_{13} & a_{23} & a_{33} & \cdots & \vdots \\
        \vdots & \vdots & \vdots  & \ddots & \vdots \\
        a_{1n} & a_{2n} & a_{3n} & \cdots & a_{\rho n} 
        \end{bmatrix}
        \begin{bmatrix}
            z_1\\
            z_2\\
            z_3\\
            \vdots\\
            z_\rho
        \end{bmatrix} =
       \begin{bmatrix}
           a_{11} z_1+a_{21} z_2 + \cdots + a_{ \rho 1}z_\rho \\
           a_{12} z_1+a_{22} z_2 + \cdots + a_{ \rho 2}z_\rho\\
           \vdots\\
           a_{1n} z_1+a_{2n} z_2 + \cdots + a_{ \rho n}z_\rho
       \end{bmatrix}_{(n,1)}. \nonumber
\end{align}
Then, we create a new variable matrix $Q$ to substitute $A_1^\top \mathbf{z}$, with each of its element: $q_{ij} := a_{ji}z_{i}, (\forall i \in \{1,2,\cdots,n\},j \in \{1,2,\cdots,\rho\})  $. That is:
\begin{align}
    Q = 
    \begin{bmatrix}
        q_{11} & q_{12} & \cdots & q_{1 \rho}\\
        q_{21} & q_{22} & \cdots & q_{2 \rho}\\
        \vdots & \vdots & \ddots & \vdots\\
        q_{n1} & q_{n2} & \cdots & q_{n \rho}\\
    \end{bmatrix}
    =
       \begin{bmatrix}
           a_{11} z_1 & a_{21} z_2 & \cdots & a_{ \rho 1}z_\rho \\
           a_{12} z_1 & a_{22} z_2 & \cdots & a_{ \rho 2}z_\rho\\
           \vdots& \vdots & \ddots & \vdots\\
           a_{1n} z_1 & a_{2n} z_2 & \cdots & a_{ \rho n}z_\rho
       \end{bmatrix}. \nonumber
\end{align}
We now have $A_1^\top \mathbf{z} = Q\mathbf{1}_\rho$. 
Assuming that $\tau \leq \rho$, for each quadratic term $A^\top_{1(ij)}z_j$ ($\forall i\in \{1,\cdots,n\},\forall j \in \{1,\cdots,\rho\}$) in $A_1^\top \mathbf{z}$, we create a substitution variable $Q_{(ij)}=A^\top_{1(ij)}z_j$ with corresponding constraints: $0 \leq Q_{(ij)}$, $Q_{(ij)}\leq \tau A^\top_{1(ij)}$, $Q_{(ij)}\leq z_j$, and $\tau A^\top_{1(ij)} + z_j -Q_{(ij)} \leq \tau$. 
Further, with matrix notation, we have 
\begin{align}
\label{cons:Q constraint}
    0 \leq \ & Q\leq \tau A_1^\top,  \nonumber \\
    0 \leq 1_{n}\mathbf{z}^\top &- Q \leq \tau(1-A_1^\top), \\
    A_1 \in \{0,1\}, \mathbf{z} &\in [0,\tau] , Q\in [0,\tau]. \nonumber
\end{align}

Finally, we relax all the binary variables to be continuous variables, We have problem~\eqref{opt:collective-linear2} as follows:
\begin{align}
\max_{A_1,m,\atop Q \in \mathbb{R}^{n \times \rho}}\quad & M=\mathbf{t}^\top \mathbf{m},\\ 
\text{s.t.} \quad 
& \tilde{p_1}A_1^\top\mathbf{1}_{\rho} + \tilde{p_2} A_0^\top A_1^\top\mathbf{1}_{\rho} + \tilde{p_2} Q \mathbf{1}_{\rho} \leq \mathbf{C}\circ \mathbf{m},\nonumber\\
&A_1\mathbf{1}_n + \mathbf{z} \leq \tau, \nonumber\\
&Q \leq \tau A_{1}^\top,\nonumber\\ 
&Q\leq \mathbf{1}_n \mathbf{z}^\top,\nonumber\\
&\tau A^\top_{1}+\mathbf{1}_n \mathbf{z}^\top-Q \leq \tau,\nonumber\\
&Q \in[0,\tau]^{n\times \rho},\nonumber\\
& A_1\in[0,1]^{\rho\times n},\nonumber\\
& \mathbf{z}\in  [0,\tau]^{\rho\times 1},\nonumber\\
& \mathbf{m}\in[0,1]^{n}.\nonumber
\end{align}

\section{Algorithm of our proposed methods}
\paragraph{Train a base classifier $f$.}Following the work of \cite{lai2023node}, our first step is to train a graph model to serve as the base classifier. To enhance the model's generalization ability on the smoothing samples, we incorporate random noise augmentation during the training process. The training procedure is summarized in Algorithm \ref{alg:training_noise}, providing an overview of the steps involved. Given a clean graph $G$, a smoothing distribution $\phi(G)$ with smoothing parameters $p_e$ and $p_n$, and the number of training epochs $E$, the algorithm iteratively trains the model on randomly generated graphs. In each epoch, a random graph $G_e$ is drawn from the smoothing distribution $\phi(G)$. The model is then trained on the training nodes using this randomly generated graph. This process is repeated for the specified number of training epochs.

\begin{algorithm} 
\caption{Graph model training~\cite{lai2023node}.}  
\label{alg:training_noise}  
\begin{algorithmic}[1]   
\REQUIRE Clean graph $G$, smoothing distribution $\phi(G)$ with smoothing parameters $p_e$ and $p_n$, training epoch $E$.
\FOR{$e=1,\cdots,E$}
\STATE{Draw a random graph $G_e\sim \phi(G)$. }
\STATE{$f=train\_model(f(G_e))$ on training nodes.}
\ENDFOR
\STATE \textbf{return} A base classifier $f(\cdot)$.
\end{algorithmic}  
\end{algorithm}

\paragraph{Obtaining prediction probability of smoothed classifier $g$.}Next, we need to obtain the prediction results of a smoothed classifier. 
As depicted in Algorithm \ref{alg:monte_carlo}, we sample $N$ graphs $G_1, G_2, \ldots, G_N$ from the smoothed distribution $\phi(G) = (\phi_e(G), \phi_n(G))$ based on the base classifier $f$. 
To estimate the probabilistic prediction, we employ a Monte Carlo process. For each sampled graph $G_i$, we calculate the prediction probability $p_{v,y}(G)$, which represents the frequency of the predicted class $y$ for the vertex $v$. This can be approximated as $p_{v,y}(G) \approx \sum_{i=1}^{N} \mathbb{I}(f_v(G_i) = y) / N$, where $\mathbb{I}$ is the indicator function.

Let denote the top class probability $p_A:=p_{v,y^*}(G)$ and runner-up class probability $p_B:=max_{y\neq y^*}p_{v,y}(G)$, we want to bound the impact of randomness. Specifically, we compute the lower bound of $p_A$ (denoted as $\underline{p_A}$) and upper bound of $p_B$ (denoted as $\overline{p_B}$). Applying the Clopper-Pearson Bernoulli confidence interval, we obtain the $\underline{p_A}$ and the $\overline{p_B}$ under a confidence level of $\alpha/C$, where $C$ represents the number of classes in the model. 

\begin{algorithm}  
\caption{Monte Carlo sampling~\cite{lai2023node}.}  
\label{alg:monte_carlo}
\begin{algorithmic}[1]  
\REQUIRE Clean graph $G$, smoothing distribution $\phi(G)$ with smoothing parameters $p_e$ and $p_n$, trained base classifier $f(\cdot)$, sample number $N$, confidence level $\alpha$.
\STATE{Draw $N$ random graphs $\{G_i|\sim G_i \sim \phi(G)\}_{i=1}^N$.}
\STATE{$counts=|\{i: f(G_i)=y\}|$, for $y=1, \cdots, C$.}
\STATE{$y_A,y_B=$ top two indices in $counts$.}
\STATE{$n_A,n_B=counts[y_A],counts[y_B]$.}
\STATE{$\underline{p_A},\overline{p_B}=\text{CP\_Bernolli}(n_A,n_B,N,\alpha)$.}
\STATE{ \textbf{return} $\underline{p_A}$, $\overline{p_B}$.}
\end{algorithmic}  
\end{algorithm}

\paragraph{Collective certification via solving an optimization problem.} We obtain the collective certified robustness by solving the optimization problem problem~\eqref{opt:collective-linear1} or \eqref{opt:collective-linear2}. The process is described in Algorithm \ref{alg:optimiaztion_problem}. 

In this algorithm, we first set up the constant $\tilde{p}_1$ and $\tilde{p}_2$ based on the given smoothing parameters $p_e$ and $p_n$. Next, for each node $v$ in the target node set $\mathbb{T}$, we obtain the lower bound $\underline{p_A}$ and the upper bound $\overline{p_B}$ using Algorithm \ref{alg:monte_carlo}. These bounds are based on the prediction probabilities of the smoothed classifier for the current node $v$. We then compute the value $c_v = \underline{p_A} - \overline{p_B}$ and prepare the constant vector $\mathbf{C}$ with elements $\log(1 - \frac{c_v}{2})$ for each node $v$. The objective function of the optimization problem is based on either $\eqref{opt:collective-linear1}$ or $\eqref{opt:collective-linear2}$, depending on the chosen formulation. The constraints are also set up accordingly. Finally, we solve the linear programming using an LP solver, such as MOSEK, to obtain the optimal value $M^*$. The certified ratio, which represents the percentage of nodes in the target set $\mathbb{T}$ that have been successfully certified, is then computed as $(|\mathbb{T}| - M^*)/|\mathbb{T}|$.

\begin{algorithm}
\caption{Certified robustness via solving optimization problem~\eqref{opt:collective-linear1} or \eqref{opt:collective-linear2}.}  
\label{alg:optimiaztion_problem}
\begin{algorithmic}[1]  
\REQUIRE Smoothing parameters $p_e$ and $p_n$, graph adjacent matrix $A_0$, perturbation budget $\rho$ and $\tau$, target node set $\mathbb{T}$.
\STATE{Set constant $\tilde{p}_1=log(1-(\bar{p}_e\bar{p}_n))$.}
\STATE{Set constant  $\tilde{p}_2=log(1-(\bar{p}_e\bar{p}_n)^2)$.}
\FOR{$v$ in $\mathbb{T}$}
\STATE{Obtain $\underline{p_A}$, $\overline{p_B}$ from Algorithm.~\ref{alg:monte_carlo} for current node $v$.}
\STATE{Compute $c_v=\underline{p_A}-\overline{p_B}$.}
\STATE{Prepare constant vector $\mathbf{C}$ with each element: $log(1-\frac{c_v}{2})$.}
\ENDFOR
\STATE{Setup objective function in \eqref{opt:collective-linear1} or \eqref{opt:collective-linear2}.}
\STATE{Setup constraints in \eqref{opt:collective-linear1} or \eqref{opt:collective-linear2}.}
\STATE{Solve the optimization problem using  LP solver such as MOSEK to get $M^*$.}
\STATE{\textbf{Return} Certified ratio $(|\mathbb{T}|-M^*)/|\mathbb{T}|$.}
\end{algorithmic} 
\end{algorithm}

\section{Other Experimental Results}
\label{Sec:Appendix_D}
\subsection{Trade off between Clean accuracy and the certified ratio on GCN model}
\label{Sec:Appendix_GCN}

In this section, we present the remaining experiments as outlined in Section. \ref{sec:experiments}. A superior certifying method should not only achieve a higher certified ratio but also maintain or improve the clear accuracy, which represents the original model's performance. We compare the results of these two metrics for our method under different parameter settings as shown in Figure. \ref{fig:clean_certified_ratio}. In the figures, the data points situated closer to the upper right side represent higher certified ratios and clean accuracy. It is evident that both of our proposed methods consistently outperform the sample-wise method, demonstrating their superior performance under various attacker power $\rho$.
\begin{figure}[hbt!]
\centering
    \subfigure[]{
    \includegraphics[width=0.19\textwidth,height=2.8cm]{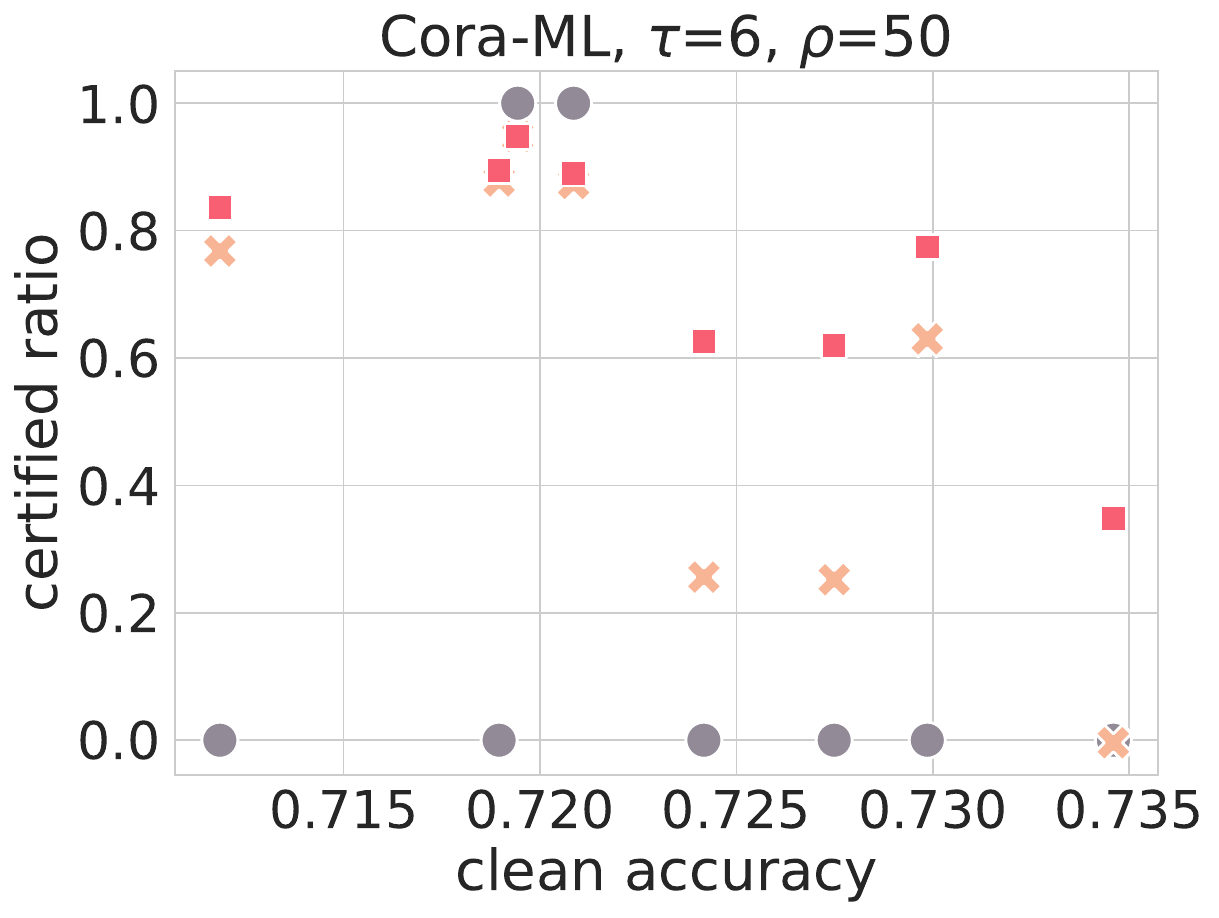}
    }
    \subfigure[]{
    \includegraphics[width=0.255\textwidth,height=2.8cm]{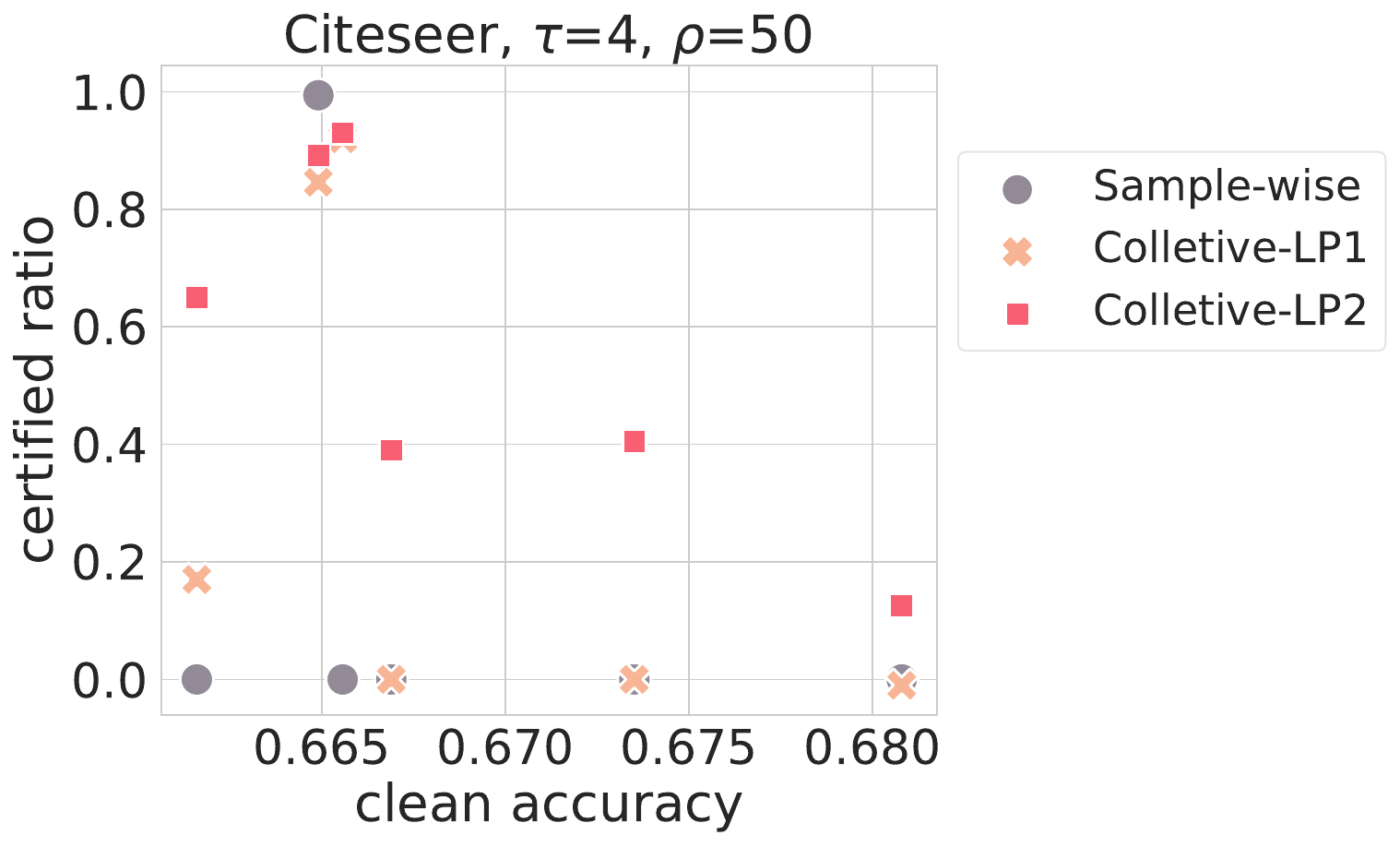}
    }
    \subfigure[]{
    \includegraphics[width=0.19\textwidth,height=2.8cm]{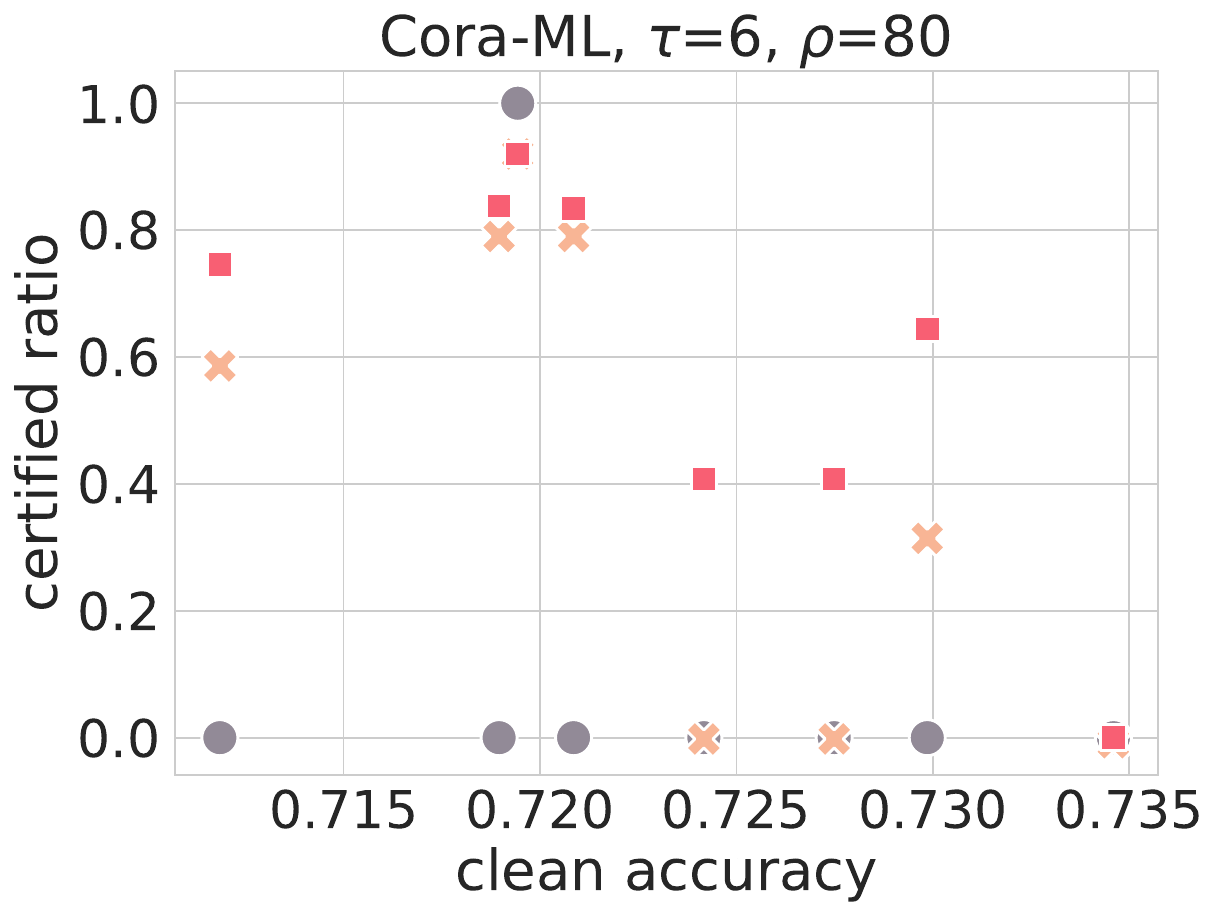}
    }
    \subfigure[]{
    \includegraphics[width=0.255\textwidth,height=2.8cm]{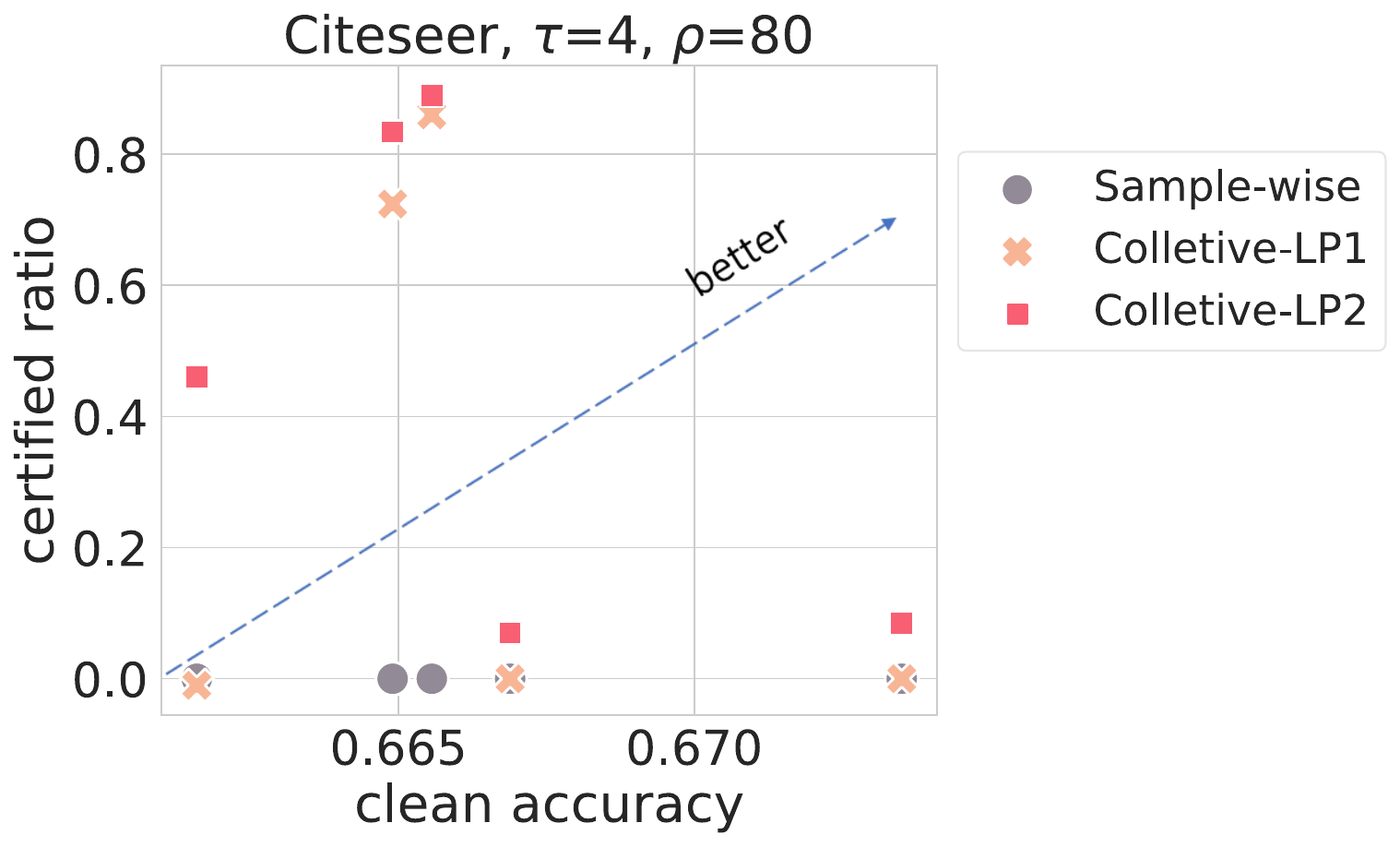}
    }
    \subfigure[]{
    \includegraphics[width=0.19\textwidth,height=2.8cm]{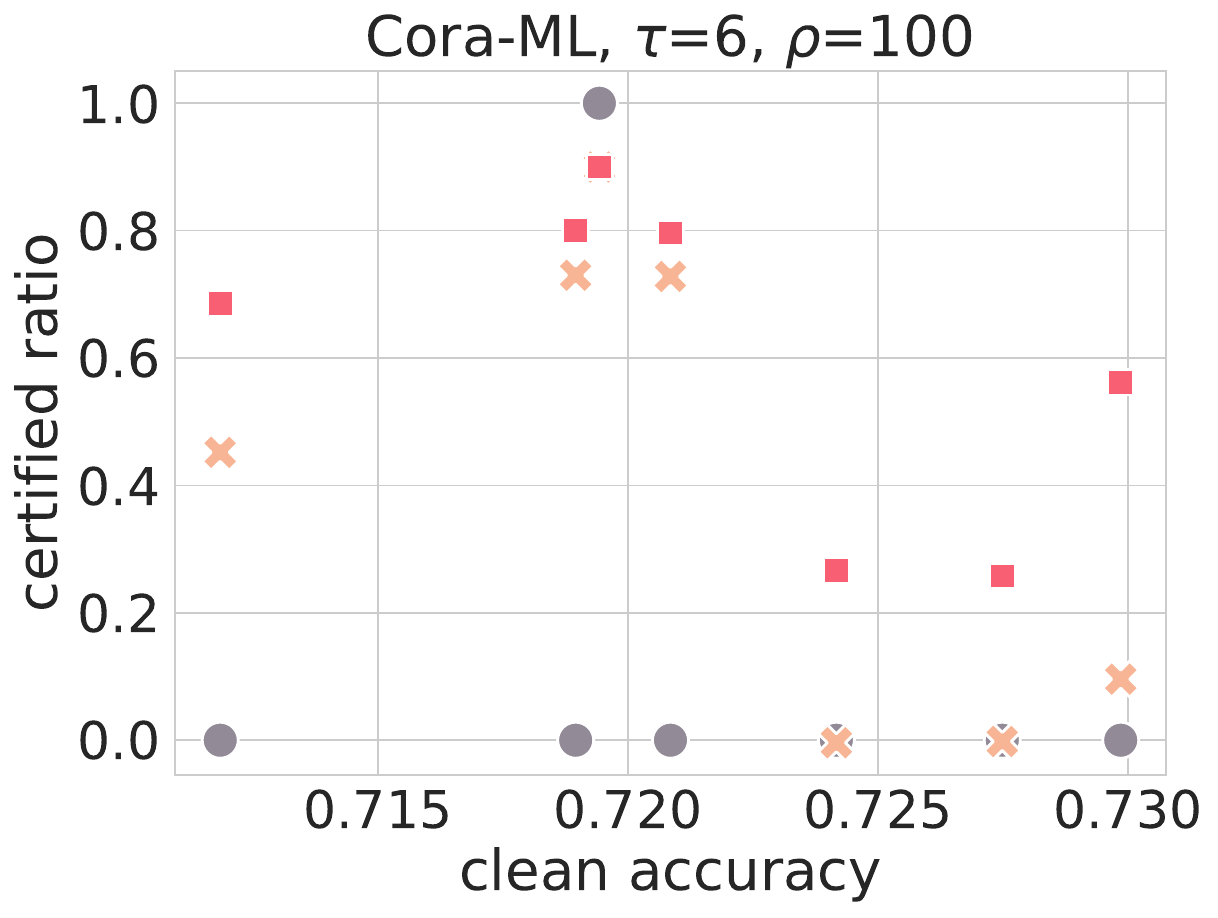}
    }
    \subfigure[]{
    \includegraphics[width=0.255\textwidth,height=2.8cm]{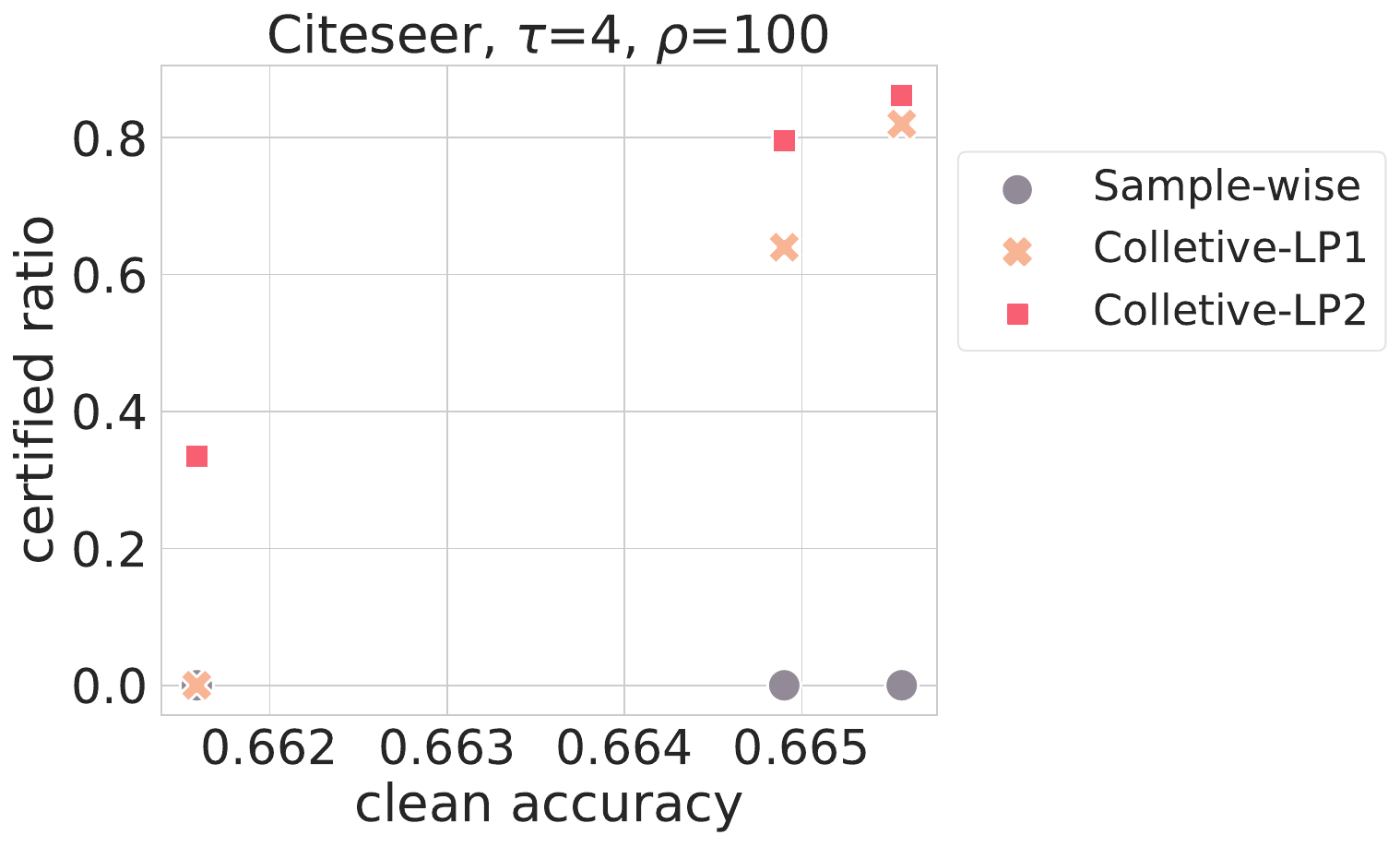}
    }
    \subfigure[]{
    \includegraphics[width=0.19\textwidth,height=2.8cm]{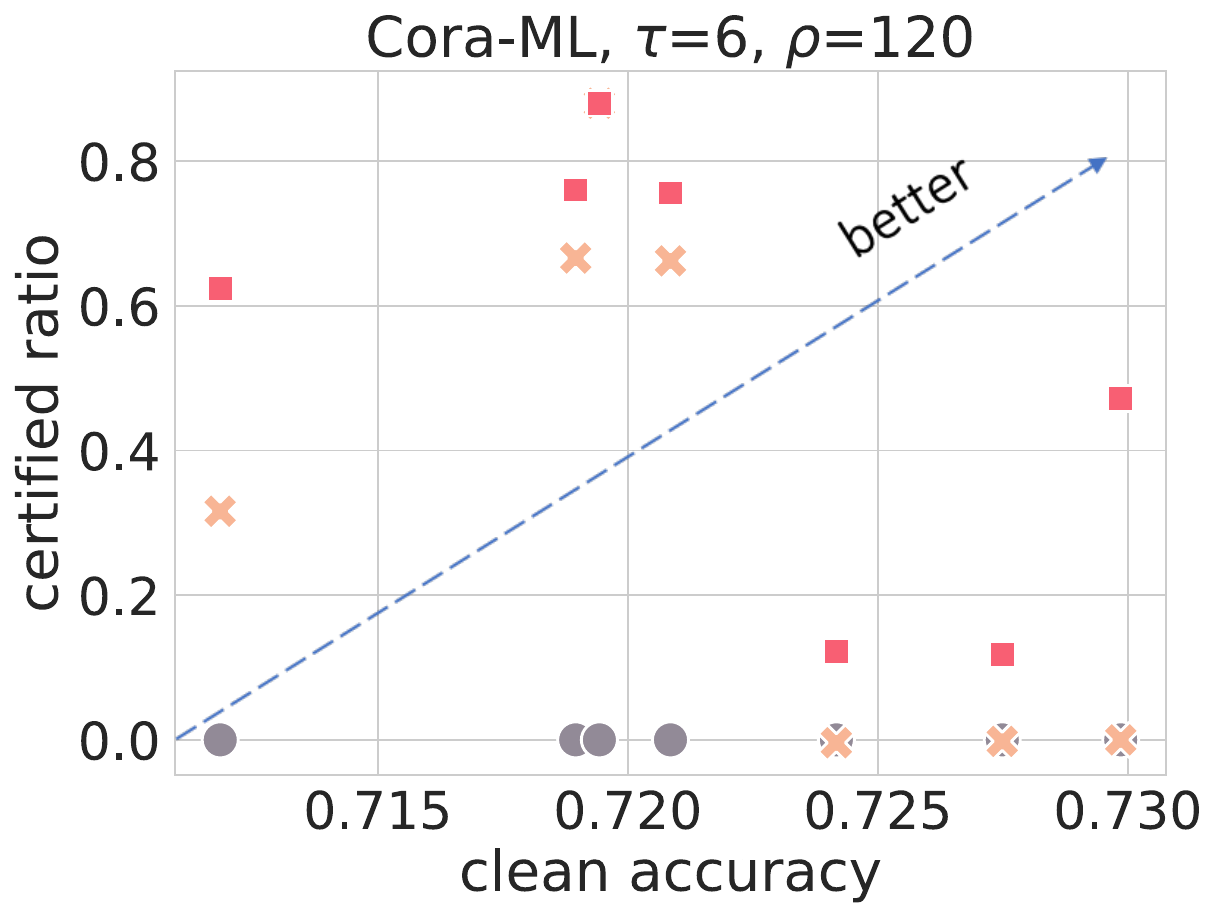}
    }
    \subfigure[]{
    \includegraphics[width=0.255\textwidth,height=2.8cm]{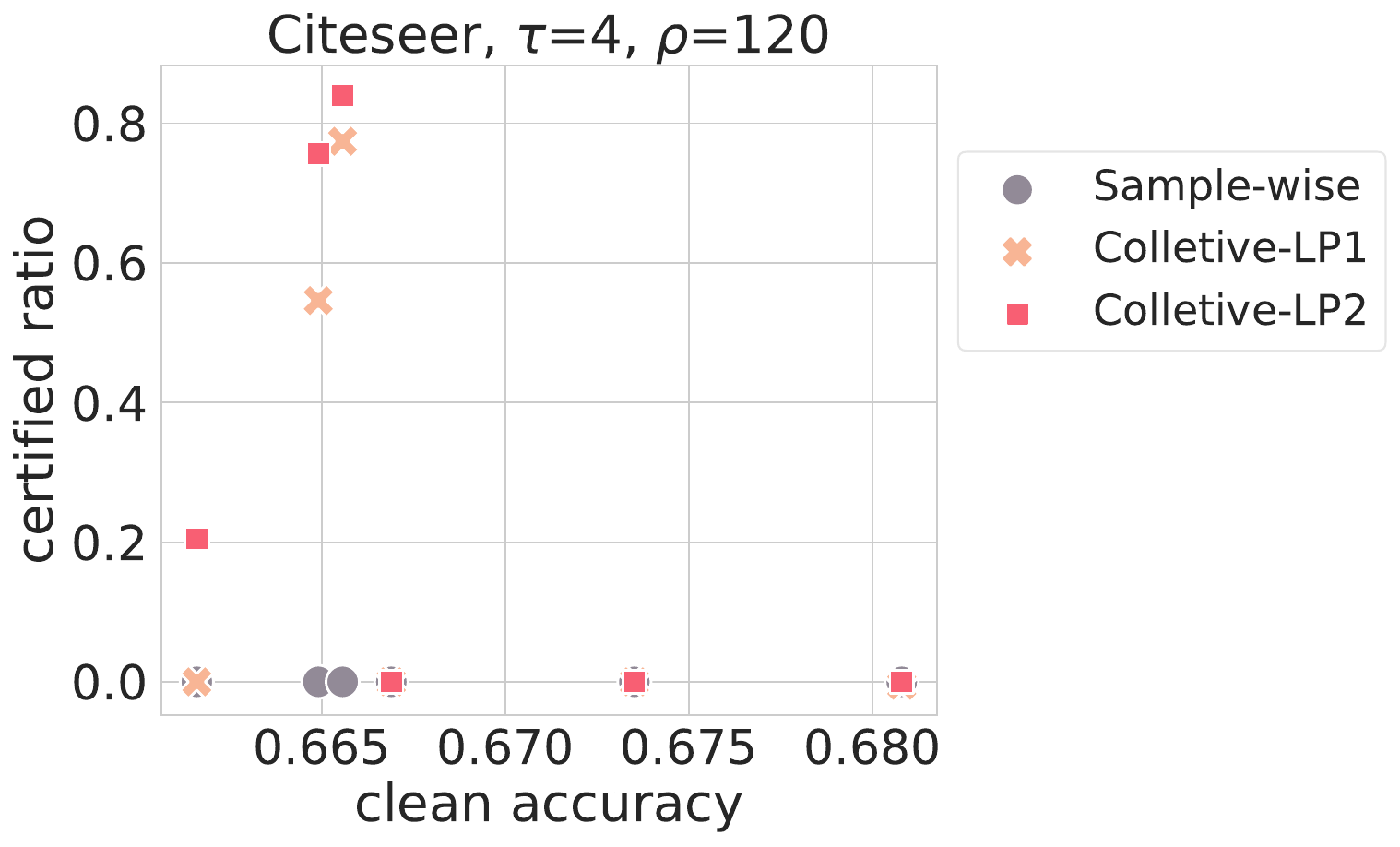}
    }
\caption{Clean accuracy and the certified ratio of our collective model under various smoothing parameters on GCN model.}
\label{fig:clean_certified_ratio}
\end{figure}

\begin{figure}[hbt!]
\centering
    \subfigure[]{
    \includegraphics[width=0.19\textwidth,height=2.8cm]{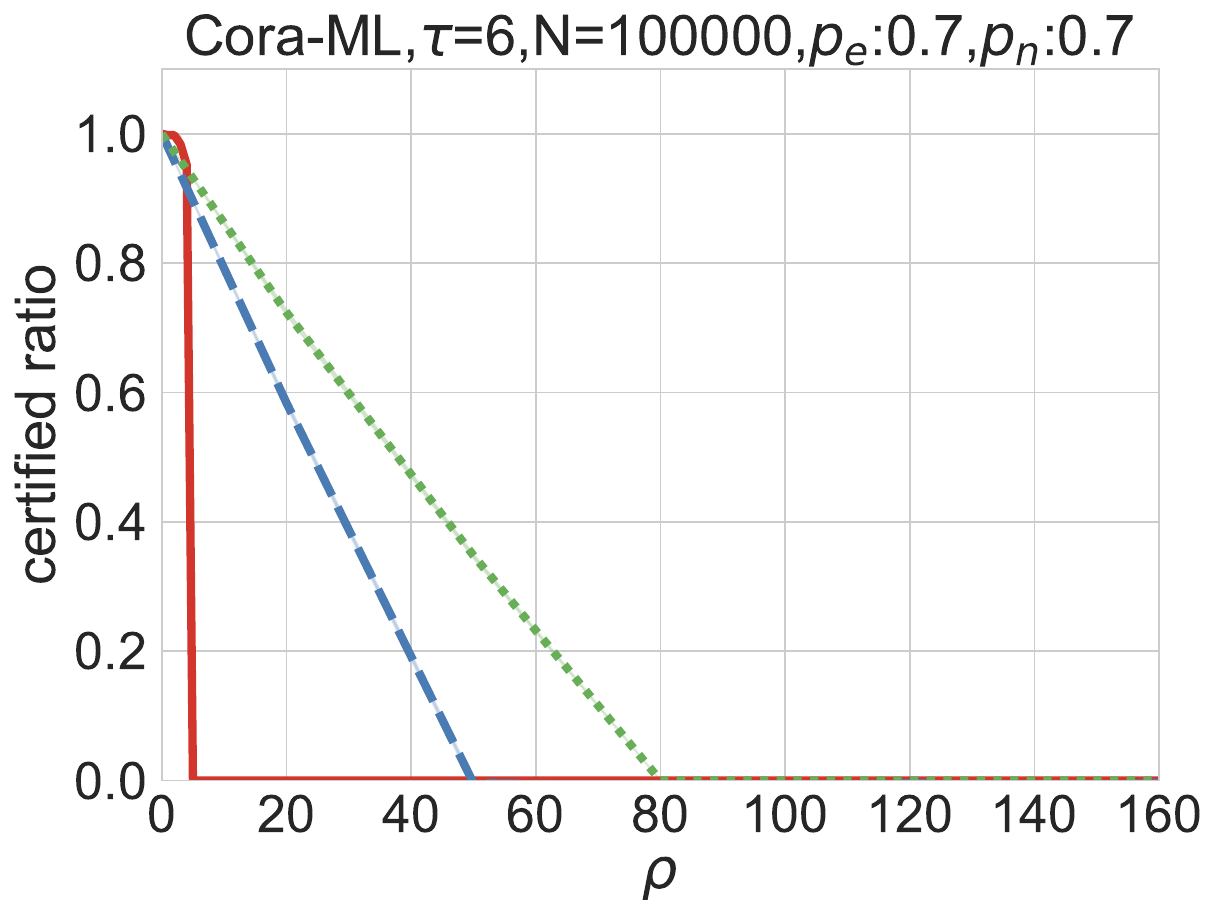}
    }
    \subfigure[]{
    \includegraphics[width=0.255\textwidth,height=2.8cm]{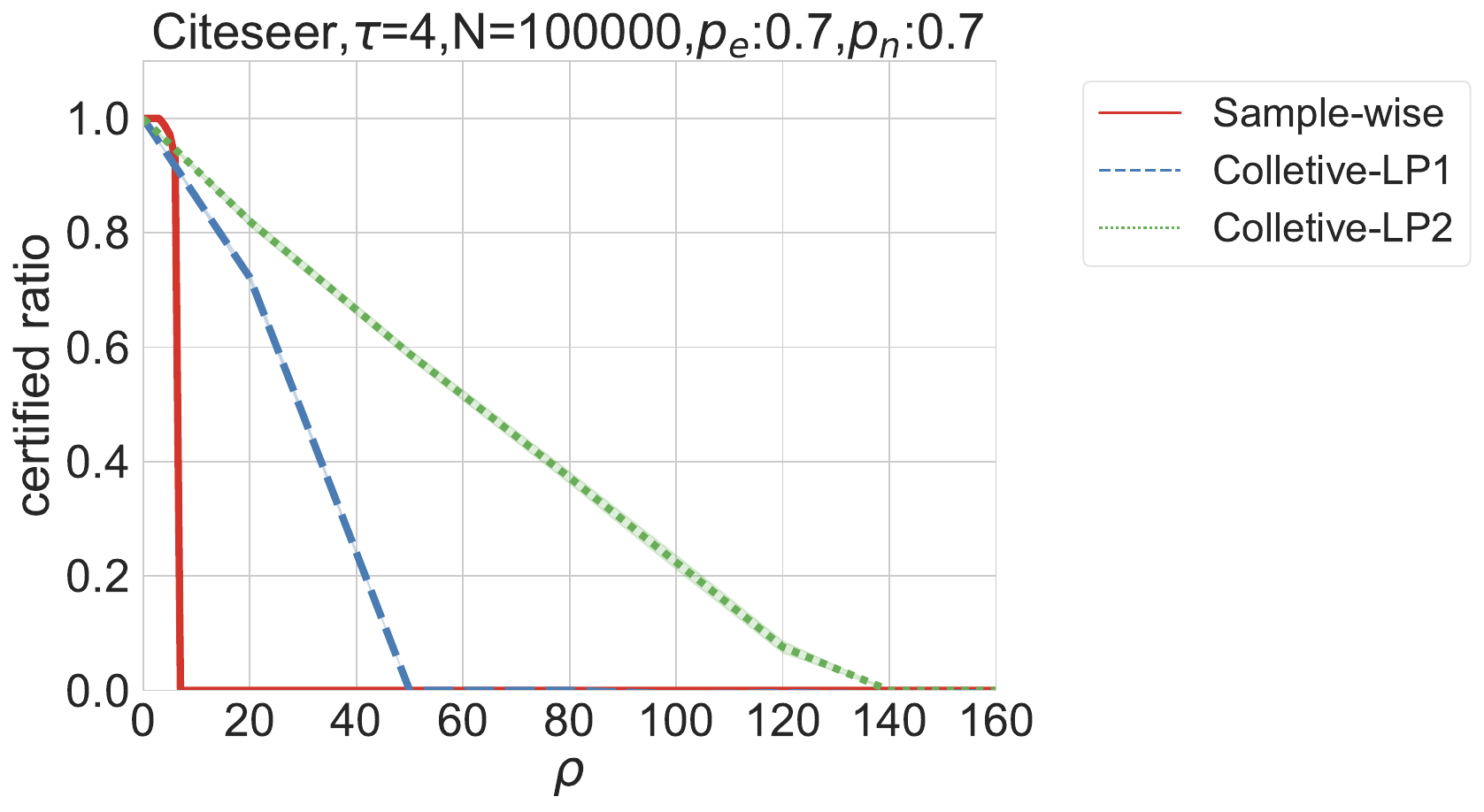}
    }
    \subfigure[]{
    \includegraphics[width=0.19\textwidth,height=2.8cm]{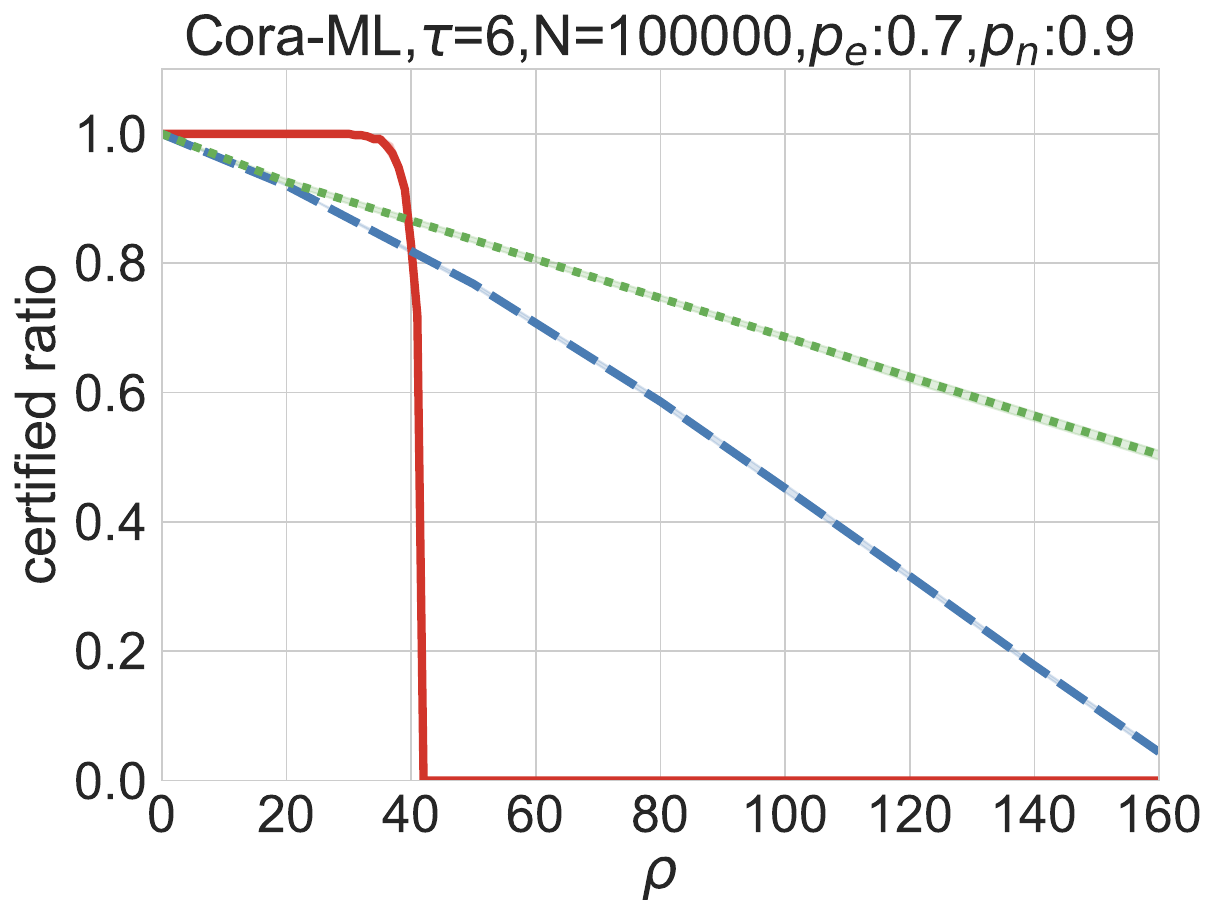}
    }
    \subfigure[]{
    \includegraphics[width=0.255\textwidth,height=2.8cm]{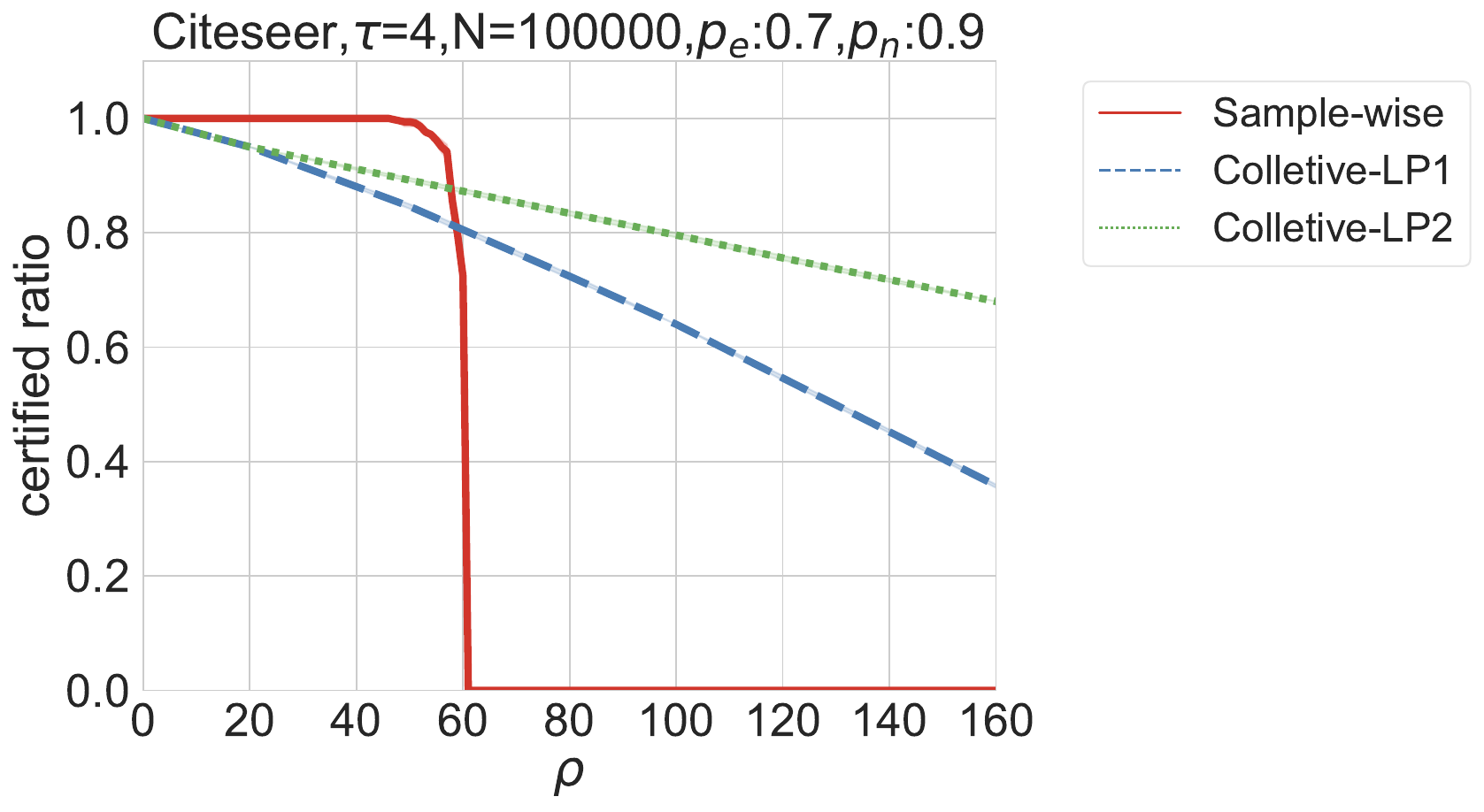}
    }
    \subfigure[]{
    \includegraphics[width=0.19\textwidth,height=2.8cm]{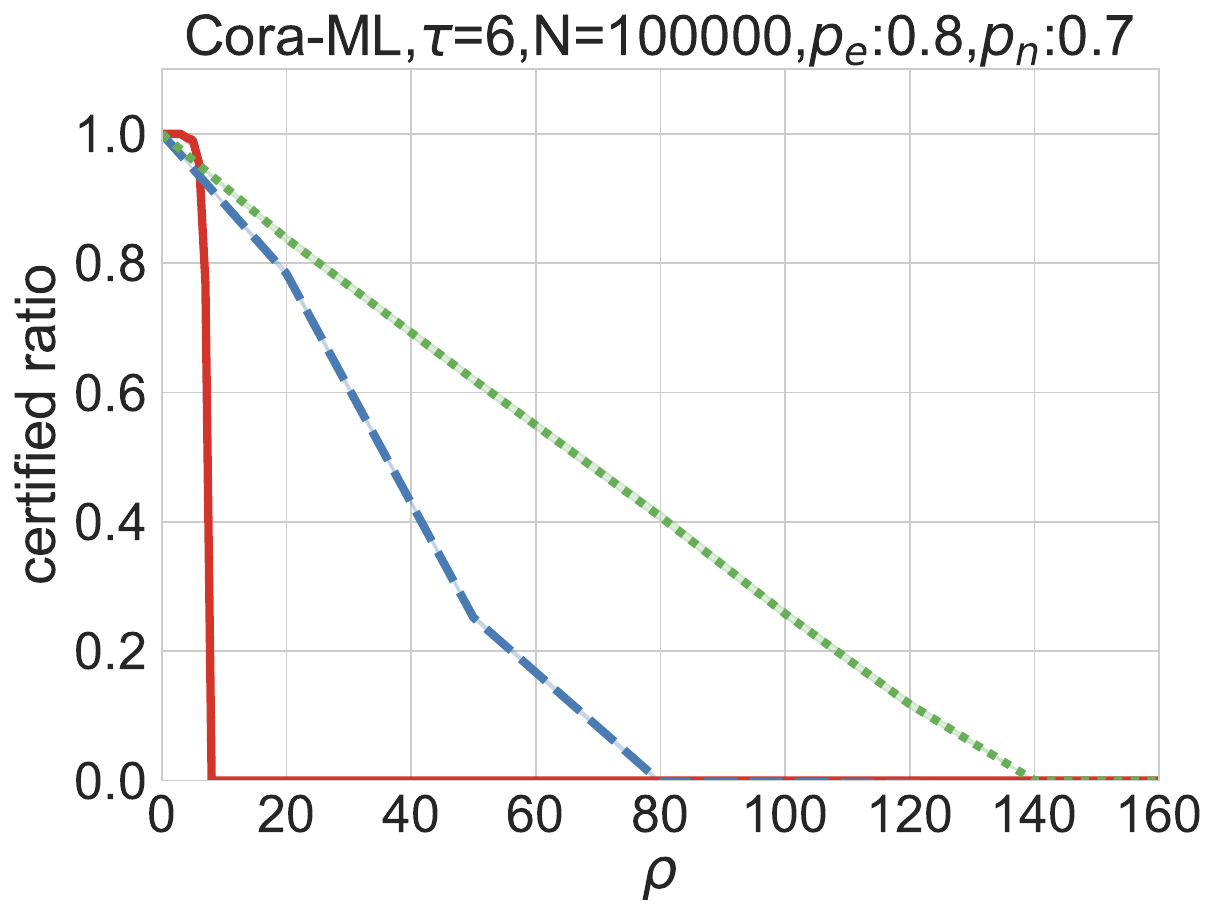}
    }
    \subfigure[]{
    \includegraphics[width=0.255\textwidth,height=2.8cm]{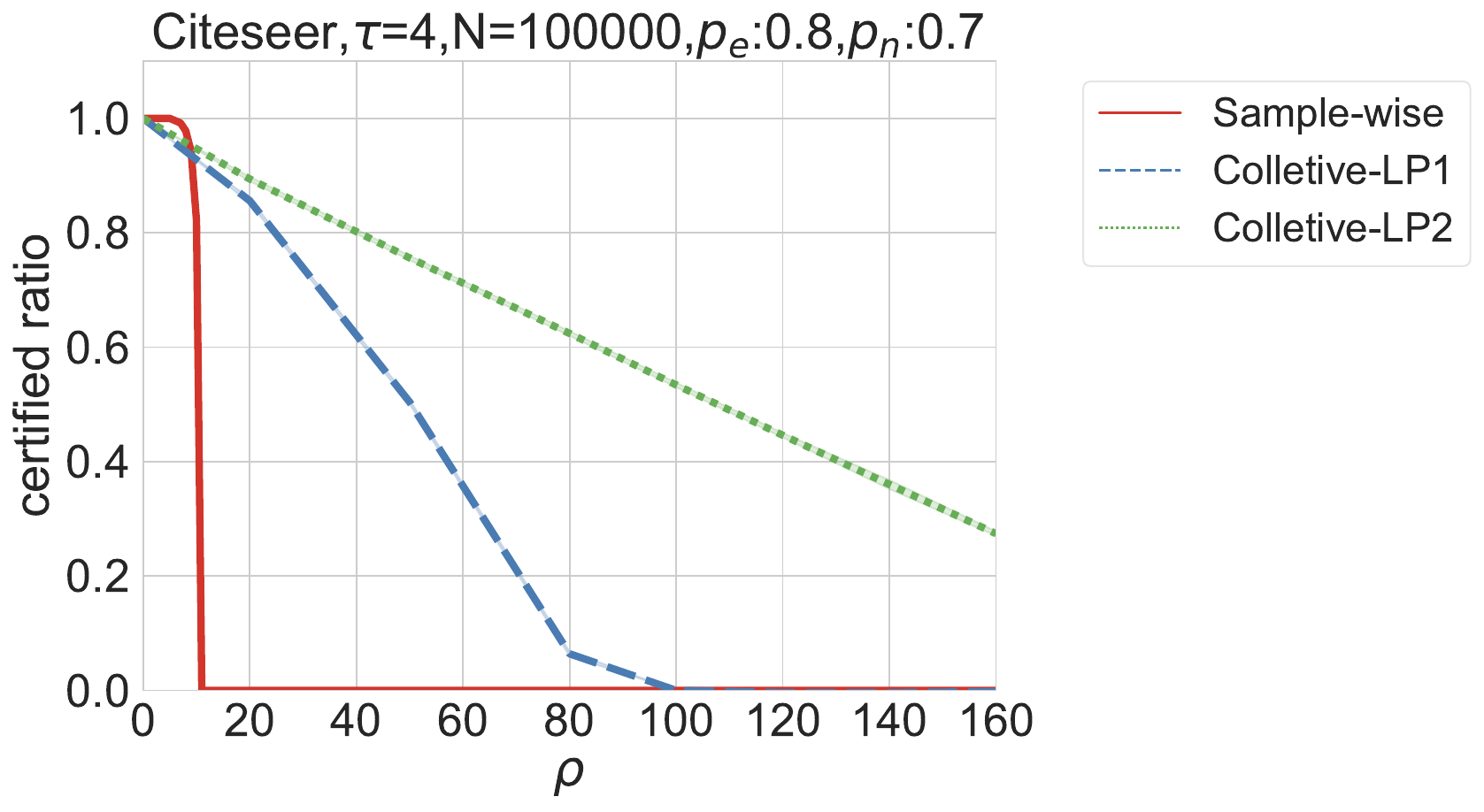}
    }
    \subfigure[]{
    \includegraphics[width=0.19\textwidth,height=2.8cm]{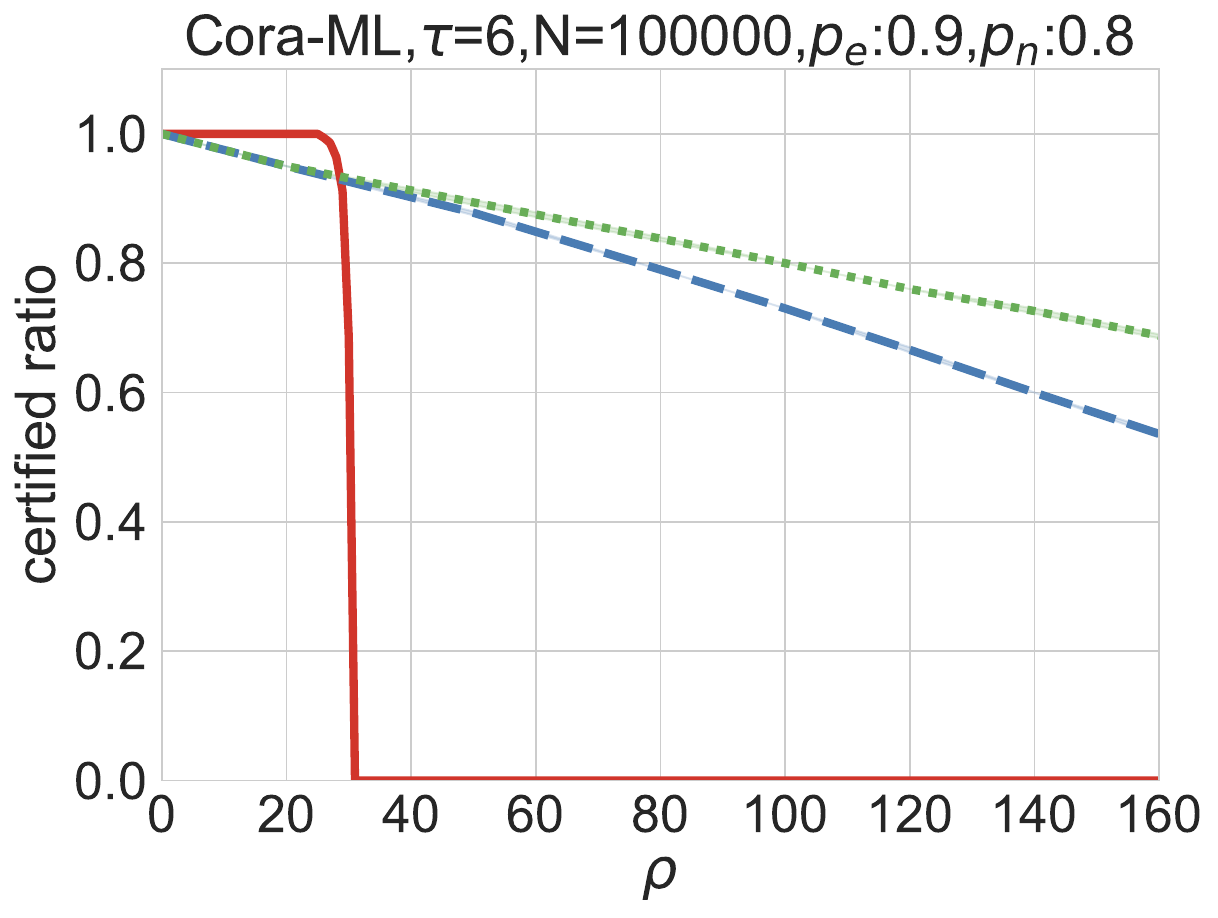}
    }
    \subfigure[]{
    \includegraphics[width=0.255\textwidth,height=2.8cm]{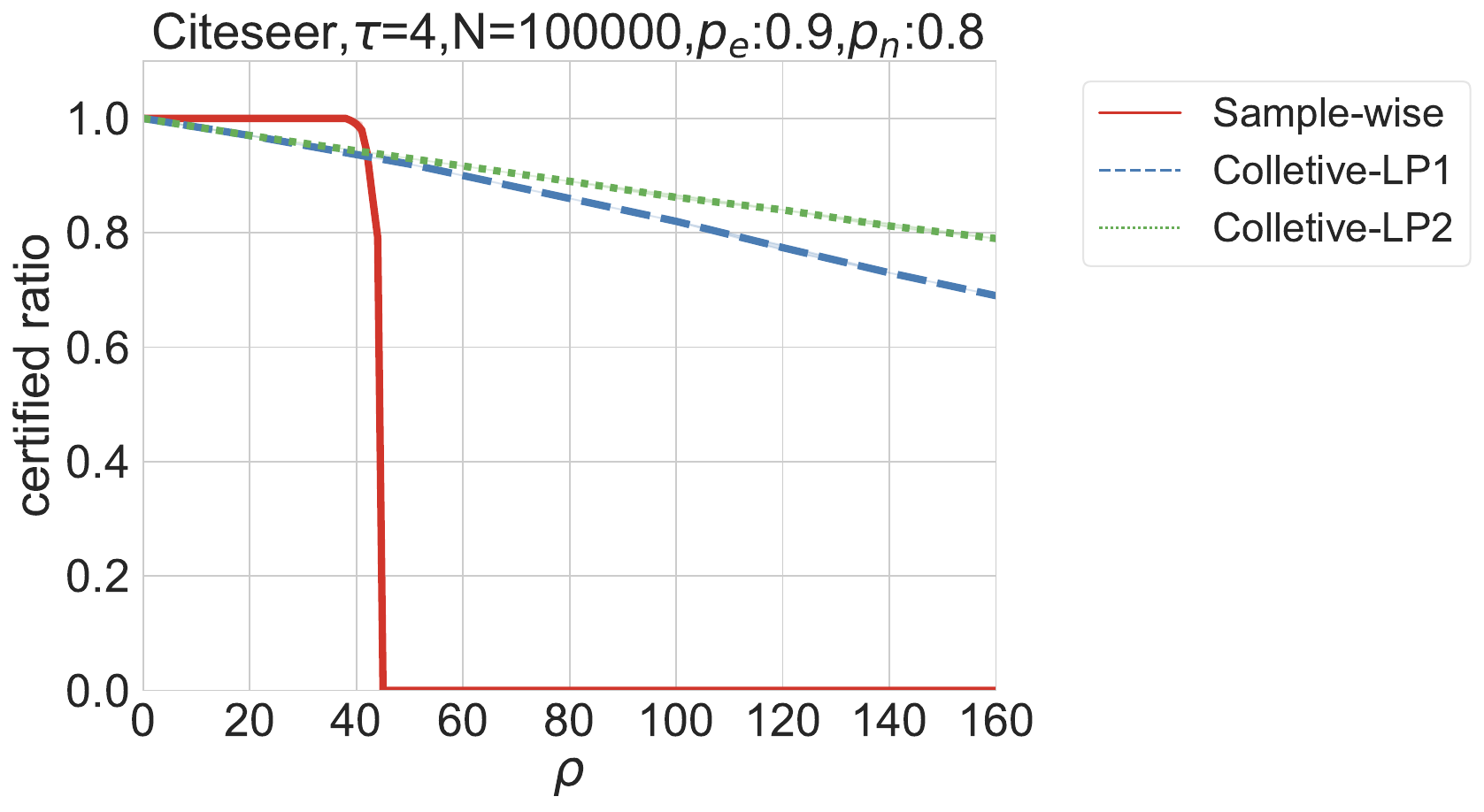}
    }
\caption{Certified ratio of our collective model under various smoothing parameters on GCN model.}
\label{fig:GCN_cer_curve_with_different_parameters}
\end{figure}

\subsection{GCN certified ratio of our methods under different smoothing parameters}

In addition, we conducted experiments to compare the performance of our methods with the sample-wise method under different combinations of parameters $p_e$ and $p_n$ on the Cora and Citeseer datasets. The results are shown in Figure. \ref{fig:GCN_cer_curve_with_different_parameters}.

From the figures, we can observe that our proposed methods always exhibit a larger certifiable radius. For example, when $\rho$ exceeds $60$, the sample-wise method fails to defend against any attacks, while our methods are still able to provide certifiable guarantees.


\subsection{Time complexity comparison of two relaxations}
Furthermore, we provide more detailed results on the runtime of the two proposed methods with different parameters in Figure. \ref{fig:cer_time_LP}. From the figures, we can observe that as the attack budget $\rho$ increases, the proposed Collective-LP2 method demonstrates superior efficiency compared to Collective-LP1 in both datasets. This efficiency advantage is particularly evident when $\rho$ exceeds $120$.
Notably, when $\rho=160$, the Collective-LP1 takes approximately $1,000$ seconds to complete the computation. On the other hand, the time consumption of Collective-LP2 remains consistently below $90$ seconds.

These results highlight the computational advantage of Collective-LP2 over Collective-LP1, especially for larger attack budgets. The reduced runtime of Collective-LP2 ensures the practicality and efficiency of our proposed method, making it suitable for real-world scenarios with larger attack budgets.

\begin{figure}[hbt!]
\centering
    \subfigure[Collective-LP1]{
    \includegraphics[width=0.19\textwidth,height=2.8cm]{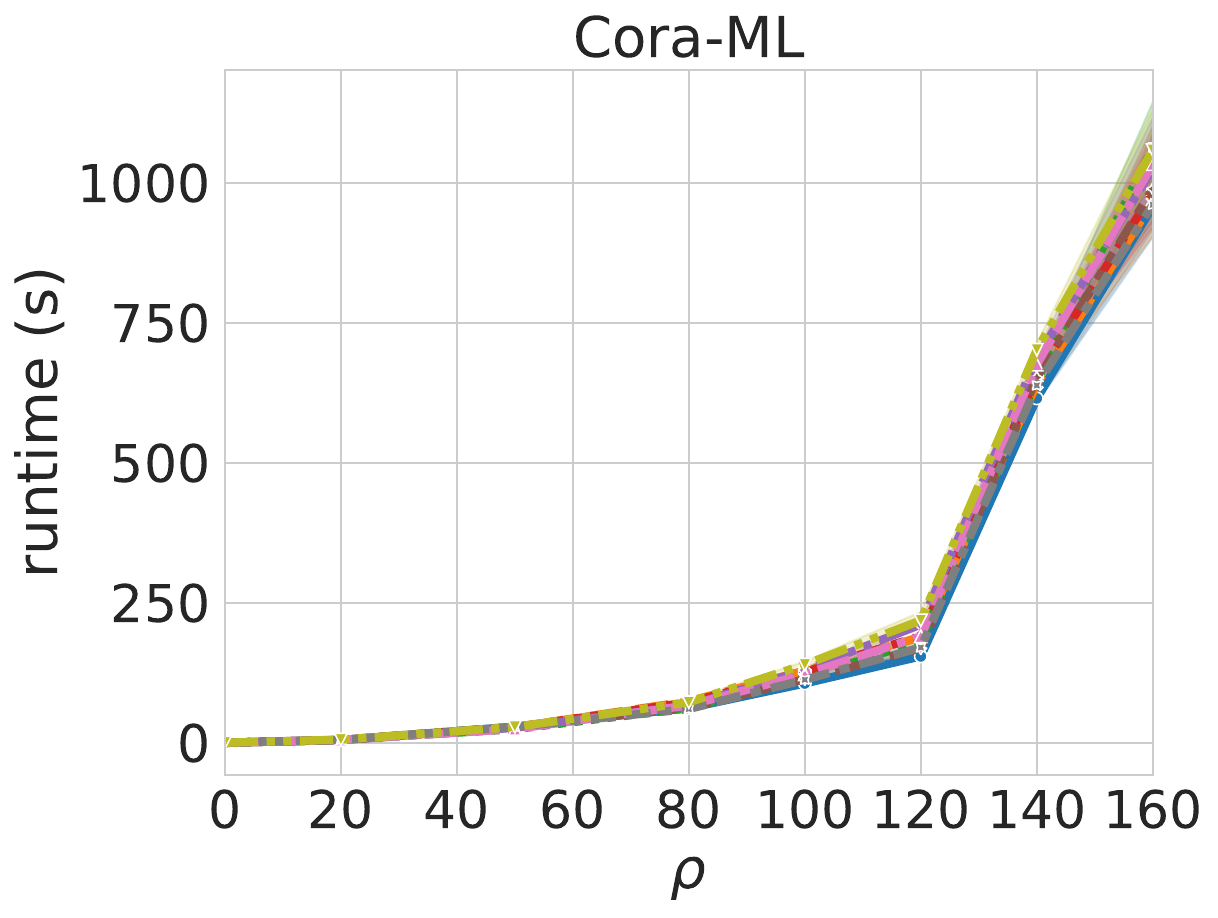}
    }
    \subfigure[Collective-LP1]{
    \includegraphics[width=0.255\textwidth,height=2.8cm]{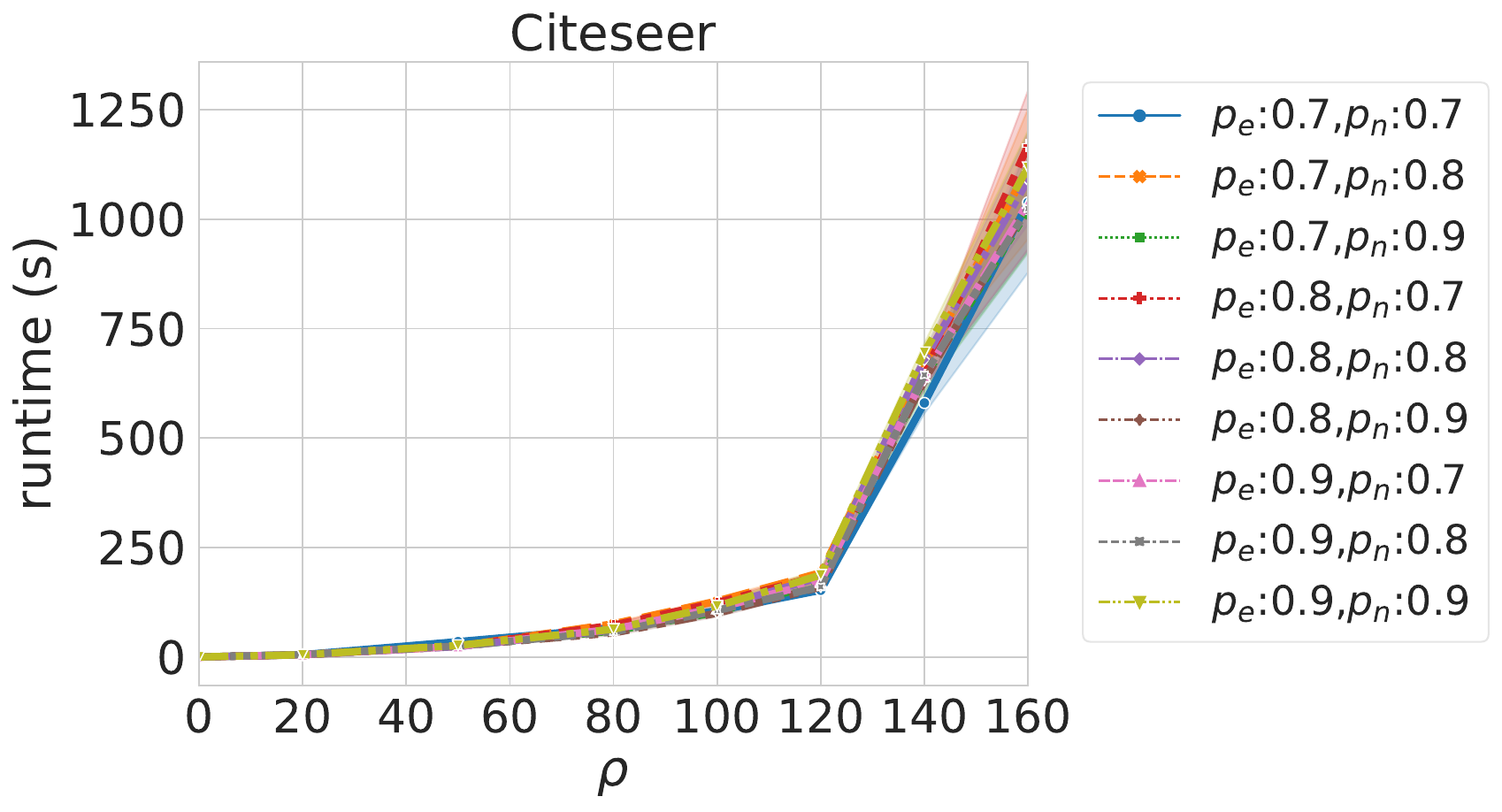}
    }
    \subfigure[Collective-LP2]{
    \includegraphics[width=0.19\textwidth,height=2.8cm]{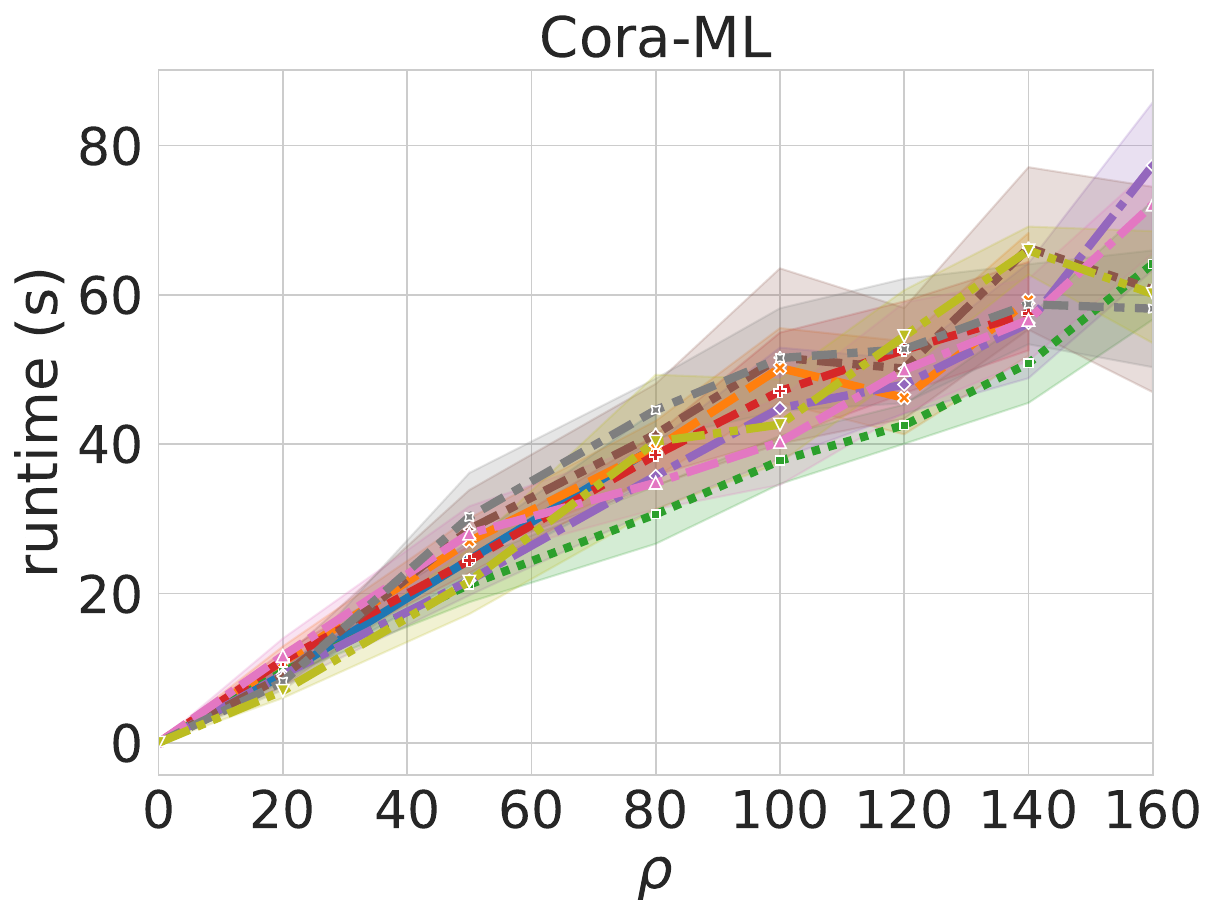}
    }
    \subfigure[Collective-LP2]{
    \includegraphics[width=0.255\textwidth,height=2.8cm]{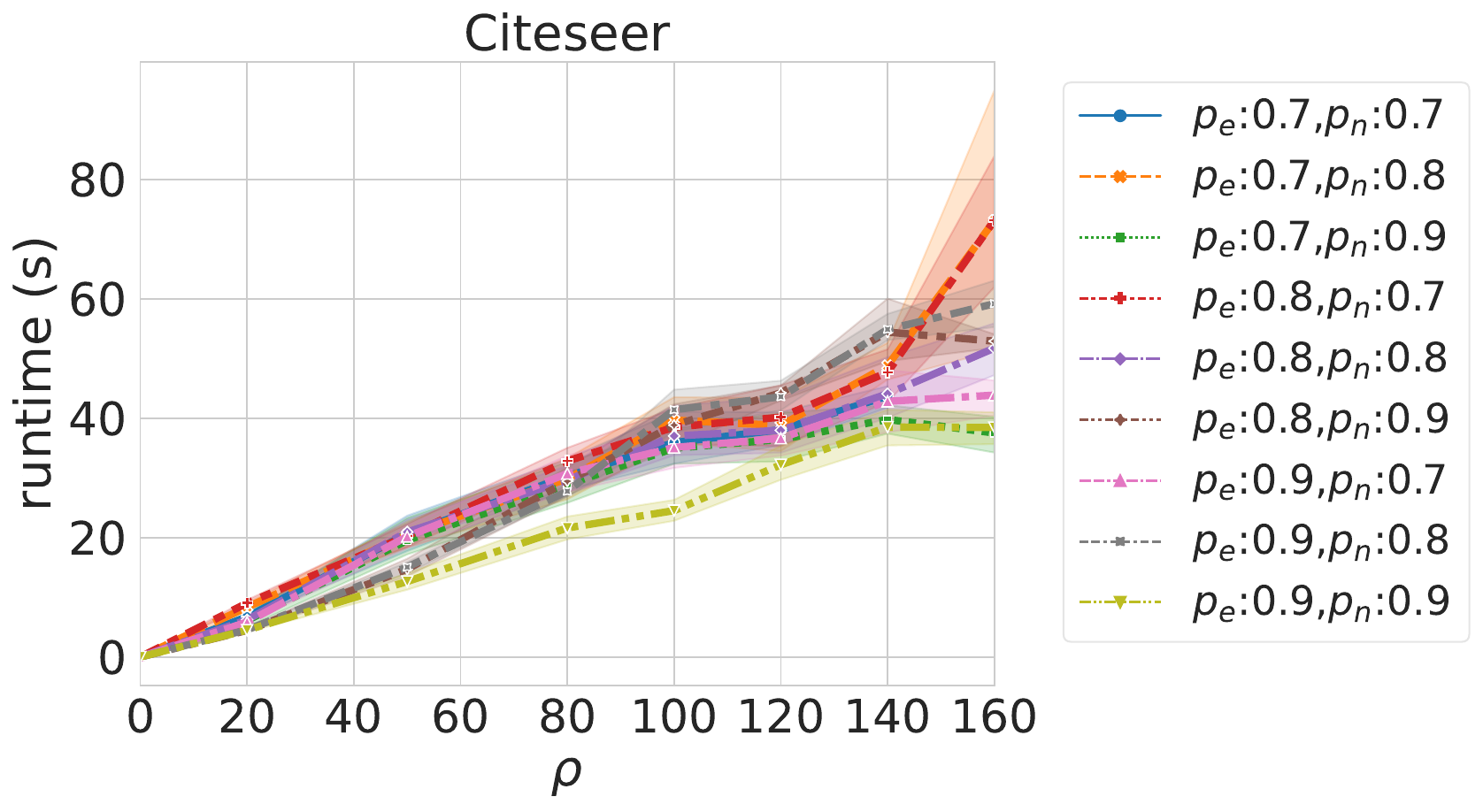}
    }
\caption{Runtime of our collective model under various smoothing parameters. }
\label{fig:cer_time_LP}
\end{figure}

\subsection{Against Global Attack: Verifying all testing nodes in a time}
\label{Sec:Appendix_global}

Alternatively, instead of verifying a subset of target nodes $\mathbb{T}$, we can extend our approach to verify all the testing nodes in the graph, as illustrated in Figure~\ref{fig:cer_all}. In this scenario, we measure the certified accuracy, which represents the ratio of nodes that are both correctly classified and certified to be consistent, as well as the runtime of our customized approach (Collective-LP2).

We have observed that the certified accuracy of our collective certificate only experiences a slight decrease as the attack budget increases, while the sample-wise approach can only certify the case of $\rho$ less than $50$. This indicates that our approach maintains a high level of certified robustness even when facing more severe adversarial attacks.

Furthermore, it is worth noting that our Collective-LP2 formulation exhibits excellent computational efficiency. Despite the presence of more than $1500$ testing nodes, the problem can be solved in less than $3$ minutes, even when the number of injected nodes $\rho$ is set to $140$ (approximately $5\%\times n$). This demonstrates the scalability and practicality of our customized relaxation approach (Collective-LP2) in real-world scenarios.

\begin{figure}[hbt!]
\centering
    \subfigure[Certified accuracy]{
    \includegraphics[width=0.19\textwidth,height=2.8cm]{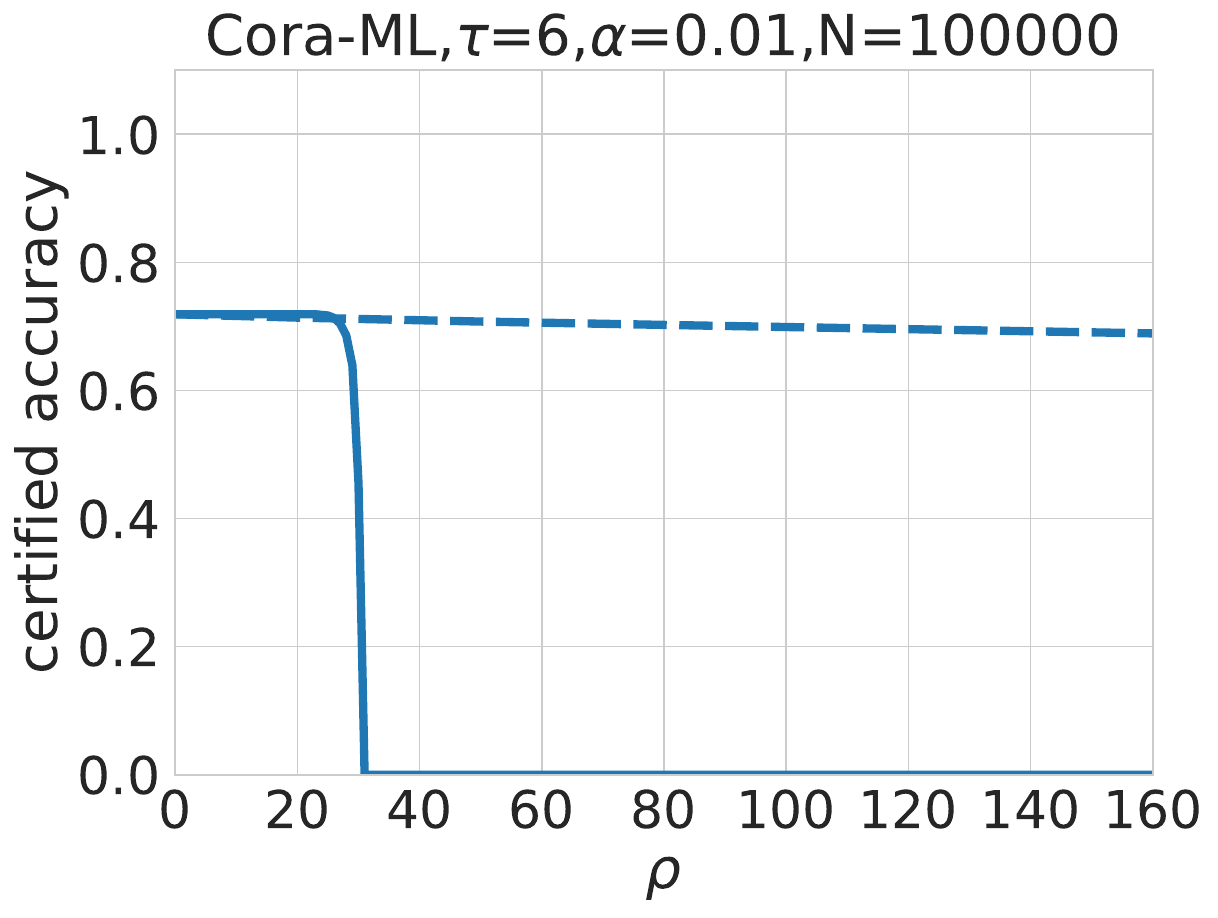}
    }
    \subfigure[Certified accuracy]{
    \includegraphics[width=0.255\textwidth,height=2.8cm]{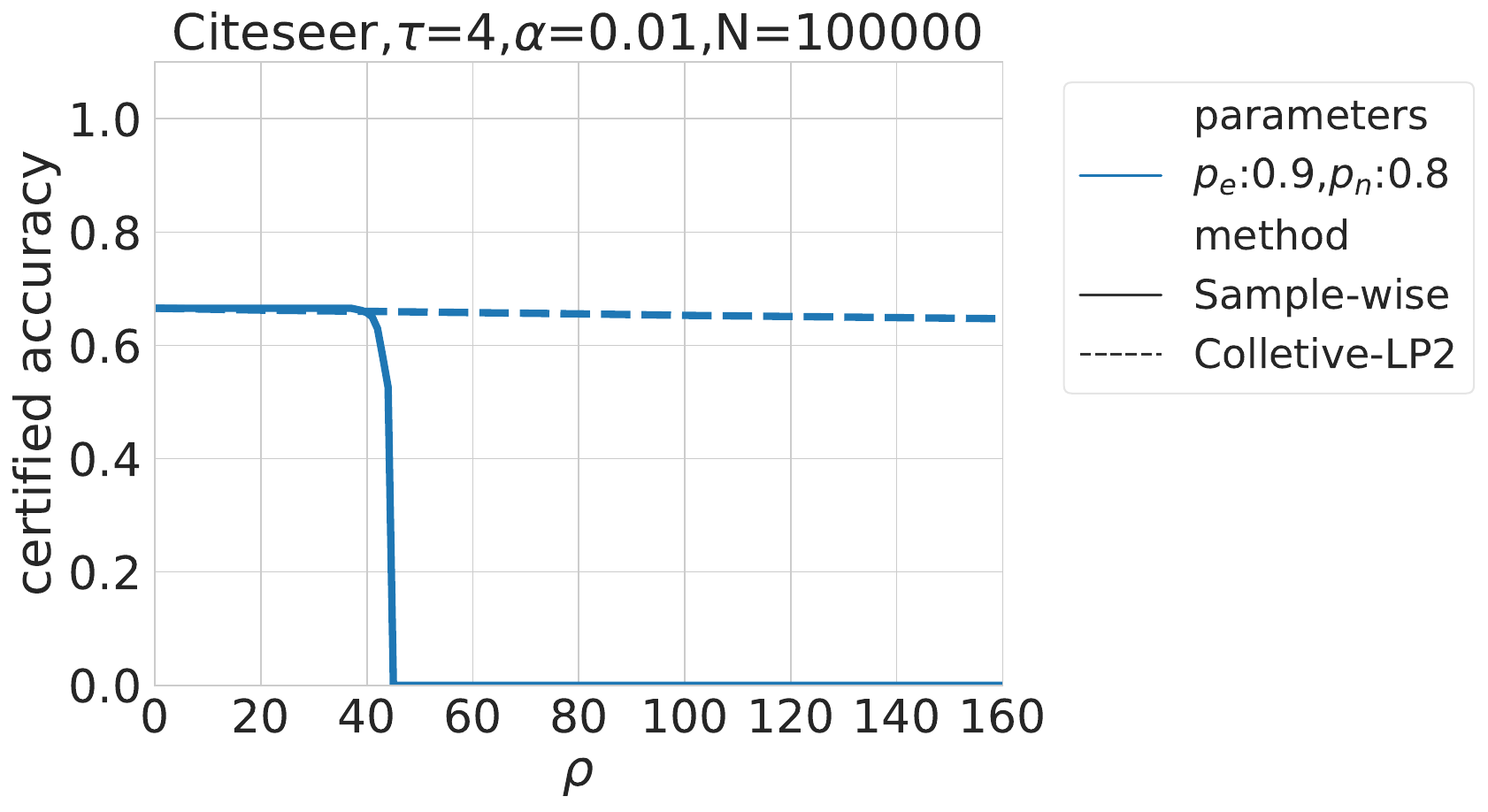}
    }
    \subfigure[Runtime]{
    \includegraphics[width=0.19\textwidth,height=2.8cm]{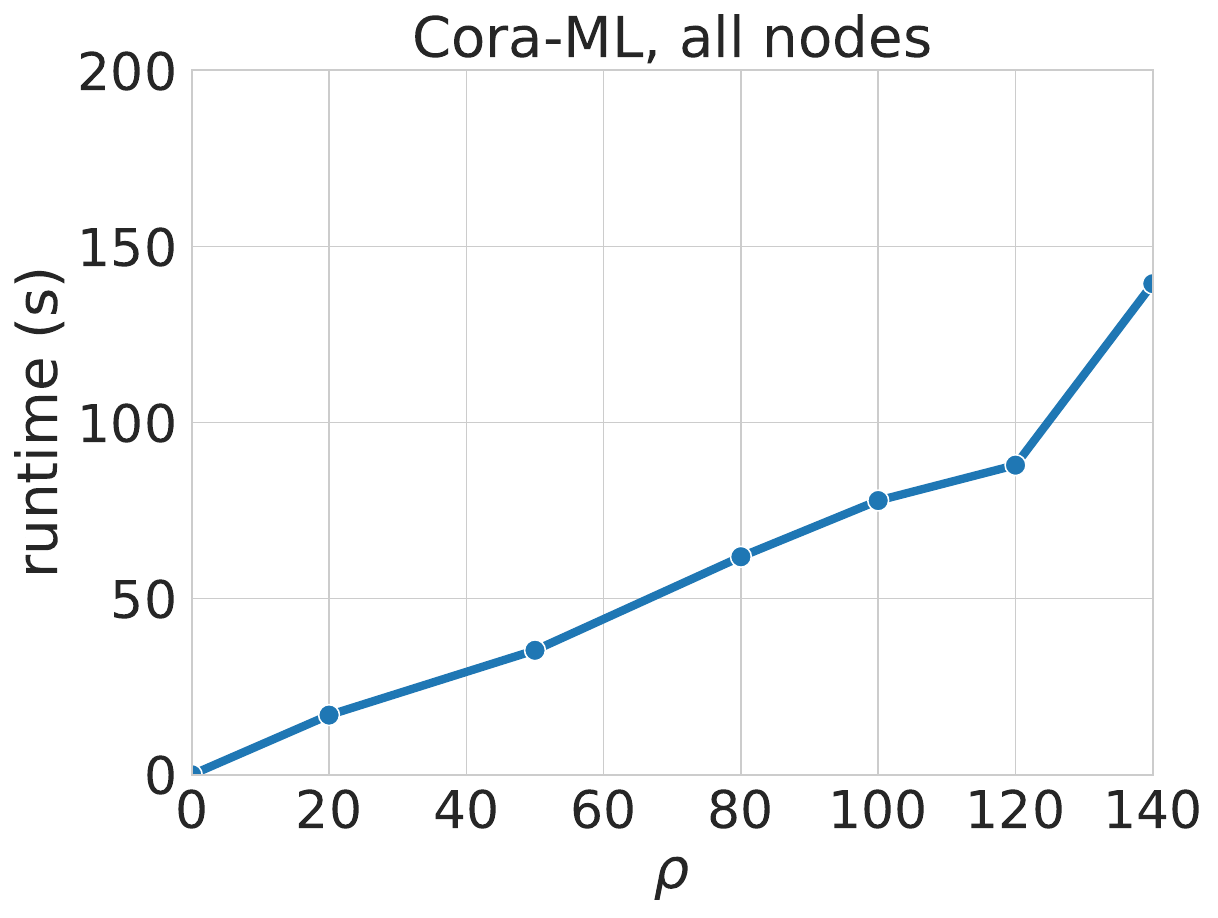}
    }
    \subfigure[Runtime (Collective-LP2)]{
    \includegraphics[width=0.255\textwidth,height=2.8cm]{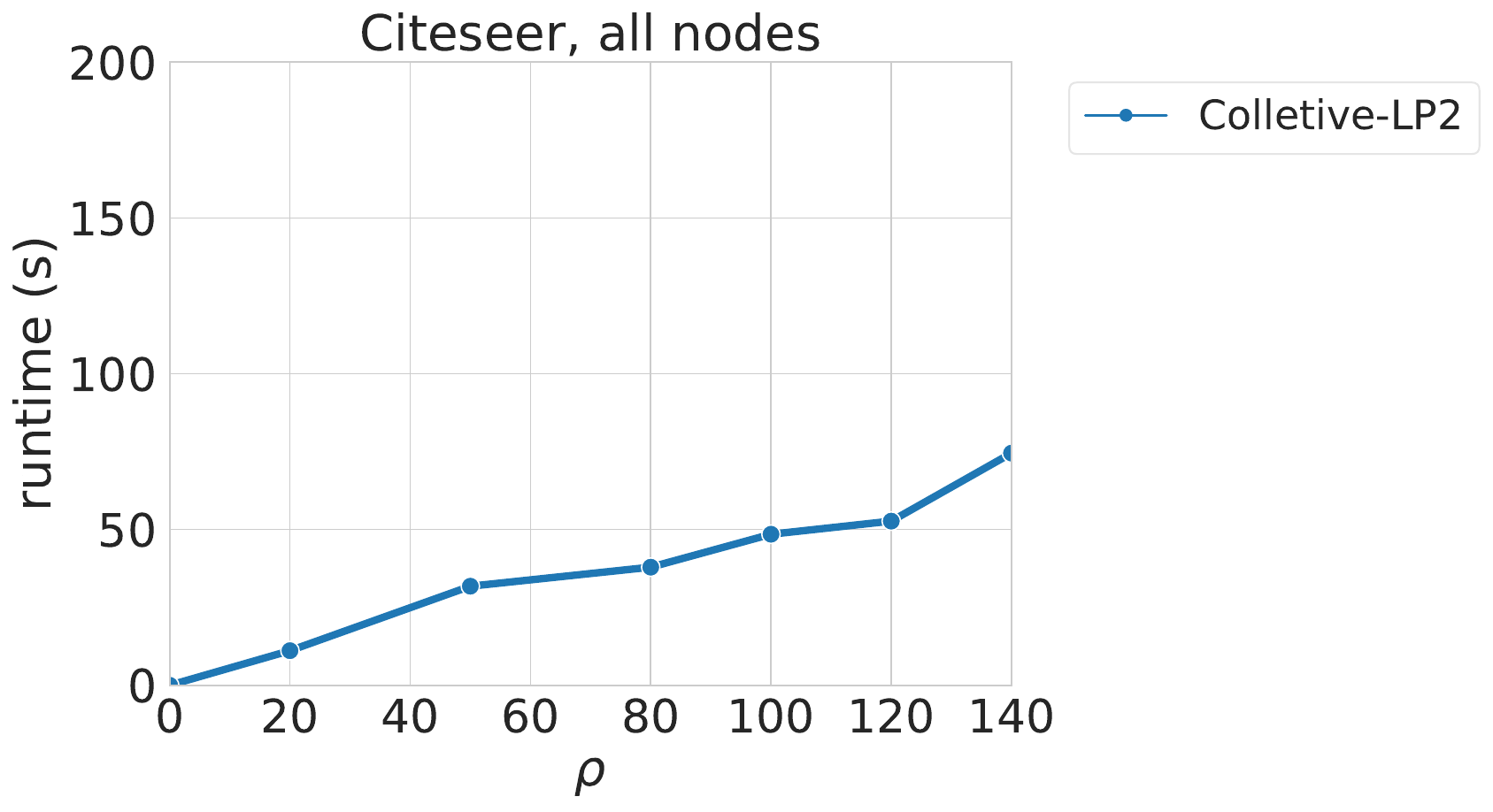}
    }
\caption{Certified accuracy and runtime in the case of setting all the testing nodes as $\mathbb{T}$. }
\label{fig:cer_all}
\end{figure}

\section{Limitations and Future Works}

Our collective certificate is obtained through the solution of a relaxed Linear Programming (LP) problem, which effectively reduces the computational complexity to linear time. However, this relaxation does come at a cost, as it introduces an integrity gap that compromises the certified performance. Consequently, in situations where the attack budget $\rho$ is small and the sample-wise certificate proves effective, the collective certificate may not yield superior results.

Nevertheless, in practical scenarios, we can easily combine the sample-wise and collective certificates with minimal effort to achieve stronger certified performance across a range of attack budgets, whether small or large. It is worth noting that since both the sample-wise and collective models share the same smoothed model, we only need to estimate the smoothing prediction once, avoiding computational overhead. By integrating both certificates, we can leverage their respective strengths and enhance the overall robustness of the system.

In future research, we plan to explore the development of tighter relaxations, such as semi-definite programming (SDP), to better handle the quadratic constraints. This could potentially yield improved certified performance and further enhance the robustness of our approach. Furthermore, we plan to extend the relaxation technique to accommodate polynomial constraints for deeper Graph Neural Networks (GNNs) where $k>2$. This extension will allow us to address more complex scenarios and further strengthen the applicability of our approach in real-world settings. 

    

\end{document}